%% file: JAS_Robust_NLR_MM_main_0.tex
\documentclass[10pt]{article}
\usepackage[utf8]{inputenc}
\usepackage{authblk}
\usepackage{amsfonts}
\usepackage{graphicx}
\usepackage{bm}
\usepackage{tikz}
\usepackage{enumitem}
\usepackage{multirow}
\usepackage{amsmath}
\usepackage{amsthm}
\usepackage{amssymb}
\usepackage{hyperref}
\usepackage{rotating}
\usepackage{float}
\usepackage{dirtytalk}
\usepackage{lscape}
\usepackage{appendix}
\newtheorem{theorem}{Theorem}[section]

\newtheorem{lemma}[theorem]{Lemma}
\newtheorem{prop}{\textbf{\textcolor{black}{Proposition}}}[section]
\theoremstyle{remark}
\newtheorem{remark}{\textbf{Remark}}[section]
\usepackage[a4paper, total={6in, 9in}]{geometry}
\usepackage{subcaption}
\newcommand{\btxt}[1]{\textcolor{black}{#1}}

\title{\textbf{Robust Inference for Non-Linear Regression Models with Applications in Enzyme Kinetics}}
\author[1]{\Large{Suryasis Jana}}
\author[1]{\Large{Abhik Ghosh\footnote{Corresponding author: abhik.ghosh@isical.ac.in}}}
\affil[1]{Indian Statistical Institute, Kolkata, India}

\date{}

\begin{document}
\maketitle

\begin{abstract}
Despite linear regression being the most popular statistical modelling technique, in real-life we often need to deal with situations where the true relationship between the response and the covariates is nonlinear in parameters. In such cases, one needs to adopt appropriate non-linear regression (NLR) analysis, having wider applications in biochemical and medical studies among many others. In this paper, we propose a new improved robust estimation and testing methodologies for general NLR models based on the minimum density power divergence approach and apply our proposal to analyze the widely popular Michaelis-Menten (MM) model in enzyme kinetics. We establish the asymptotic properties of our proposed estimator and tests, along with their theoretical robustness characteristics through influence function analysis. For the particular MM model, we have further empirically justified the robustness and the efficiency of our proposed estimator and the testing procedure through extensive simulation studies and several interesting real data examples of enzyme-catalyzed (biochemical) reactions.\\

\end{abstract} 

\maketitle

\section{Introduction}\label{intro}
Linear regression is arguably the most used statistical modeling technique to investigate relationships between variables. However, even for normally distributed continuous response variables, the linearity assumption appears to be restrictive in several real-life situations where the true relationship between the response and the covariates may be best modeled by a suitable non-linear function in parameters. Such scenarios frequently arise in biomedical and biochemical researches, besides statistical signal analyses and many other domains. For example, a popular such non-linear relationship in enzyme kinetics is known as the Michaelis-Menten (MM) equation. This MM equation, originally developed by Leonor Michaelis and Maud Menten in 1913 \cite{michaelis1913kinetik}, has now become the foundational basis of enzymology for studying the rates of enzyme-catalyzed chemical reactions; see, e.g., \cite{ataide2009optimal, baronas2013optimization, khanin2007statistical, oliveira2008chemometric, srinivasan2022guide}. For such a reaction, if the reaction velocity is $v$ at a substrate concentration $s$, then the MM equation says
\begin{equation}\label{mmeqn}
    v = \frac{V_{max}s}{K_m + s},
\end{equation}
where $K_m$ is the pseudo-equilibrium constant of the reaction, called the Michaelis constant, and $V_{max}$ is the maximum velocity attainable (theoretically) at an infinite substrate concentration saturated with the enzyme. \btxt{Also, from Eq.~\eqref{mmeqn}, we observe that $v=V_{max}/2$ if and only if $K_m = s$, and hence the Michaelis constant $K_m$ can further be interpreted as the substrate concentration needed to attain half of the maximum attainable velocity $V_{max}$.} Here $K_m$ and $V_{max}$ are the two reaction-specific unknown (model) parameters that are estimated from experimental data, consisting of the observed reaction velocities ($v_i$) at different (often pre-fixed) substrate concentrations ($s_i$), for a specific chemical reaction, assuming additive (and independent) Gaussian error in each observation $v_i$ (for $i=1,\hdots,n$, say). This estimation setup can be expressed as a non-linear regression (NLR) model given by 
\begin{equation}\label{mmeqn2}
    v_i = \frac{\beta_1 s_i}{\beta_2 + s_i} + \epsilon_i = \mu(s_i, \bm{\beta}) + \epsilon_i,~~~~~~~i=1, \ldots, n,
\end{equation}
where $\bm{\beta}=(\beta_1,\beta_2)^T=(V_{max}, K_m)^T$. The estimation of the model parameter $\bm{\beta}$ and the subsequent inference is commonly performed using the non-linear least-squares or the maximum likelihood approach; see e.g., \cite{goudar1999parameter, toulias2016fitting, jukic2007total}. These inferences are also extended to variants of the MM model for specific applications; see, e.g., \cite{walters2023modified, douglas2024hetmm}. The MM model \eqref{mmeqn} can also be analyzed using a linearizing transformation, called the Woolf transformation, used by Currie \cite{currie1982estimating}. However, in that case, it is not possible to take additive errors, as in model \eqref{mmeqn2}.

More generally, given $n$ observations $y_i$ and $\bm{x}_i$ on a (real-valued) response variable $Y$ and associated $k$-dimensional covariates $\bm{X}$, respectively, for $i=1,\ldots, n$, a typical NLR model can be expressed as \cite{batesnonlinear, seber2005nonlinear} 
\begin{equation} \label{nlr}
    y_i = \mu(\bm{x}_i, \bm{\beta}) + \epsilon_i,\ i = 1,2,\hdots,n,
\end{equation}
where $\bm{\beta}$ is a $p$-dimensional regression parameter vector and $\epsilon_i$'s are independently and identically distributed (IID) random errors with mean $0$ and variance $\sigma^2$. Here the form of the mean (or regression) function $\mu(\cdot,\cdot)$ is assumed to be known (pre-specified), continuous, and non-linear in the parameter $\bm{\beta}$. The full parameter vector in this NLR model is then given by $\bm{\theta} = (\bm{\beta}^T, \sigma^2)^T \in \bm{\Theta} \subseteq \mathbb{R}^p \times \mathbb{R}^+$. It is to be noted that, since $\mu(\bm{x}_i, \bm{\beta})$ is non-linear in $\bm{\beta}$, dimensions of $\bm{x}_i$ and $\bm{\beta}$ are not necessarily the same (i.e., $p$ may or may not be equal to $k$). There is a vast amount of literature on the classical ordinary least-square (OLS) and the maximum likelihood (ML) methods under NLR models. Asymptotic properties of the least-square estimator under the NLR model were initially studied by Jennrich \cite{jennrich1969asymptotic} and Wu \cite{wu1981asymptotic}. It can also be easily observed that under the normality assumption on the IID errors $\epsilon_i$'s in \eqref{nlr}, both the OLS estimator and the ML estimator (MLE) of the parameter $\bm{\beta}$ coincide. Discussions on the MLE under the NLR model \eqref{nlr} and its asymptotic properties are available in, e.g., Seber and Wild \cite{seber2005nonlinear}. Further refinements and advancements in both theory and computation of these estimators under general NLR models and several of its important special cases are available in, e.g., \cite{puxty2006tutorial, kundu1991asymptotic, kundu1993asymptotic, kundu1996asymptotic}, among many others.

However, the main drawback of the least-squares or the maximum likelihood-based inference is their high sensitivity to the presence of outliers and other contaminations in the observed data, even in a small proportion. Since we frequently encounter outliers and influential observations in real-life datasets, there is a clear need for suitable robust estimation and inference procedures yielding stable and accurate results even under data contamination. But, although there exists a vast literature on robust statistical procedures for linear and generalized linear regression models, there have been comparatively fewer works developing robust procedures for NLR models. In particular, Liu et al.~\cite{Rliu2005robust} and Marasovic et al.~\cite{marasovic2017robust} studied robust M-estimators under the NLR models with Huber's and Tukey's loss/weight functions, respectively. Tabatabai et al.~\cite{tabatabai2014new} proposed a modified M-estimation procedure (to be referred to as the KPS estimator), whereas Liu et al.~\cite{Pliu2021robust} studied the median of means (MOM) estimator under the NLR model. \btxt{Girardi et al.}~\cite{girardi2020robust} used the Tsallis score for robust estimation under a particular example of 3-parameter log-logistic NLR model, for analyzing COVID-19 contagion data of Italy, but they avoided many mathematical details and justifications of the estimation procedure. Further, all these existing works only dealt with the estimation problem, and there is no concrete work on robust hypotheses testing available in the literature of NLR models. In the context of MM models in enzyme kinetics, only the KPS and the M-estimators are explored to generate robust solutions under possible data contaminations; see, e.g., \cite{tabatabai2014new, marasovic2017robust, Rliu2005robust}, among others.

In this paper, our aim is to develop an improved robust estimator and associated robust testing procedure for the general class of NLR models based on the popular minimum density power divergence (DPD) approach of Basu et al.~\cite{basu1998robust}, and apply them for robust analysis of MM models with greater accuracy compared to the existing alternatives. Due to various nice properties, the DPD-based approach has been successfully implemented in several statistical problems, including regressions; see, e.g., \cite{ghosh2013robust, ghosh2015robust, ghosh2016robust, basu2016generalized, basu2018robust, ghosh2019robust, castilla2018new, basak2021optimal}. For the sake of completeness, a brief background about the DPD and the minimum DPD estimator (MDPDE) is provided in the online supplementary material. Here, we extend the definition of the MDPDE for the general NLR models and discuss its asymptotic properties (consistency and asymptotic normality) under simple verifiable conditions. We justify the robustness property of our proposed estimator by examining the boundedness of its influence function. We additionally develop a robust hypothesis testing procedure for general NLR models based on our proposed estimator (i.e., the MDPDE) and examine the influence functions of our test statistics to justify their claimed robustness. \btxt{We then apply our proposed methodologies for robust inference under the MM model \eqref{mmeqn2}, examine their theoretical properties by further simplifying the formulae for asymptotic variances and influence functions, and illustrate their finite-sample performances through extensive simulation studies.} We also develop and study a robust test procedure for testing the validity of the MM model for a given (possibly contaminated) experimental dataset. Further, we discuss some data-driven procedures for finding an optimal choice of the MDPDE tuning parameter in the online supplementary material. Finally, we present some interesting real data examples of different enzyme kinetic reactions.

The rest of the paper is organized as follows: In Section \ref{RobustNLR}, we present our proposed MDPDE along with the derivation of its asymptotic and robustness properties. In Section \ref{hyptesting}, we discuss the proposed robust hypothesis testing procedures utilizing the MDPDEs and study their asymptotic properties and their influence functions. Section \ref{Application} is devoted to the application of our proposed methods to the Michaelis-Menten model \eqref{mmeqn2} along with extensive simulation studies illustrating their improved finite-sample performances in comparison to the existing robust approaches. Section \ref{Optalpha} consists of a discussion on some data-driven procedures for finding an optimal choice of the MDPDE tuning parameter. Section \ref{realdata} presents analyses of some real data examples, and final concluding remarks are given in Section \ref{Conclution}. For brevity in the presentation, some background details, proofs of the results, and some numerical illustrations are provided in the web appendices of the online supplementary material.

\section{Robust parameter estimation under general NLR models}\label{RobustNLR}
\subsection{The MDPDE}\label{MDPDNLRM}
Let us consider the NLR model given in \eqref{nlr} with the error assumption $\epsilon_i \sim \mathcal{N}(0,\sigma^2)$, independently, for all $i=1,2,\hdots,n$. So, we have $n$ independent and non-homogeneous (INH) observations $y_1,y_2,\hdots,y_n$ with 
$y_i \sim \mathcal{N}\left(\mu(\bm{x}_i, \bm{\beta}), \sigma^2\right)$, independently, for $i=1,\hdots,n$.
So, we can follow the general theory of the MDPDE under INH setup, as described briefly in Web Appendix A.1, to define and study the MDPDE of the parameters under the present NLR model. Following the notations of Web Appendix A.1, the model density for $i$-th observation is $f_i(\cdot; \bm{\theta}) \equiv \mathcal{N}\left(\mu(\bm{x}_i, \bm{\beta}), \sigma^2\right)$ and we have
\begin{equation} \label{Vinlr}
    V_i(y_i;\bm{\theta}) = \frac{1}{\left(\sigma\sqrt{2\pi}\right)^\alpha\sqrt{1+\alpha}} - \frac{1+\alpha}{\alpha}\frac{1}{\left(\sigma\sqrt{2\pi}\right)^\alpha}e^{-\frac{\alpha}{2\sigma^2}(y_i-\mu(\bm{x}_i, \bm{\beta}))^2},\ \ \alpha>0.
\end{equation}
Then the MDPDE of the parameter $\bm{\theta} = \left(\bm{\beta}^T, \sigma^2\right)^T$ under our NLR model can be defined as the minimizer of the objective function
\begin{equation}
\label{Hnlr}
H_n(\bm{\theta}) = \frac{1}{\left(\sigma\sqrt{2\pi}\right)^\alpha\sqrt{1+\alpha}} - \frac{1+\alpha}{\alpha}\frac{1}{\left(\sigma\sqrt{2\pi}\right)^\alpha}\frac{1}{n}\sum_{i=1}^n e^{-\frac{\alpha}{2\sigma^2}(y_i-\mu(\bm{x}_i, \bm{\beta}))^2},\ \ \alpha>0.
\end{equation}
Now, differentiating \eqref{Hnlr} partially with respect to $\bm{\beta}$ and $\sigma^2$, and equating to $0$, the estimating equations for the MDPDE are obtained as 
\begin{equation}
\label{esteq1}
    \sum_{i=1}^n e^{-\frac{\alpha}{2\sigma^2}(y_i-\mu(\bm{x}_i, \bm{\beta}))^2}(y_i-\mu(\bm{x}_i, \bm{\beta}))\nabla_{\bm{\beta}}\mu(\bm{x}_i, \bm{\beta}) = 0,
\end{equation}
\begin{equation}
\label{esteq2}
    \sum_{i=1}^n e^{-\frac{\alpha}{2\sigma^2}(y_i-\mu(\bm{x}_i, \bm{\beta}))^2} \left[1 - \frac{(y_i-\mu(\bm{x}_i, \bm{\beta}))^2}{\sigma^2}\right] = \frac{n\alpha}{(1+\alpha)^{3/2}},
\end{equation}
where $\nabla_{\bm{\beta}}$ denotes the first order partial derivative with respect to $\bm{\beta}$. To obtain the MDPDE $\bm{\widehat{\theta}}_\alpha = (\bm{\widehat{\beta}}_\alpha^T, \widehat{\sigma}_\alpha^2)^T$ with tuning parameter $\alpha>0$, in practice, one can numerically solve the above two estimating equations simultaneously for $\bm{\theta} = (\bm{\beta}^T, \sigma^2)^T$, which is equivalent to minimize \eqref{Hnlr}. Clearly, the MDPDE coincides with the MLE at $\alpha=0$ in a limiting sense, and provides a robust generalization of the MLE at $\alpha>0$.

Note that, an important advantage of our MDPDE over its existing robust competitors (briefly described in Web Appendix A.2) is that the MDPDEs additionally yield a robust estimate of the scale parameter $\sigma$ simultaneously with the estimation of the regression parameter $\bm{\beta}$. \btxt{Further, unlike other estimating equation based estimators, the MDPDE is associated with a nice objective function, which helps us to avoid the problem of multiple roots of the estimating equations. If we get multiple solutions to the MDPDE estimating equations \eqref{esteq1}-\eqref{esteq2}, we should choose the one yielding the minimum value of the associated objective function given in \eqref{Hnlr}. The same can also be used while choosing between multiple local minima of the said objective function.}

\subsection{Asymptotic properties}\label{APNLR}
For simplicity, let us assume that the true data-generating density $g_i$ of $y_i$ belongs to the corresponding parametric model family $\mathcal{F}_i = \{\mathcal{N}(\mu(\bm{x}_i, \bm{\beta}), \sigma^2): \bm{\theta} = (\bm{\beta}^T,\sigma^2)^T \in \bm{\Theta}\}$, $i=1,2,\hdots,n,$ i.e., $g_i \equiv \mathcal{N}\left(\mu(\bm{x}_i, \bm{\beta}_0), \sigma_0^2\right), \mbox{ for some } \bm{\theta}_0 = \left(\bm{\beta}_0^T, \sigma_0^2 \right)^T\in\bm{\Theta}$. Our NLR model \eqref{nlr} can be expressed in vectorized form as
\begin{equation}
\label{vnlr2}
    \bm{y} = \bm{\mu}(\bm{\beta}) + \bm{\epsilon},
\end{equation}
where $\bm{y}=(y_1,\hdots,y_n)^T, \bm{\epsilon} = (\epsilon_1,\hdots,\epsilon_n)^T$ and $\bm{\mu}(\bm{\beta}) = \left(\mu(\bm{x}_1, \bm{\beta}),\hdots, \mu(\bm{x}_n, \bm{\beta})\right)^T$. 
Now we define the $n\times p$ matrix
\begin{equation}\label{F.b}
    \bm{\dot{\mu}}(\bm{\beta}) = \frac{\partial\bm{\mu}(\bm{\beta})}{\partial\bm{\beta}^T} = 
\begin{bmatrix}
    \nabla_{\bm{\beta}}\mu(\bm{x}_1,\bm{\beta}) &
    \nabla_{\bm{\beta}}\mu(\bm{x}_2,\bm{\beta}) &
    \hdots &
    \nabla_{\bm{\beta}}\mu(\bm{x}_n,\bm{\beta})\\
\end{bmatrix}^T,
\end{equation}
where $\nabla_{\bm{\beta}}$ denotes the partial derivative with respect to $\bm{\beta}$, and the $(p+1)\times(p+1)$ matrices
\begin{equation} \label{Psi}
    \bm{\Psi}_n =
    \begin{bmatrix}
        \frac{\zeta_\alpha}{n} \left(\bm{\dot{\mu}}(\bm{\beta})^T\bm{\dot{\mu}}(\bm{\beta})\right) & \bm{0}_p\\
        \bm{0}_p^T & \varsigma_\alpha
    \end{bmatrix} \mbox{ and }
   \bm{\Omega}_n =
    \begin{bmatrix}
        \frac{\zeta_{2\alpha}}{n} \left(\bm{\dot{\mu}}(\bm{\beta})^T\bm{\dot{\mu}}(\bm{\beta})\right) & \bm{0}_p\\
        \bm{0}_p^T & \varsigma_{2\alpha} - \frac{\alpha^2}{4}\zeta_\alpha^2
    \end{bmatrix},
\end{equation}
where $\zeta_\alpha = (2\pi)^{-\alpha/2}\sigma^{-(\alpha+2)}(1+\alpha)^{-3/2}$ and $\varsigma_\alpha = (2\pi)^{-\alpha/2}\sigma^{-(\alpha+4)}\frac{2+\alpha^2}{4(1+\alpha)^{5/2}}$. Let us further denote the $(i,j)$-the entry of the matrix $\bm{\dot{\mu}}(\bm{\beta})$ by $\dot{\mu}_{ij}(\bm{\beta})$ for all $i=1,\ldots n$ and $j=1,\ldots,p$.

We now present some conditions on the NLR model, under which the asymptotic properties of the MDPDE $\bm{\widehat{\theta}}_\alpha$ would be derived. Let us denote $\bm{\Theta_\beta}\ (\subseteq \mathbb{R}^p)$ be the parameter space of the model parameter $\bm{\beta}$ of the NLR model \eqref{nlr} and $\sigma^2>0$. So, $\bm{\Theta} = \bm{\Theta_\beta} \times \mathbb{R}^+$.\\\\
\textbf{Conditions for consistency and asymptotic normality:}\\
\textbf{(R1)} The mean function $\mu(\bm{x}_i, \bm{\beta})$ is thrice continuously differentiable with respect to $\bm{\beta}\in\omega$, 
an open subset of $\bm{\Theta_\beta}$ containing the true parameter value ($\bm{\beta}_0$). Further, for  all $i,j,k,l$ and  $\bm{\beta}\in\omega$, we have
\begin{equation} \label{R11}
    \sup_{n>1}\max_{1\leq i\leq n}\left|\dot{\mu}_{ij}(\bm{\beta})\right| = O(1),\ \ \ \ \sup_{n>1}\max_{1\leq i\leq n}\left|\dot{\mu}_{ij}(\bm{\beta})\dot{\mu}_{ik}(\bm{\beta})\right| = O(1),
\end{equation}
\begin{equation} \label{R12}
 \mbox{and} ~~~~~~~~~~  \frac{1}{n}\sum_{i=1}^n\left|\dot{\mu}_{ij}(\bm{\beta})\dot{\mu}_{ik}(\bm{\beta})\dot{\mu}_{il}(\bm{\beta})\right| = O(1).~~~~~~~~~~
\end{equation}
\textbf{(R2)} For all $\bm{\beta}\in\omega$, the matrix $\bm{\dot{\mu}}(\bm{\beta})$ satisfies
\begin{equation} \label{R21}
    \inf_n \left[\mbox{minimum eigenvalue of } \frac{1}{n}\bm{\dot{\mu}}(\bm{\beta})^T\bm{\dot{\mu}}(\bm{\beta})\right] > 0,
\end{equation}
which says that $\bm{\dot{\mu}}(\bm{\beta})$ is of full column rank, and
\begin{equation} \label{R22}
    \max_{1\leq i\leq n}\left[\nabla_{\bm{\beta}}\mu(\bm{x}_i, \bm{\beta})^T \left(\bm{\dot{\mu}}(\bm{\beta})^T\bm{\dot{\mu}}(\bm{\beta})\right)^{-1} \nabla_{\bm{\beta}}\mu(\bm{x}_i, \bm{\beta})\right] = O(n^{-1}).
\end{equation}
With the above-mentioned conditions, the asymptotic distributions of the MDPDEs of the NLR model parameters are given in the following theorem; its proof is given in Web Appendix B.
\begin{theorem}\label{th-2.1}
    Under the setup of the NLR model \eqref{nlr}, suppose that the true data-generating distributions belong to the corresponding parametric model families with the true parameter value being $\bm{\theta}_0 = \left(\bm{\beta}_0^T, \sigma_0^2\right)^T$ and Conditions (R1)-(R2) hold. Then, for any  $\alpha\geq 0$, we have the following results.
    \begin{enumerate}[label=(\roman*)]
    \item There exists a consistent sequence of roots $\bm{\widehat{\theta}}_\alpha = \left(\bm{\widehat{\beta}}_\alpha^T, \widehat{\sigma}_\alpha^2\right)^T$ 
    of the MDPDE estimating equations \eqref{esteq1}--\eqref{esteq2}.
    \item Asymptotic distributions of $\bm{\widehat{\beta}}_\alpha$ and $\widehat{\sigma}_\alpha^2$ are independent.
    \item
    Asymptotically, as $n\to\infty$,
    \[\left(\bm{\dot{\mu}}(\bm{\beta}_0)^T\bm{\dot{\mu}}(\bm{\beta}_0)\right)^{1/2}\left(\bm{\widehat{\beta}}_\alpha - \bm{\beta}_0\right) \xrightarrow{\mathcal{D}} \mathcal{N}_p\left(\bm{0}_p,\ v_{1\alpha}\sigma_0^2 \mathbb{I}_p\right),
\mbox{    and }
\sqrt{n}\left(\widehat{\sigma}_\alpha^2 - \sigma_0^2\right)  \xrightarrow{\mathcal{D}} \mathcal{N}\left(0,\ v_{2\alpha}\sigma_0^4\right),\]
    \[ \mbox{where }~~~~~~~~~~
    v_{1\alpha} = \frac{1}{\sigma^2}\frac{\zeta_{2\alpha}}{\zeta_\alpha^2} = \left(1 + \frac{\alpha^2}{1+2\alpha}\right)^{3/2},
    ~~~~~~~~~~~~~~~~~~~~~~~~~~~~~~~~~~~~~~~~~~~~~~~~~~~~~~~~~~~~~~~~~~~~~~~~~~~\]
    \[ \mbox{and }~~~~~~~~v_{2\alpha} = \frac{1}{\sigma^4}\frac{\varsigma_{2\alpha}-\frac{\alpha^2}{4}\zeta_\alpha^2}{\varsigma_\alpha^2} = \frac{4}{\left(2+\alpha^2\right)^2}\left[2\left(1+2\alpha^2\right)\left(1 + \frac{\alpha^2}{1+2\alpha}\right)^{5/2} - \alpha^2(1+\alpha)^2\right].\]
    \end{enumerate}
\end{theorem}

\begin{remark}
    For $\alpha=0$, the above theorem gives the asymptotic normality result for the OLS or the MLE under NLR models (with IID normal errors), which is exactly the same as discussed in Theorem 5 of Wu \cite{wu1981asymptotic}, for the OLS.
\end{remark}
Let us now study the asymptotic relative efficiency (ARE) of the MDPDE with respect to the fully efficient MLE. From Theorem \ref{th-2.1} it is easy to see that the ARE of the MDPDE $\bm{\widehat{\beta}}_\alpha$ of $\bm{\beta}$ is the same for all $\beta_i$'s and is given by
\[\frac{v_{10}}{v_{1\alpha}}\times 100\% = \left(1 + \frac{\alpha^2}{1+2\alpha}\right)^{-3/2}\times 100\%.\]
Similarly the ARE of $\widehat{\sigma}_\alpha^2$ is given by
\[\frac{v_{20}}{v_{2\alpha}}\times 100\% = \frac{(2+\alpha^2)^2}{2}\left[2(1+2\alpha^2)\left(1 + \frac{\alpha^2}{1+2\alpha}\right)^{5/2} - \alpha^2(1+\alpha)^2\right]^{-1}\times 100\%.\]
These AREs are independent of the choice of the (non-linear) regression function $\mu(\cdot,\cdot)$ and exactly the same as those obtained for the linear regression model in Ghosh and Basu \cite{ghosh2013robust}. The exact values of the ARE of the MDPDE of $\bm{\beta}$ for different $\alpha$ are reported in \btxt{Table C.1} of Web Appendix C, along with the same for the existing robust competitors. From the table, it is evident that the loss in efficiency of the MDPDE is very small for smaller values of $\alpha$, compared to other competing estimators at suitable tuning parameter values, expecting to produce similar robustness properties.

\subsection{Influence function and sensitivity analysis of the proposed MDPDE}\label{IFNLR}
Influence function (IF) is a measure of the local robustness property of an estimator. Here, we will follow the general theory from Ghosh and Basu \cite{ghosh2013robust}, as presented briefly in Web Appendix A.1. 

Firstly, the minimum DPD functional $\bm{T}_\alpha = \left(\bm{T}_\alpha^{\bm{\beta}}, T_\alpha^{\sigma^2}\right)$ of the parameter $\bm{\theta} = \left(\bm{\beta}^T, \sigma^2\right)^T$ under the setup of the NLR model \eqref{nlr}, with true distribution $g_i$ and $f_i(\cdot; \bm{\theta}) \in \mathcal{F}_i$, is defined by
\begin{equation}\label{nlrMDPDfn}
    \bm{T}_\alpha(\bm{G}) = \arg\min_{\bm{\theta\in\bm{\Theta}}}\frac{1}{n}\sum_{i=1}^n\left[\frac{1}{\left(\sigma\sqrt{2\pi}\right)^\alpha\sqrt{1+\alpha}} - \frac{1+\alpha}{\alpha}\int f_i^\alpha(y;\bm{\theta})g_i(y)dy \right],
\end{equation}
where $\bm{G} = (G_1, \cdots, G_n)^T$ with $G_i$ being the distribution function for the density $g_i$ for each $i=1, \ldots, n$. Then, assuming the true distributions belonging to the corresponding model families and using the general expression of IF under INH setup (as given in Eq. (3) in Web Appendix A.1), the IF of $\bm{T}_\alpha$ with contamination only in the $i_0$-th observation at the contamination point $t_{i_0}$ is obtained as
\begin{equation*}
    IF_{i_0}(t_{i_0}, \bm{T}_\alpha, \bm{F}_{\bm{\theta}}) = 
    \begin{pmatrix}
       \frac{1}{\zeta_\alpha} f_{i_0}^\alpha(t_{i_0};\bm{\theta}) \frac{(t_{i_0}-\mu(\bm{x}_{i_0}, \bm{\beta}))}{\sigma^2}\left(\bm{\dot{\mu}}(\bm{\beta})^T\bm{\dot{\mu}}(\bm{\beta})\right)^{-1} \nabla_{\bm{\beta}}\mu(\bm{x}_{i_0}, \bm{\beta})\\\\
       
       \frac{1}{n\varsigma_\alpha}\left[\left\{\frac{(t_{i_0}-\mu(\bm{x}_{i_0}, \bm{\beta}))^2}{2\sigma^4} - \frac{1}{2\sigma^2} \right\}f_{i_0}^\alpha(t_{i_0};\bm{\theta}) + \frac{\alpha}{2}\zeta_\alpha\right]
    \end{pmatrix},
\end{equation*}
where $\bm{F}_{\bm{\theta}} = \left(F_{1,\bm{\theta}}, F_{2,\bm{\theta}},\hdots,F_{n,\bm{\theta}} \right)^T$ with $F_{i,\bm{\theta}}$ being the distribution function of the model density $f_i(\cdot;\bm{\theta})$. On simplification, the IFs of $\bm{T}_\alpha^{\bm{\beta}}$ and $T_\alpha^{\sigma^2}$ with contamination only in the $i_0$-th observation at the contamination point $t_{i_0}$, respectively, become
\begin{equation}\label{IFbeta}
    IF_{i_0}(t_{i_0}, \bm{T}_\alpha^{\bm{\beta}}, \bm{F}_{\bm{\theta}}) = (1+\alpha)^\frac{3}{2} (t_{i_0}-\mu(\bm{x}_{i_0}, \bm{\beta})) e^{-\frac{\alpha}{2\sigma^2}(t_{i_0}-\mu(\bm{x}_{i_0}, \bm{\beta}))^2} \left(\bm{\dot{\mu}}(\bm{\beta})^T\bm{\dot{\mu}}(\bm{\beta})\right)^{-1} \nabla_{\bm{\beta}}\mu(\bm{x}_{i_0}, \bm{\beta}),
\end{equation}
\begin{multline}\label{IFsigma}
    IF_{i_0}(t_{i_0}, T_\alpha^{\sigma^2}, \bm{F}_{\bm{\theta}}) = 
    \frac{2(1+\alpha)^\frac{5}{2}}{n(2+\alpha^2)} \left\{(t_{i_0}-\mu(\bm{x}_{i_0}, \bm{\beta}))^2 - \sigma^2\right\} e^{-\frac{\alpha}{2\sigma^2}(t_{i_0}-\mu(\bm{x}_{i_0}, \bm{\beta}))^2} + \frac{2\alpha(1+\alpha)^2}{n(2+\alpha^2)}.
\end{multline}
Note that, the functions $z\mapsto e^{-kz^2},\ z\mapsto ze^{-kz^2},\ z\mapsto z^2e^{-kz^2}$ are all bounded in $\mathbb{R}$, for any $k>0$, and hence the above IFs in \eqref{IFbeta}, \eqref{IFsigma} are bounded in the contamination point $t_{i_0}, \mbox{ for all } \alpha>0 \mbox{ and any } i_0$. This ensures that, for $\alpha>0$, the MDPDEs are robust with respect to any outlying observation. Also note that, the IFs are unbounded at $\alpha = 0$, which is actually the case of the non-robust MLEs.

Further, according to the general formula given in Eq. (4) of Web Appendix A.1, the IFs of $\bm{T^\beta}_\alpha$ and $T_\alpha^{\sigma^2}$ with contamination in all the observations at the contamination points in $\bm{t} =(t_1, \cdots, t_n)^T$, are, respectively, given by
\begin{multline}
    IF(\bm{t}, \bm{T}_\alpha^{\bm{\beta}}, \bm{F}_{\bm{\theta}}) = 
    (1+\alpha)^\frac{3}{2}\left(\bm{\dot{\mu}}(\bm{\beta})^T\bm{\dot{\mu}}(\bm{\beta})\right)^{-1} \sum_{i=1}^n\nabla_{\bm{\beta}}\mu(\bm{x}_i, \bm{\beta}) (t_i-\mu(\bm{x}_i, \bm{\beta})) e^{-\frac{\alpha}{2\sigma^2}(t_i-\mu(\bm{x}_i, \bm{\beta}))^2},
\end{multline}
\begin{multline}
    IF(\bm{t}, T_\alpha^{\sigma^2}, \bm{F}_{\bm{\theta}}) =
    \frac{2(1+\alpha)^\frac{5}{2}}{n(2+\alpha^2)} \sum_{i=1}^n\left\{(t_i-\mu(\bm{x}_i, \bm{\beta}))^2 - \sigma^2\right\} e^{-\frac{\alpha}{2\sigma^2}(t_i-\mu(\bm{x}_i, \bm{\beta}))^2} + \frac{2\alpha(1+\alpha)^2}{(2+\alpha^2)}.
\end{multline}
It is easy to observe that these IFs are also bounded in each argument $t_i,\mbox{ for } i=1,2,\hdots,n$ for all $\alpha>0$, and unbounded at $\alpha=0$, as before.

Let us now consider some IF-based summary measures of robustness for our NLR model. The gross error sensitivity of the minimum DPD functional $\bm{T^\beta}_\alpha$ of $\bm{\beta}$ at the vector of model distributions $\bm{F_\theta}$  with contamination only in the $i_0$-th observation is given by
\btxt{\begin{align} \label{GES}
    \nonumber
    \gamma_{i_0}^u(\bm{T}_\alpha^{\bm{\beta}}, \bm{F}_{\bm{\theta}}) &= \sup_t \left\{\lVert IF_{i_0}(t, \bm{T}_\alpha^{\bm{\beta}}, \bm{G})\rVert \right\} \\
    &=\left\{ 
    \begin{matrix}
        \frac{(1+\alpha)^\frac{3}{2}}{\sqrt{\alpha}}\sigma e^{-\frac{1}{2}} \lVert\left(\bm{\dot{\mu}}(\bm{\beta})^T\bm{\dot{\mu}}(\bm{\beta})\right)^{-1} \nabla_{\bm{\beta}}\mu(\bm{x}_{i_0}, \bm{\beta})\rVert, & \mbox{if} & \alpha>0,\\
        \infty, & \mbox{if} & \alpha=0.
    \end{matrix}
    \right.
\end{align}}
But this measure does not satisfy the scale invariance property with respect to scale transformation of the components of $\bm{\beta}$. To overcome this problem, we consider a modified measure, namely the self-standardized sensitivity. For contamination only in the $i_0$-th observation, it is given by
\begin{multline}\label{SSS}
    \gamma_{i_0}^s(\bm{T}_\alpha^{\bm{\beta}}, \bm{G})
    =\sup_t \left\{IF_{i_0}(t, \bm{T}_\alpha^{\bm{\beta}}, \bm{G})^T (\bm{\Psi}_{pn}^{-1}\bm{\Omega}_{pn}\bm{\Psi}_{pn}^{-1})^{-1} IF_{i_0}(t, \bm{T}_\alpha^{\bm{\beta}}, \bm{G})\right\}^{1/2}\\
    =\left\{ 
    \begin{matrix}
        \frac{(1+2\alpha)^\frac{3}{4}}{\sqrt{n\alpha}} e^{-\frac{1}{2}} \left\{\nabla_{\bm{\beta}}\mu(\bm{x}_{i_0},\bm{\beta})^T \left(\bm{\dot{\mu}}(\bm{\beta})^T\bm{\dot{\mu}}(\bm{\beta})\right)^{-1} \nabla_{\bm{\beta}}\mu(\bm{x}_{i_0},\bm{\beta})\right\}^{1/2}, & \mbox{if} & \alpha>0,\\
        \infty, & \mbox{if} & \alpha=0,
    \end{matrix}
    \right.
\end{multline}
where $\bm{\Psi}_{pn},\bm{\Omega}_{pn}$ are the $p\times p$ principal diagonal blocks of the matrices $\bm{\Psi}_n,\bm{\Omega}_n$, respectively. Note that both the above sensitivity measures are decreasing functions in $\alpha\ (>0)$. So, it is evident that, with the presence of outliers, the robustness of the MDPDE of $\bm{\beta}$ increases with increasing values of $\alpha$. However, in Section \ref{APNLR} we have observed that the ARE of MDPDE of $\bm{\beta}$ (under pure data) decreases with increasing values of $\alpha$. Together, the tuning parameter $\alpha$ acts as a trade-off between the efficiency and robustness of the MDPDE under NLR models, consistent with the existing literature of DPD-based procedures.

For contamination in all or some observations, the gross error sensitivity and the self-standardized sensitivity can similarly be defined as in \eqref{GES} and \eqref{SSS}, but replacing $IF_{i_0}(t, \bm{T}_\alpha^{\bm{\beta}}, \bm{G})$ by $IF(\bm{t}, \bm{T}_\alpha^{\bm{\beta}}, \bm{G})$ and taking supremum over all possible $t_1\hdots,t_n$. The resulting implications are again the same, and hence these details are omitted for brevity.

\section{Robust Wald-type tests based on the MDPDE} \label{hyptesting}
Basu et al.~\cite{basu2016generalized} first defined and studied a generalized class of robust Wald-type tests based on the MDPDE under the IID setup, which was later extended for the general INH setup in Basu et al.~\cite{basu2018robust}. Here, we use the same approach to develop robust tests of hypotheses for our NLR models. 

\subsection{General Wald-type tests of NLR model parameters}\label{WTTNHOMO}
Consider the setup of the NLR model \eqref{nlr} and suppose that we wish to test the general linear hypothesis on the regression parameters $\bm{\beta}$, as given by
\begin{equation}\label{nlrhyp}
    H_0: \bm{L}\bm{\beta} = \bm{l}_0 \mbox{ against } H_1: \bm{L\beta}\neq \bm{l}_0,
\end{equation}
where $\bm{L}$ is a given $r\times p$ matrix with rank$(\bm{L}) = r$ and $\bm{l}_0$ is a given $r$-vector. Recall that the full parameter vector is $\bm{\theta} = (\bm{\beta}^T,\sigma^2)^T$ with the parameter space $\bm{\Theta} \subseteq \mathbb{R}^p \times \mathbb{R}^+$. The restricted parameter space under the null hypothesis in \eqref{nlrhyp} is given by
$    \bm{\Theta}_0 = \left\{\left(\bm{\beta}^T,\sigma^2\right)^T: \bm{L\beta = l_0} \mbox{ and } 0<\sigma^2<\infty \right\}$.
Also, the MDPDE of $\bm{\theta} = \left(\bm{\beta}^T,\sigma^2\right)^T$ with tuning parameter $\alpha\geq 0$ has been denoted by $\bm{\widehat{\theta}}_\alpha = (\bm{\widehat{\beta}}_\alpha^T, \widehat{\sigma}_\alpha^2)^T$. Then, we define the corresponding Wald-type test statistics for testing \eqref{nlrhyp} as given by
\begin{equation*}
    W_{n,\alpha} = \frac{1}{v_{1\alpha}\widehat{\sigma}_\alpha^2}\bm{\left(L\widehat{\beta}_\alpha - l_0\right)^T \left[L\left(\dot{\mu}\left(\widehat{\beta}_\alpha\right)^T\dot{\mu}\left(\widehat{\beta}_\alpha\right)\right)^{-1}L^T\right]^{-1} \left(L\widehat{\beta}_\alpha - l_0\right)},
\end{equation*}
where $v_{1\alpha}$ is as defined in Theorem 2.1. Note that, under the null hypothesis in \eqref{nlrhyp}, one can easily show that, $W_{n,\alpha}$ asymptotically follows $\chi^2_r$, the chi-square distribution with $r$ degrees of freedom. So, the null hypothesis will be rejected at level of significance $\gamma$ if the observed value of $W_{n,\alpha} > \chi^2_{\gamma;r}$, the $(1 - \gamma)$-th quantile of $\chi^2_r$ distribution.
\begin{remark}
    The above Wald-type test coincides with the classical Wald-test at $\alpha=0$, due to the fact that the MDPDE coincides with the MLE at $\alpha = 0$.
\end{remark}
By the general theory discussed in Basu et al.~\cite{basu2018robust}, the above Wald-type test can be shown to be consistent at any fixed alternative for every value of $\alpha\geq 0$. Now, let us consider the contiguous alternative hypotheses of the form
\begin{equation}\label{nlrcontalt}
    H_{1,n}: \bm{L\beta} = \bm{l}_n, \mbox{ where } \bm{l}_n = \bm{l}_0 + n^{-1/2}\bm{d},
\end{equation}
with $\bm{d}\in\mathbb{R}^r \setminus \{\bm{0}_r\}$ be a fixed vector. Then, by Theorem 8 of Basu et al.~\cite{basu2018robust}, under $H_{1,n}$, $W_{n,\alpha} \xrightarrow{\mathcal{D}} \chi^2_r(\delta)$, the non-central chi-square distribution with degrees of freedom $r$ and the non-centrality parameter $\delta$, where
\begin{equation}\label{ncp}
    \delta = \frac{1}{v_{1\alpha}\sigma_0^2}\bm{d}^T \left[\bm{L}\bm{M}^{-1}(\bm{\beta}_0) \bm{L}^T\right]^{-1}\bm{d}, \mbox{ with } \bm{M}(\bm{\beta}_0) = \lim_{n\to\infty}\frac{1}{n} \bm{\dot{\mu}}(\bm{\beta}_0)^T\bm{\dot{\mu}}(\bm{\beta}_0), ~~ \left(\bm{\beta}_0^T, \sigma_0^2\right)^T \in \bm{\Theta}_0.~~
\end{equation}
Hence, an approximate expression of asymptotic contiguous power under $H_{1,n}$ in \eqref{nlrcontalt} is given by
$\Pi_{\alpha} = 1 - G_{\chi^2_r(\delta)}\left(\chi^2_{\gamma;r}\right)$,
where $G_{\chi^2_r(\delta)}(\cdot)$ denotes the distribution function of $\chi^2_r(\delta)$.

\subsection{Influence functions of the Wald-type tests}
Let us now consider the setup of Section \ref{IFNLR} to study the robustness of the proposed Wald-type test statistics. We first define the statistical functional corresponding to the test statistics $W_{n,\alpha}$ as 
\begin{equation*}
    W_\alpha(\bm{G}) = \left(\bm{LT^\beta}_\alpha(\bm{G}) - \bm{l_0}\right)^T \left[\bm{L}\left(\frac{\bm{\dot{\mu}}\left(\bm{T^\beta}_\alpha(\bm{G})\right)^T\bm{\dot{\mu}}\left(\bm{T^\beta}_\alpha(\bm{G})\right)}{nv_{1\alpha}T_\alpha^{\sigma^2}(\bm{G})}\right)^{-1}\bm{L}^T\right]^{-1}\left(\bm{LT^\beta}_\alpha(\bm{G})- \bm{l_0}\right),
\end{equation*}
where $\bm{T^\beta}_\alpha$ and $T_\alpha^{\sigma^2}$ are the minimum DPD functional corresponding to parameter $\bm{\beta}$ and $\sigma^2$, respectively, as defined in Section \ref{IFNLR}. 

Let $\bm{\theta}_0 = \left(\bm{\beta}_0^T, \sigma_0^2\right)^T \in \bm{\Theta}_0$ and $\bm{F}_{\bm{\theta}_0}$ be the vector of the model distributions under the null hypotheses. It follows from the theory of general Wald-type tests (see, e.g., \cite{basu2016generalized,basu2018robust}) that the first-order IF of the test functional $W_\alpha$ (as defined for the estimators earlier) is zero when evaluated at $\bm{G}=\bm{F}_{\bm{\theta}_0}$, and so we need to consider the second-order IF.

\btxt{The second-order IF of any statistical functional $W$, at the vector of distribution functions $\bm{G}$, with contamination only at the $i$-th observation at the contamination point $t_i$ is defined as
\begin{equation*}
    IF_i^{(2)}(t_i; W,\bm{G}) = \frac{\partial^2}{\partial\epsilon^2} W(G_1,\ldots,G_{i-1}, G_{i,\epsilon}, G_{i+1}, \ldots, G_n)\bigg|_{\epsilon=0},
\end{equation*}
where $G_{i,\epsilon} = (1-\epsilon)G_i + \epsilon\Delta_{t_i}$, with $\Delta_{t_i}$ being the one-point distribution supported on $\{t_i\}$ (see, e.g., \cite{ghosh2016influence, basu2018robust}). 
Intuitively, it measures the second-order bias approximation caused by infinitesimal contamination at the outlying point. 
For the present Wald-type test functional $W_\alpha$, this second-order IF at $\bm{G} = \bm{F}_{\bm{\theta}_0}$ turns out to be 
\begin{multline*}
    IF_{i_0}^{(2)}\left(t_{i_0}; W_\alpha, \bm{F}_{\bm{\theta}_0}\right)
    = \frac{2}{nv_{1\alpha}\sigma_0^2}IF_{i_0}\left(t_{i_0}; \bm{T^\beta}_\alpha, \bm{F}_{\bm{\theta}_0}\right)^T\bm{L}^T\left[\bm{L}\left(\bm{\dot{\mu}}(\bm{\beta}_0)^T\bm{\dot{\mu}}(\bm{\beta}_0)\right)^{-1}\bm{L}^T\right]^{-1}\bm{L}~IF_{i_0}\left(t_{i_0}; \bm{T^\beta}_\alpha, \bm{F}_{\bm{\theta}_0}\right),
\end{multline*}
when there is contamination only in the $i_0$-th observation.}  Note that, this second-order IF directly depends on the IF of the MDPDE functional $\bm{T^\beta}_\alpha$ of $\bm{\beta}$. 
Since $IF_{i_0}(t_{i_0}; \bm{T^\beta}_\alpha, \bm{F}_{\bm{\theta}_0})$ is bounded for all $\alpha > 0$ and unbounded at $\alpha = 0$, 
so is $IF_{i_0}^{(2)}\left(t_{i_0}; W_\alpha, \bm{F}_{\bm{\theta}_0}\right)$. This ensures the robustness of the Wald-type test for any $\alpha>0$ as well. 

We can also study the robustness of the level and power of these Wald-type tests by examining their level influence function (LIF) and power influence function (PIF). 
\btxt{A brief discussion of the LIF and PIF under general INH set-ups (including their definitions) and for the Wald-type test statistics is provided in Web Appendix A.3 for the sake of completeness.} 
From Equations (9)-(10) presented there (or, directly from Theorem 12 of Basu et al.~\cite{basu2018robust}), we can conclude that the LIF of the MDPDE-based Wald-type test is identically zero whenever the IF of the MDPDE used in constructing the test is bounded, i.e., for all $\alpha>0$. Further, its PIF at level of significance $\gamma$ is given by 
\begin{equation*}
    PIF\left(\bm{t}; W_\alpha, \bm{F}_{\bm{\theta}_0}\right)
    = K_r^*(\delta) \frac{1}{nv_{1\alpha}\sigma_0^2}\bm{d}^T\left[\bm{L}\left(\bm{\dot{\mu}}(\bm{\beta}_0)^T\bm{\dot{\mu}}(\bm{\beta}_0)\right)^{-1}\bm{L}^T\right]^{-1} IF\left(\bm{t}; \bm{T}_\alpha, \bm{F}_{\bm{\theta}_0}\right),
\end{equation*}
where $\delta$ is as given in \eqref{ncp} and 
\[K_r^*(s) = e^{-\frac{s}{2}}\sum_{v=0}^\infty\frac{s^{v-1}}{v!2^v}(2v-s) P\left(\chi^2_{r+2v} > \chi^2_{\gamma; p}\right).\]
Note that, the PIF is linear in the IF of the MDPDE, and hence is bounded for any $\alpha > 0$. This implies the robustness of the power of our Wald-type tests.

\subsection{Tests for scalar parameter against one-sided alternative}
Let us now consider the problem of testing a single parameter component against a one-sided alternative, i.e., testing the hypothesis  
\begin{equation}\label{osalt}
    H_0: \beta_k = \beta_k^0 \mbox{ against } H_1: \beta_k > \beta_k^0,
\end{equation}
where $\beta_k$ is the $k$-th component of regression parameter $\bm{\beta}$ in our NLR model \eqref{nlr}. Under the both-sided alternative $H_1^*: \beta_k \neq \beta_k^0$, the general MDPDE-based Wald-type test statistics can be obtained from the discussion of the previous subsection, which asymptotically follow $\chi^2_1$. However, for testing the one-sided alternative in \eqref{osalt}, since we are dealing with a scalar parameter, we can use the test statistic
\begin{equation}\label{WTTSforSP}
    \widetilde{W}_{n,\alpha} = \frac{\widehat{\beta}_{k\alpha} - \beta_k^0}{\sqrt{v_{1\alpha}\widehat{\sigma}_\alpha^2s_{kk}}},
\end{equation}
where $\widehat{\beta}_{k\alpha}$ is the $k$-th component of $\bm{\widehat{\beta}}_\alpha$, and $s_{kk}$ is the $(k, k)$-th element of the matrix $(\bm{\dot{\mu}}(\widehat{\bm{\beta}}_\alpha)^T\bm{\dot{\mu}}(\widehat{\bm{\beta}}_\alpha))^{-1}$. It is evident from Theorem \ref{th-2.1} that, under $H_0$ in \eqref{osalt}, $\widetilde{W}_{n,\alpha}$ asymptotically follows a standard normal distribution. So, this null hypothesis will be rejected against $H_1$ in \eqref{osalt}, at level of significance $\gamma$, if the observed value of $W_{n,\alpha}$ exceeds $z_\gamma$, the $(1 - \gamma)$-th quantile of $\mathcal{N}(0,1)$.

\textcolor{black}{Further, one can observe that, under the one-sided contiguous alternative hypotheses of the form
$H_{1,n}: \beta_k = \beta_k^0 + n^{-1/2}d,\ d>0$,
we have
\begin{equation*}
    \frac{\left(\widehat{\beta}_{k\alpha} - \beta_k^0 - n^{-1/2}d\right)}{\sqrt{v_{1\alpha}\widehat{\sigma}_\alpha^2s_{kk}}} \xrightarrow[]{\mathcal{D}} \mathcal{N}(0,1), 
    \mbox{ and hence }~ \left[\widetilde{W}_{n,\alpha} - \frac{d}{\sqrt{v_{1\alpha}\widehat{\sigma}_\alpha^2 n s_{kk}}}\right] \xrightarrow[]{\mathcal{D}} \mathcal{N}(0,1).
\end{equation*}
Thus, the asymptotic contiguous power of the test based on $\widetilde{W}_{n,\alpha}$ is given by
\begin{equation} \label{cont-power}
    \widetilde{\Pi}_{\alpha} = \lim_{n\to\infty} P_{H_{1,n}}\left[\widetilde{W}_{n,\alpha}>z_\gamma\right] = 1 - \Phi\left(z_\gamma - \left(1 + \frac{\alpha^2}{1+2\alpha}\right)^{-3/4}\frac{d^*}{\sigma}\right),
\end{equation}
where $d^* = \frac{d}{\sqrt{s_{kk}^*}}$ with $s_{kk}^* = \lim_{n\to\infty} n s_{kk}$. 
Note that this asymptotic contiguous power $\widetilde{\Pi}_{\alpha}$ depends on the covariates and the model mean function only through the quantity $d^*$}. 
\btxt{Table C.2 of Web Appendix C shows the empirical values of $\widetilde{\Pi}_{\alpha}$ at 5\% level of significance, over different values of tuning parameter $\alpha$ and for various values of $d^*$ (which may be obtained based on suitable covariate values and $d$), with $\sigma = 1$}. 
Clearly, for fixed $\alpha$, the power increases with $d^*$; in fact, $d^*=0$ gives the level of the test. But, for fixed $d^* > 0$, the power of our proposed test decreases with increasing values of $\alpha$, which is quite natural, following the pattern of the ARE of the MDPDE used to construct these tests. Additionally, the loss in power is not quite significant compared to the classical MLE-based test (at $\alpha=0$) for small positive values of $\alpha$ (and also for larger $\alpha$ values when $d^*$ is large).

\section{Application to the Michaelis-Menten (MM) Model}\label{Application}
Let us now apply the proposed robust inference procedure to the popular MM model, as defined by Equation \eqref{mmeqn2} in Section \ref{intro}. To be consistent with our notation so far, let us denote the observed reaction velocities as $y_i$ and the substrate concentrations as $x_i$ for the $i$-th observation, and rewrite the MM model equation as
\begin{equation}\label{mmnlr}
    y_i = \mu(x_i, \bm{\beta}) + \epsilon_i,\ i = 1,2,\hdots,n,
\end{equation}
where $\mu(x_i, \bm{\beta}) = \frac{\beta_1 x_i}{\beta_2 + x_i}$ with $\bm{\beta}=(\beta_1,\beta_2)^T=(V_{max}, K_m)^T$ being the regression parameter vector and the random errors $\epsilon_i\sim\mathcal{N}(0, \sigma^2)$, independently  for all $i=1,\hdots,n$. As the mean function here is not defined at $\beta_2 = -x_i$, for any $i$, so we must take the parameter space for $\beta_2$ to be a subset of $\mathbb{R}\setminus\{-x_1,\hdots,-x_n\}$; in applications, we generally have $\beta_2>0$, which ensures that the corresponding parameter space ($\mathbb{R}^+$) is independent of the values of the covariate, and hence the resulting mean function always remains well-defined. The parameter space for $\beta_1$ may be taken as $\mathbb{R}$. For the MM model \eqref{mmnlr}, let us denote $\widehat{\bm{\theta}}_\alpha = (\widehat{\bm{\beta}}_\alpha^T, \widehat{\sigma}_\alpha^2)^T = (\widehat{\beta}_{1\alpha}, \widehat{\beta}_{2\alpha}, \widehat{\sigma}_\alpha^2)^T$ to be the MDPDE of $\bm{\theta} = \left(\bm{\beta}^T,\sigma^2\right)^T$ corresponding to a tuning parameter $\alpha\geq 0$, defined as the minimizer of the objective function in \eqref{Hnlr} with the particular model specific form of $\mu(x_i, \bm{\beta}) = \frac{\beta_1 x_i}{\beta_2 + x_i}$. The corresponding estimating equations are given by \eqref{esteq1} and \eqref{esteq2}, where we now have $\nabla_{\bm{\beta}}\mu(\bm{x}_i, \bm{\beta}) = \left(\frac{x_i}{\beta_2 + x_i},\ -\frac{\beta_1 x_i}{(\beta_2 + x_i)^2}\right)^T$.

\subsection{Asymptotic properties}
In order to find the asymptotic distributions of the MDPDEs of the parameters of the MM model, we follow the notation of Section \ref{APNLR} to get $\bm{\mu}(\bm{\beta}) = \left( \frac{\beta_1 x_1}{\beta_2 + x_1},\hdots,\frac{\beta_1 x_n}{\beta_2 + x_n}\right)^T$ and hence
\begin{equation}\label{mu.}
    \bm{\dot{\mu}}(\bm{\beta}) =
    \begin{bmatrix}
        \frac{x_1}{\beta_2 + x_1} & -\frac{\beta_1 x_1}{(\beta_2 + x_1)^2}\\
        \vdots & \vdots\\
        \frac{x_n}{\beta_2 + x_n} & -\frac{\beta_1 x_n}{(\beta_2 + x_n)^2}
    \end{bmatrix}.
\end{equation}
Then, under the general Conditions $(R1)$ and $(R2)$, it follows from Theorem \ref{th-2.1} that 
\begin{equation*}
    \left(\bm{\dot{\mu}}(\bm{\beta})^T\bm{\dot{\mu}}(\bm{\beta})\right)^{1/2} \left(\widehat{\bm{\beta}}_\alpha - \bm{\beta}\right) \xrightarrow{\mathcal{D}} \mathcal{N}_2\left(\bm{0}, v_{1\alpha}\sigma^2 \mathbb{I}_2 \right),
\end{equation*}
where
\begin{equation}
\label{F.TF.}
    \bm{\dot{\mu}}(\bm{\beta})^T\bm{\dot{\mu}}(\bm{\beta}) = 
    \begin{bmatrix}
        \sum_{i=1}^n a_i^2 & -\sum_{i=1}^n a_i b_i\\
        -\sum_{i=1}^n  a_i b_i & \sum_{i=1}^n b_i^2
    \end{bmatrix},
\end{equation}
with $a_i = \frac{x_i}{\beta_2 + x_i}$, $b_i = \frac{\beta_1 x_i}{(\beta_2 + x_i)^2}$. Upon simplification, we get
\begin{equation*}
    \sqrt{n}\left(\widehat{\beta}_{1\alpha} - \beta_1\right) \xrightarrow{\mathcal{D}}\mathcal{N}\left(0, \sigma_{\beta_1}^2\right) \mbox{ and } \sqrt{n}\left(\widehat{\beta}_{2\alpha} - \beta_2\right) \xrightarrow{\mathcal{D}}\mathcal{N}\left(0, \sigma_{\beta_2}^2\right),
\end{equation*}
where $\sigma_{\beta_1}^2 = \lim\limits_{n\to\infty} \frac{nv_{1\alpha}\sigma^2\sum_{i=1}^nb_i^2}{\sum_{i=1}^n a_i^2\sum_{i=1}^n b_i^2 - \left(\sum_{i=1}^n a_i b_i\right)^2}$ 
and $\sigma_{\beta_2}^2 = \lim\limits_{n\to\infty} \frac{nv_{1\alpha}\sigma^2 \sum_{i=1}^n a_i^2}{\sum_{i=1}^n a_i^2\sum_{i=1}^n b_i^2 - \left(\sum_{i=1}^n a_i b_i\right)^2}$, but they are not independent in general.

In order to help verification of the required Conditions (R1)-(R2) for a given MM model, we have proved the following propositions; see Web Appendix B for their proofs.
\begin{prop}\label{th-R1}
    Under the MM model given in \eqref{mmnlr}, if the parameter space for $\beta_2$ is $[0,\infty)$ and $x_i>0$ for all $i$, then Condition (R1) is satisfied.
\end{prop}
\begin{prop}\label{th-R2}
    In order to satisfy Condition \eqref{R21} in (R2) for the MM model given by \eqref{mmnlr}, it is necessary that the values $x_i$'s are not all same, $n\geq 2$, and $\beta_1 \neq 0$. 
\end{prop}
Note that, Proposition \ref{th-R1} provides a set of easily verifiable sufficient conditions for (R1) which are relatively stricter and not necessary for (R1). It may be possible to verify (R1) directly even if these sufficient conditions are not satisfied (although they are assumed to be true in most applications of the MM model). On the other hand, Proposition \ref{th-R2} provides a set of conditions that are necessary for (R2) to hold.

\subsection{Influence functions}\label{IFMM}
We will now study the IFs of the MDPDEs of the parameters of the MM model following the general theory developed in Section \ref{IFNLR}. Using \eqref{IFbeta} and \eqref{F.TF.}, the expressions of IFs of minimum DPD functionals $T_\alpha^{\beta_1}$ and $T_\alpha^{\beta_2}$ corresponding to the parameters $\beta_1$ and $\beta_2$ of the MM model \eqref{mmnlr}, respectively, simplify to 
\begin{equation} \label{IFth1mm}
    IF_{i_0}(t_{i_0}, T_\alpha^{\beta_1}, \bm{F_\theta}) = 
    \frac{(1+\alpha)^\frac{3}{2} (t_{i_0}-\mu(x_{i_0}, \bm{\beta})) e^{-\frac{\alpha}{2\sigma^2}(t_{i_0}-\mu(x_{i_0}, \bm{\beta}))^2}}{\det\left(\bm{\dot{\mu}}(\bm{\beta})^T\bm{\dot{\mu}}(\bm{\beta})\right)}\frac{\beta_1^2 x_{i_0}}{(\beta_2 + x_{i_0})^2}\sum_{j=1}^n\frac{x_j^2(x_{i_0}-x_j)}{(\beta_2+x_j)^4},
\end{equation}
\begin{equation} \label{IFth2mm}
    IF_{i_0}(t_{i_0}, T_\alpha^{\beta_2}, \bm{F_\theta}) = \frac{(1+\alpha)^\frac{3}{2} (t_{i_0}-\mu(x_{i_0}, \bm{\beta})) e^{-\frac{\alpha}{2\sigma^2}(t_{i_0}-\mu(x_{i_0}, \bm{\beta}))^2}}{\det\left(\bm{\dot{\mu}}(\bm{\beta})^T\bm{\dot{\mu}}(\bm{\beta})\right)}\frac{\beta_1 x_{i_0}}{(\beta_2 + x_{i_0})^2}\sum_{j=1}^n\frac{x_j^2(x_{i_0}-x_j)}{(\beta_2+x_j)^3},
\end{equation}
when there is contamination only in $i_0$-th observation at the contamination point $t_{i_0}$. \btxt{Further, the corresponding IF of the minimum DPD functional $T_\alpha^{\sigma^2}$, denoted by $IF_{i_0}(t_{i_0}, T_\alpha^{\sigma^2}, \bm{F_\theta})$, remains the same as given in the expression \eqref{IFsigma}, with $\mu(\bm{x}_{i_0}, \bm{\beta}) = \beta_1 x_{i_0}/(\beta_2 + x_{i_0})$.} For infinitesimal contamination in all observations at the contamination points in $\bm{t}=(t_1, \cdots, t_n)^T$, these IFs are given by
\begin{equation} \label{IFadmm}
    IF(\bm{t}, T_\alpha^{\beta_k}, \bm{F_\theta}) = \sum_{i=1}^n IF_i(t_i, T_\alpha^{\beta_k}, \bm{F_\theta}),\ \ k=1,2, \btxt{\mbox{ and } IF(\bm{t}, T_\alpha^{\sigma^2}, \bm{F_\theta}) = \sum_{i=1}^n IF_i(t_i, T_\alpha^{\sigma^2}, \bm{F_\theta})}.
\end{equation}

\begin{figure}[!b]
     \centering
     \begin{subfigure}[b]{0.43\textwidth}
         \centering
         \includegraphics[width=\textwidth]{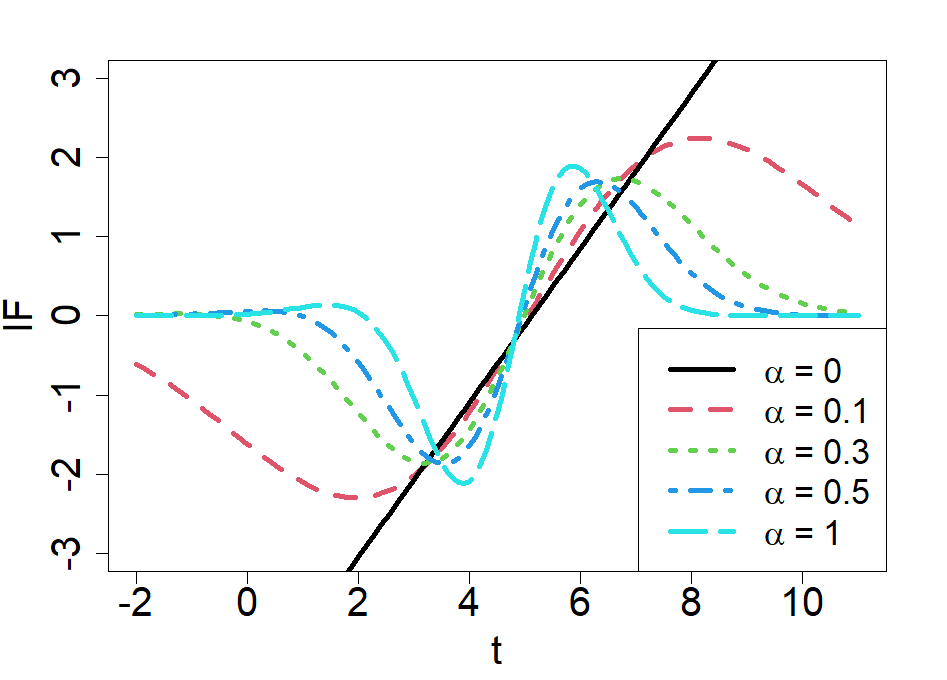}
         \caption{$IF(\bm{t}^*, T_\alpha^{\beta_1}, \bm{F_\theta})$, under Setup 1}
         \label{fig:IF-ad-S1-th1}
     \end{subfigure}
     \hfill
     \begin{subfigure}[b]{0.43\textwidth}
         \centering
         \includegraphics[width=\textwidth]{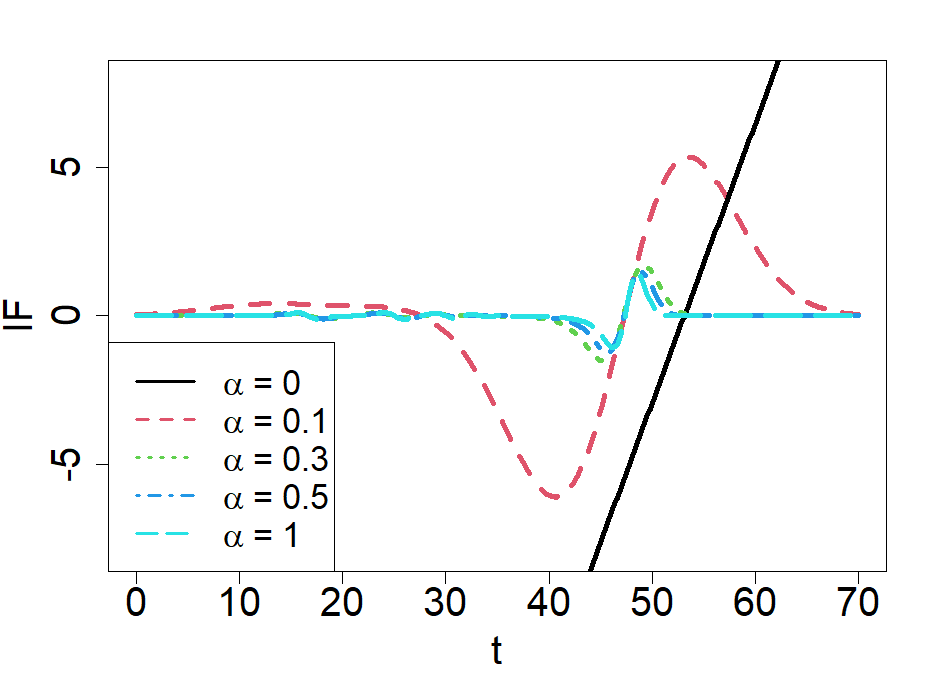}
         \caption{$IF(\bm{t}^*, T_\alpha^{\beta_1}, \bm{F_\theta})$, under Setup 2}
         \label{fig:IF-ad-S2-th1}
     \end{subfigure}
     \hfill
     \begin{subfigure}[b]{0.43\textwidth}
         \centering
         \includegraphics[width=\textwidth]{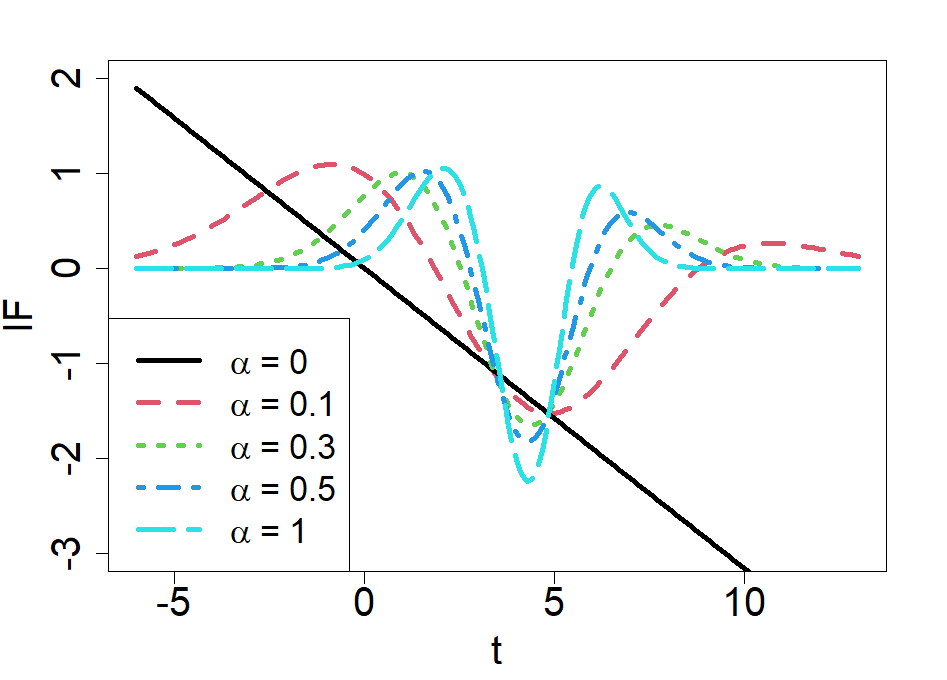}
         \caption{$IF(\bm{t}^*, T_\alpha^{\beta_2}, \bm{F_\theta})$, under Setup 1}
         \label{fig:IF-ad-S1-th2}
     \end{subfigure}
     \hfill
     \begin{subfigure}[b]{0.43\textwidth}
         \centering
         \includegraphics[width=\textwidth]{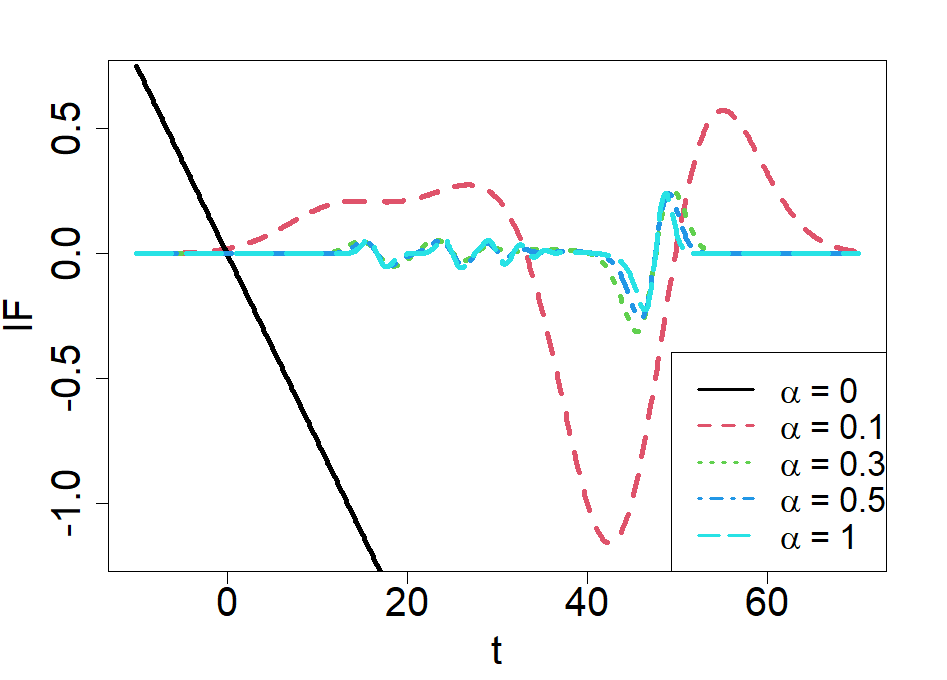}
         \caption{$IF(\bm{t}^*, T_\alpha^{\beta_2}, \bm{F_\theta})$, under Setup 2}
         \label{fig:IF-ad-S2-th2}
     \end{subfigure}
     \hfill
     \begin{subfigure}[b]{0.43\textwidth}
         \centering
         \includegraphics[width=\textwidth]{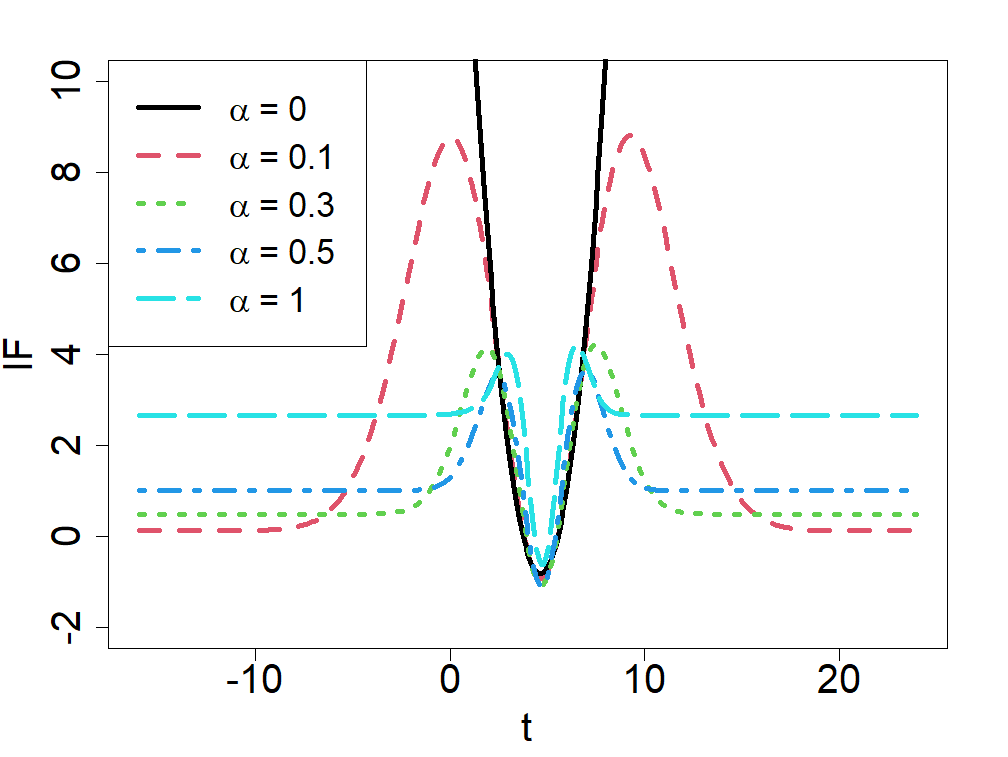}
         \caption{$IF(\bm{t}^*, T_\alpha^{\sigma^2}, \bm{F_\theta})$, under Setup 1}
         \label{fig:IF-ad-S1-sig}
     \end{subfigure}
     \hfill
     \begin{subfigure}[b]{0.43\textwidth}
         \centering
         \includegraphics[width=\textwidth]{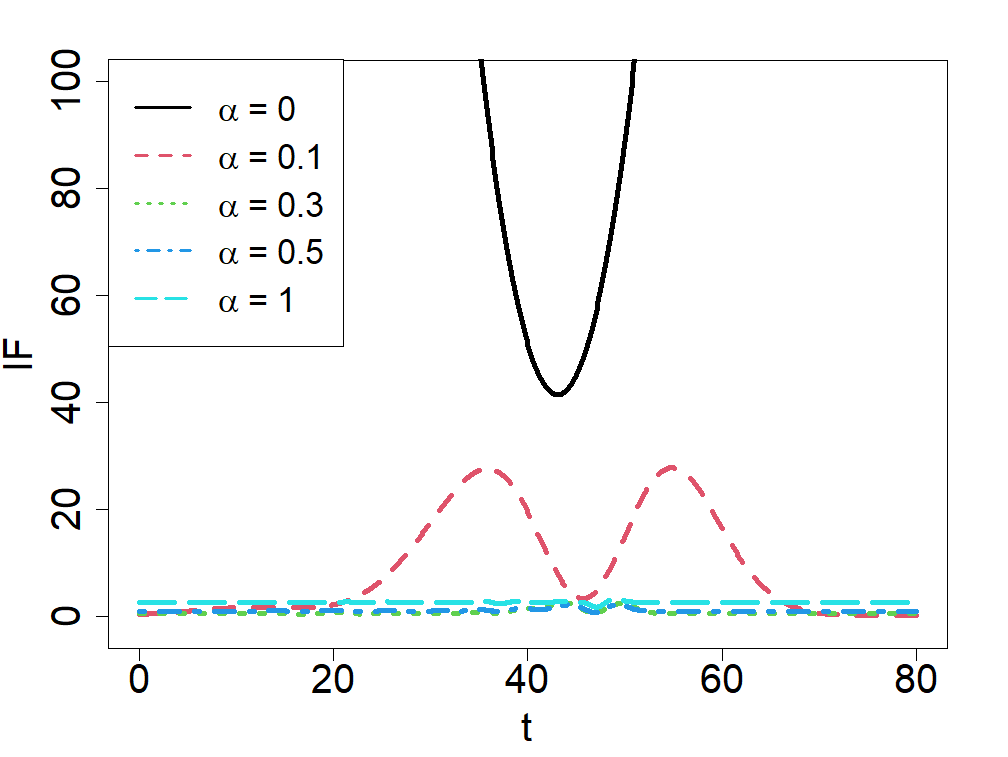}
         \caption{$IF(\bm{t}^*, T_\alpha^{\sigma^2}, \bm{F_\theta})$, under Setup 2}
         \label{fig:IF-ad-S2-sig}
     \end{subfigure}
     \caption{\btxt{IFs of the MDPDEs of the parameters $\beta_1$, $\beta_2$, and $\sigma^2$ in the MM model, with contamination in all observations, at the same contamination point, i.e., for $\bm{t}^* = (t,\ldots,t)^T$.}}
     \label{fig:IF-S12}
\end{figure}

For illustrations, we have plotted these IFs of the MDPDEs under the MM model, by generating data from the following two particular setups:
\begin{itemize}
	\item \textbf{Setup 1}: $n = 50,\ (\beta_1, \beta_2, \sigma) = (5,1,1),$ and $x_i = i$, for all $i=1,2,\hdots,n$.
	\item \textbf{Setup 2}: $n = 50,\ (\beta_1, \beta_2, \sigma) = (50,2,2),$ and $(x_1,x_2,\hdots,x_n)$ is a random sample of size $n$ from Uniform$(0,40)$.
\end{itemize}

The resulting IFs for contaminations in all observations are presented in the \btxt{Figure \ref{fig:IF-S12}} for both Setups 1-2, where the contamination points are taken as $\bm{t}^*=(t,\hdots,t)$; those obtained for contamination in particular observation are provided in \btxt{Figures C.1-C.2} of Web Appendix C. From all these figures, it is clear that the IFs of the MDPDEs under the MM model are bounded for $\alpha>0$ and are redecending in nature for increasing values of $\alpha$. For $\alpha=0$, the IFs are unbounded as expected, since the MDPDE with $\alpha=0$ is the same as the non-robust MLE also for the MM models.

\subsection{Finite sample performances of the MDPDE}\label{Simulation Studies}
We now present an extensive simulation study to justify the improved performances of the MDPDEs of parameters of the MM model \eqref{mmnlr}. For this purpose, we have simulated data following the two setups (Setup 1 and Setup 2) as mentioned in the Section \ref{IFMM}. \btxt{Note that the choice of true parameter values for Setup 1, i.e., $\bm{\beta}=(5,1)$, is made directly following Tabatabai et al.~\cite{tabatabai2014new}. The parameter values of the second setup are chosen to be moderately large, namely (50, 2), to closely reproduce real-life examples following the discussions in Marasovic et al.~\cite{marasovic2017robust}}.  Additionally, we have incorporated data contamination by multiplying the observed values of response or covariate (or both) by a fixed number, say $C$, in a certain proportion (say $e_c$) of cases selected randomly from the full sample. We report the results for the choice of $C=5$ under Setup 1 and for $C=2$ under Setup 2. \btxt{In each case, we have used the \texttt{R} function \texttt{optim}\footnote{\btxt{\url{https://stat.ethz.ch/R-manual/R-devel/library/stats/html/optim.html}}} (with the BFGS method) to compute the MDPDEs by minimizing the objective function in \eqref{Hnlr}, where the initial values of the parameters are chosen to be their true values; this helped us to compare the actual behavior of the estimates, avoiding any numerical instability and the effects of initial values.}

For each setup, we present the absolute empirical bias (EBias) and the empirical mean square error (EMSE) of the MDPDE and different other robust estimators of the parameters of the MM model \eqref{mmnlr}, based on $1000$ replications. \btxt{For a generic model parameter $\theta$ (which is either of $\beta_1$, $\beta_2$ or $\sigma$ in our case) having true value $\theta_0$, the summary measures EBias and EMSE are computed as
\begin{equation*}
    EBias = \left|\frac{1}{1000} \sum_{r=1}^{1000} \widehat{\theta}^{(r)} - \theta_0\right|, \mbox{ and } EMSE = \frac{1}{1000} \sum_{r=1}^{1000} \left(\widehat{\theta}^{(r)} - \theta_0\right)^2,
\end{equation*}
where $\widehat{\theta}^{(1)},\ldots, \widehat{\theta}^{(1000)}$ denote the estimated values of the parameter $\theta$ across 1000 replications.} We have repeated this exercise for pure data under both the setups, as well as contaminated data with different choices of the contamination proportion $e_c$.  The results for $e_c=0\%$ (pure data), $e_c=10\%$ (light contamination) and $e_c=40\%$ (heavy contamination) are reported in Table \ref{tab:S2BM} (for Setup 2) and \btxt{Table C.4} of Web Appendix C (for Setup 1); the results obtained for $e_c = 20\%$ and $30\%$ (moderate contamination) are reported in \btxt{Tables C.5-C.6} of Web Appendix C for both setups.

In all these simulations, we have compared the performances of the MDPDEs at different values of $\alpha$ with the OLS (which is the same as the MDPDE at $\alpha=0$), the MOM estimator, the KPS estimator with various choices of its tuning parameter $w$, and M-estimators based on the Huber’s and the Tukey’s loss/weight functions (described in Web Appendix A.2). 
\btxt{These competing estimators are computed following the existing algorithms provided in their respective reference papers, and a comparison of the runtime complexity of these estimates with our proposed MDPDEs is provided in Web Appendix C.3 (for the simulation Setup 2).}

In addition, we have also computed the average prediction errors for the model fitted with different estimators, which we define based only on the clean observations as given by
\begin{equation*}
    \frac{1}{1000}\sum_{r=1}^{1000}\frac{1}{|\mathcal{A}_r|}\sum_{i\in\mathcal{A}_r}\left(y_i - \mu\left(\widehat{\bm{\beta}}^{(r)}, x_i\right)\right)^2,
\end{equation*}
where $\mathcal{A}_r = \{i:$ $i$-th data-point is not contaminated in the dataset generated at $r$-th replication, $i=1,2,\ldots,n\}$ and $\widehat{\bm{\beta}}^{(r)}$ be the parameter estimate obtained at $r$-th replication for a particular method. \btxt{This measures the expected prediction accuracy obtained from a single dataset using any of the robust estimators.}

\begin{sidewaystable}
    \centering
    \caption{absolute empirical bias and empirical MSE obtained by different estimators in simulations with Setup 2 and different amounts of contamination (the minimum values obtained for each parameter are highlighted in bold font)}
    \resizebox{\textwidth}{!}
    {\begin{tabular}{|l|rrr|rrr|rrr|rrr|rrr|}
        \hline
        Outlier Direction & \multicolumn{3}{c|}{No outlier (pure data)} & \multicolumn{6}{c|}{Contamination only in response} & \multicolumn{6}{c|}{Contamination in both response and covariate} \\
        Outlier proportion & ~ & ~&~ & \multicolumn{3}{c|}{10\%} & \multicolumn{3}{c|}{40\%} & \multicolumn{3}{c|}{10\%} & \multicolumn{3}{c|}{40\%} \\ \hline
        ~ & $\beta_1$  & $\beta_2$  & $\sigma$  & $\beta_1$ & $\beta_2$  & $\sigma$ & $\beta_1$ & $\beta_2$ & $\sigma$ & $\beta_1$ & $\beta_2$  & $\sigma$ & $\beta_1$ & $\beta_2$  & $\sigma$ \\ \hline \hline
        \multicolumn{16}{|l|}{\textbf{\underline{Empirical Bias}}} \\
        \multicolumn{16}{|l|}{} \\
        OLS & \textbf{0.0053} & \textbf{0.0021} & ~ & 5.3409 & 0.1399 & ~ & 20.73 & 0.2362 & ~ & 11.99 & 2.7578 & ~ & 48.00 & 11.62 & ~ \\
        MOM & 0.0227 & 0.0089 & ~ & 1.3296 & 0.0845 & ~ & 22.67 & 0.7303 & ~ & 2.8428 & 0.6894 & ~ & 55.55 & 12.25 & ~ \\
        KPS(w=0.5) & 0.0154 & 0.0047 & ~ & 0.0061 & 0.0055 & ~ & 0.5786 & 0.5918 & ~ & 0.0221 & 0.0128 & ~ & 0.3363 & 0.1575 & ~ \\
        KPS(w=1) & 0.0359 & 0.0096 & ~ & 0.0225 & 0.0042 & ~ & 0.2046 & 0.0864 & ~ & 0.0114 & 0.0001 & ~ & 0.2167 & 0.1067 & ~ \\
        KPS(w=1.5) & 0.0637 & 0.0176 & ~ & 0.0455 & 0.0107 & ~ & 0.0415 & 0.0197 & ~ & 0.0404 & 0.0092 & ~ & 0.1108 & 0.0518 & ~ \\
        M-Est (Huber) & 0.0403 & 0.0147 & ~ & 0.3804 & 0.0233 & ~ & 8.4130 & 0.1436 & ~ & 0.8063 & 0.1190 & ~ & 37.64 & 6.9652 & ~ \\
        M-Est (Tukey) & 0.0423 & 0.01525 & ~ & 0.0364 & 0.0103 & ~ & 4.7816 & 0.0202 & ~ & 0.0197 & 0.0031 & ~ & 37.61 & 7.0127 & ~ \\ \hline
        MDPDE($\alpha$=0.05) & 0.0056 & 0.0023 & 0.0521 & 0.0139 & 0.0684 & 0.5923 & 14.09 & 0.0759 & 13.80 & 0.4244 & 0.0106 & 0.9497 & 41.15 & 9.9069 & 13.49 \\
        MDPDE($\alpha$=0.1) & 0.0060 & 0.0025 & \textbf{0.0515} & 0.0168 & 0.0089 & 0.0212 & 2.4492 & 0.0940 & 2.9793 & 0.0429 & 0.0218 & \textbf{0.0012} & 11.45 & 2.5475 & 4.4030 \\
        MDPDE($\alpha$=0.2) & 0.0075 & 0.0030 & 0.0534 & \textbf{0.0001} & 0.0019 & 0.0268 & 0.0468 & 0.0288 & \textbf{0.1563} & 0.0166 & 0.0094 & 0.0192 & 0.1417 & 0.0779 & \textbf{0.2318} \\
        MDPDE($\alpha$=0.3) & 0.0094 & 0.0036 & 0.0583 & 0.0048 & \textbf{0.0001} & 0.0206 & 0.0240 & 0.0144 & 0.2089 & 0.0084 & 0.0059 & 0.0158 & 0.1111 & 0.0538 & 0.2487 \\
        MDPDE($\alpha$=0.5) & 0.0144 & 0.0048 & 0.0733 & 0.0089 & 0.0012 & \textbf{0.0049} & 0.0043 & 0.0063 & 0.4106 & \textbf{0.0003} & 0.0028 & 0.0026 & 0.0706 & 0.0356 & 0.4304 \\
        MDPDE($\alpha$=0.7) & 0.0186 & 0.0059 & 0.0914 & 0.0127 & 0.0023 & 0.0097 & \textbf{0.0005} & 0.0043 & 0.6268 & 0.0057 & \textbf{0.0007} & 0.0107 & 0.0566 & 0.0292 & 0.6364 \\
        MDPDE($\alpha$=1) & 0.0270 & 0.0081 & 0.1198 & 0.0207 & 0.0051 & 0.0257 & 0.0038 & \textbf{0.0030} & 0.9168 & 0.0160 & 0.0030 & 0.0257 & \textbf{0.0453} & \textbf{0.0241} & 0.9153 \\ \hline \hline
        \multicolumn{16}{|l|}{\textbf{Empirical MSE}} \\
        \multicolumn{16}{|l|}{} \\
        OLS & \textbf{0.2960} & \textbf{0.0243} & ~ & 36.82 & 0.6774 & ~ & 454.68 & 1.3956 & ~ & 185.17 & 14.72 & ~ & 2515.37 & 180.11 & ~ \\
        MOM & 0.6553 & 0.0534 & ~ & 6.1540 & 0.2222 & ~ & 648.93 & 7.1888 & ~ & 27.66 & 2.6532 & ~ & 3510.54 & 248.81 & ~ \\
        KPS(w=0.5) & 0.3349 & 0.0271 & ~ & 0.3595 & 0.0290 & ~ & 2.0150 & 0.7564 & ~ & 0.3629 & 0.0298 & ~ & 17.01 & 1.6779 & ~ \\
        KPS(w=1) & 0.4660 & 0.0373 & ~ & 0.4538 & 0.0367 & ~ & 0.8494 & 0.0884 & ~ & 0.4522 & 0.0368 & ~ & 0.7110 & 0.0748 & ~ \\
        KPS(w=1.5) & 0.7114 & 0.0554 & ~ & 0.6407 & 0.0522 & ~ & 0.6520 & 0.0516 & ~ & 0.6290 & 0.0512 & ~ & 0.6813 & 0.0600 & ~ \\
        M-Est (Huber) & 0.3165 & 0.0267 & ~ & 0.5872 & 0.0404 & ~ & 146.57 & 1.6212 & ~ & 1.1252 & 0.0586 & ~ & 1468.90 & 52.26 & ~ \\
        M-Est (Tukey) & 0.3173 & 0.0267 & ~ & 0.3406 & 0.0284 & ~ & 159.54 & 1.8608 & ~ & 0.3476 & 0.0293 & ~ & 1478.86 & 53.08 & ~ \\ \hline
        MDPDE($\alpha$=0.05) & 0.2964 & 0.0245 & \textbf{0.0445} & 0.5847 & 0.0678 & 4.7267 & 294.64 & 0.8267 & 268.86 & 2.8974 & 0.0959 & 7.0403 & 2148.77 & 153.35 & 212.34 \\
        MDPDE($\alpha$=0.1) & 0.2986 & 0.0248 & 0.0450 & \textbf{0.3369} & 0.0284 & 0.0548 & 45.49 & 0.1394 & 56.41 & \textbf{0.3460} & 0.0304 & 0.0603 & 561.17 & 36.60 & 69.36 \\
        MDPDE($\alpha$=0.2) & 0.3073 & 0.0257 & 0.0474 & 0.3402 & \textbf{0.0284} & \textbf{0.0516} & 0.6943 & 0.0554 & 0.4675 & 0.3471 & \textbf{0.0293} & \textbf{0.0529} & 2.0103 & 0.1239 & 0.4676 \\
        MDPDE($\alpha$=0.3) & 0.3206 & 0.0270 & 0.0512 & 0.3535 & 0.0297 & 0.0550 & \textbf{0.5938} & 0.0508 & \textbf{0.1592} & 0.3579 & 0.0302 & 0.0555 & \textbf{0.6463} & 0.0630 & \textbf{0.1902} \\
        MDPDE($\alpha$=0.5) & 0.3572 & 0.0304 & 0.0615 & 0.3883 & 0.0329 & 0.0645 & 0.5941 & \textbf{0.0489} & 0.3043 & 0.3902 & 0.0333 & 0.0649 & 0.6364 & 0.0579 & 0.3276 \\
        MDPDE($\alpha$=0.7) & 0.4025 & 0.0346 & 0.0738 & 0.4286 & 0.0366 & 0.0746 & 0.6092 & 0.0501 & 0.5571 & 0.4292 & 0.0369 & 0.0748 & 0.6462 & 0.0576 & 0.5725 \\
        MDPDE($\alpha$=1) & 0.4791 & 0.0412 & 0.0920 & 0.4945 & 0.0429 & 0.0882 & 0.6272 & 0.0515 & 1.0455 & 0.4935 & 0.0430 & 0.0881 & 0.6574 & \textbf{0.0573} & 1.0431 \\ \hline
    \end{tabular}}
    \label{tab:S2BM}
\end{sidewaystable}

\begin{table}
    \centering
    \caption{Average prediction errors obtained by different estimators and different amounts of contamination (the minimum values obtained for each parameter are highlighted in bold font)}
    \resizebox{\textwidth}{!}
    {\begin{tabular}{|l||r|rr|rr||r|rr|rr|}
    \hline
    Setup & \multicolumn{5}{c||}{Setup 1} & \multicolumn{5}{c|}{Setup 2}  \\ \hline
    Outlier direction & none & \multicolumn{2}{c|}{response} & \multicolumn{2}{c||}{response + covariate} & none & \multicolumn{2}{c|}{response} & \multicolumn{2}{c|}{response + covariate} \\
    Outlier proportion  & 0\%    & 10\%   & 40\%   & 10\%   & 40\%   & 0\%    & 10\%   & 40\%   & 10\%   & 40\%   \\ \hline\hline
    OLS & \textbf{0.961} & 4.910 & 57.02 & 6.236 & 16.40  & \textbf{3.828}  & 23.91  & 296.80 & 32.98  & 245.91 \\
    MOM & 1.974	& 14.02 & 4504.9 & 2.922 & 18.50 & 4.033 & 6.618 & 341.71 & 10.20 & 388.62\\
    KPS(w=0.5) & 0.967  & 1.017  & 551.94 & 0.963  & 3.874  & 3.845  & 3.829  & 9.372  & 3.831  & 5.760  \\
    KPS(w=1) & 0.983  & 0.997  & 2.155  & 0.978  & 3.817  & 3.915  & 3.876  & 3.872  & 3.878  & 3.852 \\
    KPS(w=1.5) & 1.009  & 1.528  & 5.944  & 1.939  & 1.554  & 4.032  & 3.976  & 3.775  & 3.967  & 3.806 \\
    M-Est (Huber) & 0.963  & 1.003  & 15.67  & 1.048  & 42.94  & 3.838  & 3.998  & 78.60  & 4.055  & 244.95 \\ 
    M-Est (Tukey)  & 0.963  & 0.958  & 10.61  & 0.957  & 43.92  & 3.838  & 3.818  & 64.45  & \textbf{3.821}  & 243.90 \\ \hline 
    MDPDE($\alpha$=0.05) & \textbf{0.961} & 3.147 & 54.59 & 6.333 & 18.58 & 3.829 & 4.414 & 202.00 & 4.785 & 211.52 \\
    MDPDE($\alpha$=0.1) & 0.961 & 0.965 & 52.27 & 1.819 & 16.72 & 3.830 & \textbf{3.817} & 40.69 & 3.826 & 65.29 \\
    MDPDE($\alpha$=0.2) & 0.962 & \textbf{0.960} & 47.68 & \textbf{0.957} & 13.97 & 3.835 & 3.820 & 3.924 & \textbf{3.823} & 4.056 \\
    MDPDE($\alpha$=0.3) & 0.964 & 0.968 & 42.11 & 0.959 & 11.58 & 3.842 & 3.826 & 3.763 & 3.828 & 3.802 \\
    MDPDE($\alpha$=0.5) & 0.971 & 0.965 & 4.665 & 0.964 & 4.348 & 3.861 & 3.844 & \textbf{3.760} & 3.846 & \textbf{3.789} \\
    MDPDE($\alpha$=0.7) & 0.984 & 1.000 & \textbf{0.943} & 0.972 & 1.056 & 3.886 & 3.866 & 3.768 & 3.867 & 3.789 \\
    MDPDE($\alpha$=1) & 0.995 & 0.987 & 0.963 & 0.985 & \textbf{0.943} & 3.929 & 3.902 & 3.776 & 3.903 & 3.795 \\ \hline
    \end{tabular}}
    \label{tab:S12APE}
\end{table}

Based on all these numerical results, we can easily see that the MDPDEs with suitable values of $\alpha$ perform the best among all the competitors considering all the cases of simulation setups. For lower contamination proportions, a close competitor of our MDPDE is the KPS estimator and the Tukey's M estimator, but their performance deteriorates as contamination proportion increases. 
This deterioration is more significant when there are contamination in both response and covariate. Further, 
for increasing proportion of contamination, the average prediction error gets minimized for a larger value of the tuning parameter $\alpha$ (refer to Table \ref{tab:S12APE} and \btxt{Table C.3} in Web Appendix C).

We would like to particularly mention that the median of means (MOM) method fails to perform well for almost all the contamination schemes. So, we have not considered the MOM estimator for further real data applications presented in this paper.

\subsection{Wald-type tests for validity of the Michaelis-Menten model}\label{mmtesting}
In the context of our MM model \eqref{mmnlr}, we assume $x_i > 0$, for all $i=1,2,\hdots,n$, as $x_i$'s are the substrate concentrations in a reaction. Now, $\beta_2 = 0$ implies that the mean/regression function is simply $\mu(x_i, \bm{\beta}) = \beta_1$, a constant, which is independent of the covariate $x_i$; Consequently, $y_i$’s simply become IID observations from a normal distribution with constant mean $\beta_1$ and variance $\sigma^2$. So, to test the validity of the MM model assumption, it is enough to test
\begin{equation}\label{testofth2}
    H_0: \beta_2 = 0 \mbox{ against } H_1:\beta_2 > 0.
\end{equation}
For this test we can use the test statistic $\widetilde{W}_{n,\alpha}$ given in \eqref{WTTSforSP}, i.e.,

\begin{equation}\label{mmsigniftestst}
    \widetilde{W}_{n,\alpha} = \frac{\widehat{\beta}_{2\alpha}}{\sqrt{v_{1\alpha}\widehat{\sigma}_\alpha^2s_{22}}} \xrightarrow{\mathcal{D}} \mathcal{N}(0,1), \mbox{ under } H_0.
\end{equation}
Then, for testing $H_0$ against the one-sided alternative $H_1$ in \eqref{testofth2}, $H_0$ will be rejected at level of significance $\gamma$ if the observed value of $\widetilde{W}_{n,\alpha}$ exceeds $z_\gamma$, the $(1-\gamma)$-th quantile of $\mathcal{N}(0,1)$. Similarly, one can construct Wald-type tests for testing of the parameter $\beta_1$ against one-sided or both-sided alternatives.

We have conducted simulation studies to study the performance of the Wald-type test for the above-discussed testing problem. In particular, we have computed the empirical level (ELevel) and the empirical power (EPower) of the test for various values of the tuning parameter $\alpha$ based on 1000 replications of simulated data; the simulation schemes for computations of ELevel and EPower are taken as follows:

\noindent\textbf{For ELevel:} We have simulated samples as per Setup 2, but with $(\beta_1,\beta_2,\sigma) = (20,0,1)$ and different sample sizes ($n$). So, $y_i$'s are IID $\mathcal{N}(20,1)$, i.e., the MM model is not significant for these (pure) data. To introduce contamination, $e_c$ proportion of $y_i$'s are randomly selected and each $y_i$ is replaced by a random number between $z_i$ and $y_i$, where $z_i = \frac{20x_i}{5+x_i}$ is the true mean function of an MM model with parameters $\beta_1 = 20$, $\beta_2 = 5$. We have chosen such a contamination scheme, to make the (non-robust) estimates of $\beta_2$ likely to be larger than 0, so that a non-robust test will reject the null hypothesis more than the expected due to the effect of data contamination, although the null hypothesis is true for the actual generated samples. The ideal values of ELevel, measured as the proportion of rejections, should be close to the level of the test, which is taken to be $5\%$ in our experiments.

\noindent\textbf{For EPower:} We have again simulated data following Setup 2 but with $(\beta_1,\beta_2,\sigma) = (20,1,2)$ and different sample sizes ($n$). To introduce contamination, $e_c$ proportion of $y_i$'s are randomly selected and each $y_i$ is replaced by a random number generated from Uniform$(0,20)$. In this case, the null hypothesis is actually false, i.e., the MM model is indeed valid. So, empirical power, measured as the proportion of rejections, should be high (close or equal to 1) for a good stable test.\\

The resulting values of the empirical levels and powers are plotted against the sample sizes for various values of tuning parameter $\alpha$ and various proportions ($e_c$) of contamination, which are presented in Figures \ref{fig:ELevel} and \ref{fig:EPower}, respectively.

\begin{figure}[!h]
     \centering
     \begin{subfigure}[b]{0.43\textwidth}
         \centering
         \includegraphics[width=\textwidth]{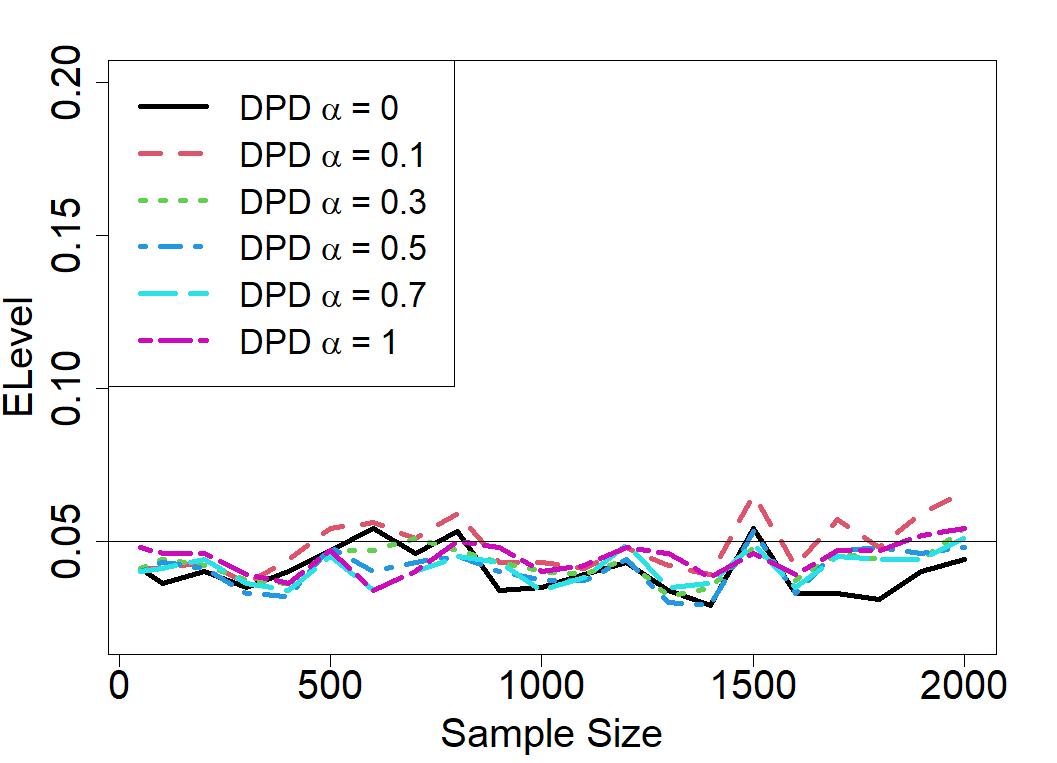}
         \caption{ELevel, 0\% contamination (pure data)}
         \label{fig:EL-0}
     \end{subfigure}
     \hfill
     \begin{subfigure}[b]{0.43\textwidth}
         \centering
         \includegraphics[width=\textwidth]{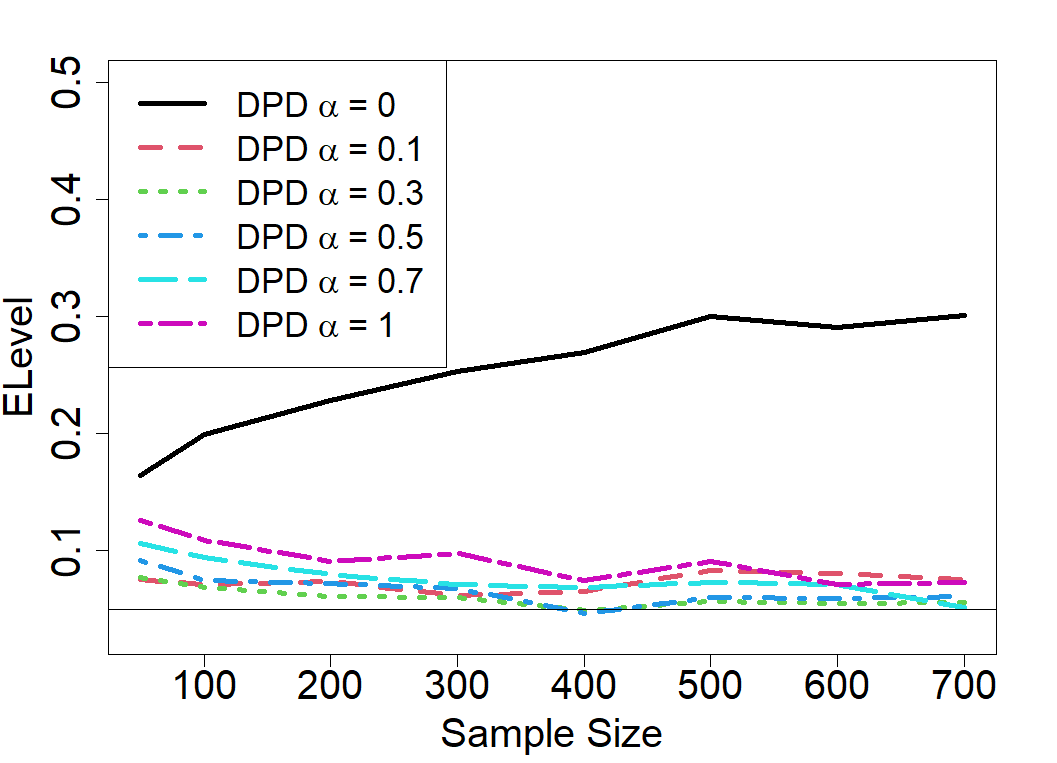}
         \caption{ELevel, 10\% contamination}
         \label{fig:EL-10}
     \end{subfigure}
     \hfill
     \begin{subfigure}[b]{0.43\textwidth}
         \centering
         \includegraphics[width=\textwidth]{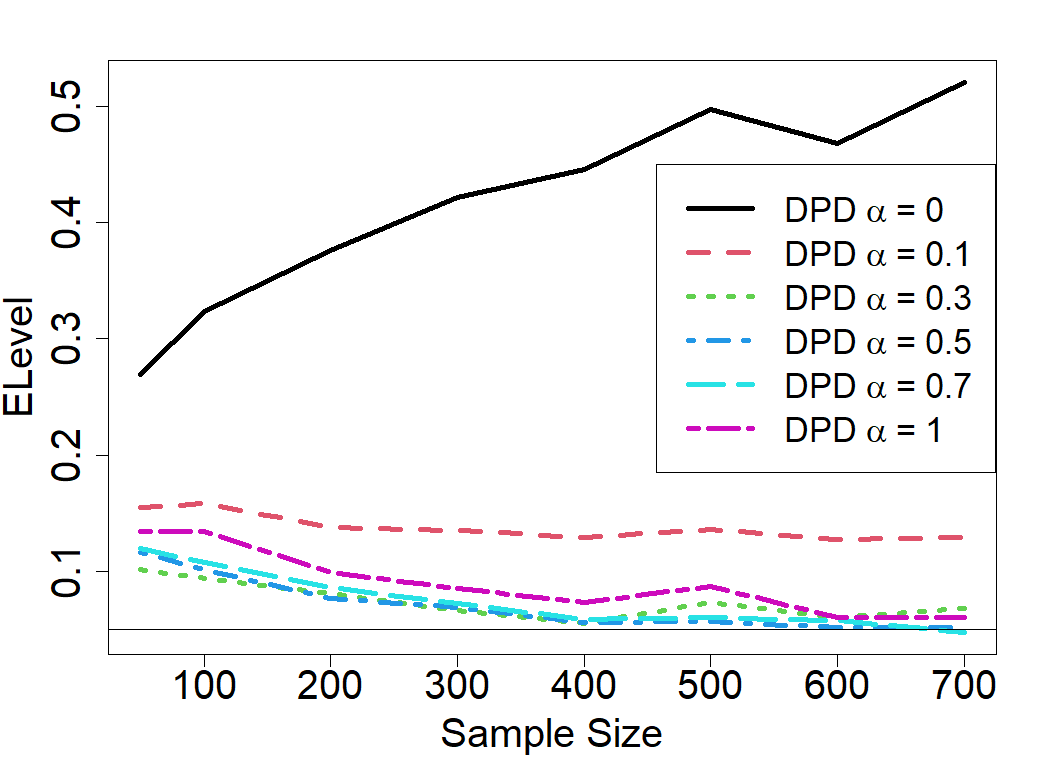}
         \caption{ELevel, 20\% contamination}
         \label{fig:EL-20}
     \end{subfigure}
     \hfill
     \begin{subfigure}[b]{0.43\textwidth}
         \centering
         \includegraphics[width=\textwidth]{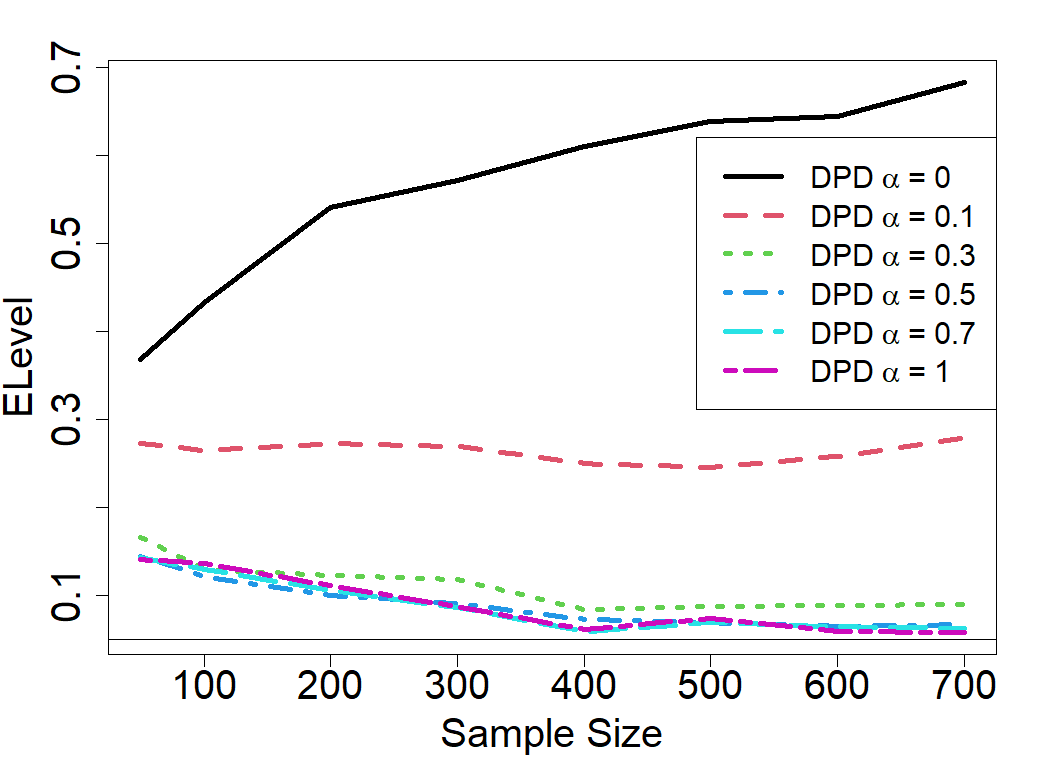}
         \caption{ELevel, 30\% contamination}
         \label{fig:EL-30}
     \end{subfigure}
     \caption{Empirical levels (ELevel) of the robust Wald-type tests}
     \label{fig:ELevel}
\end{figure}
\begin{figure}[h]
     \centering
     \begin{subfigure}[b]{0.43\textwidth}
         \centering
         \includegraphics[width=\textwidth]{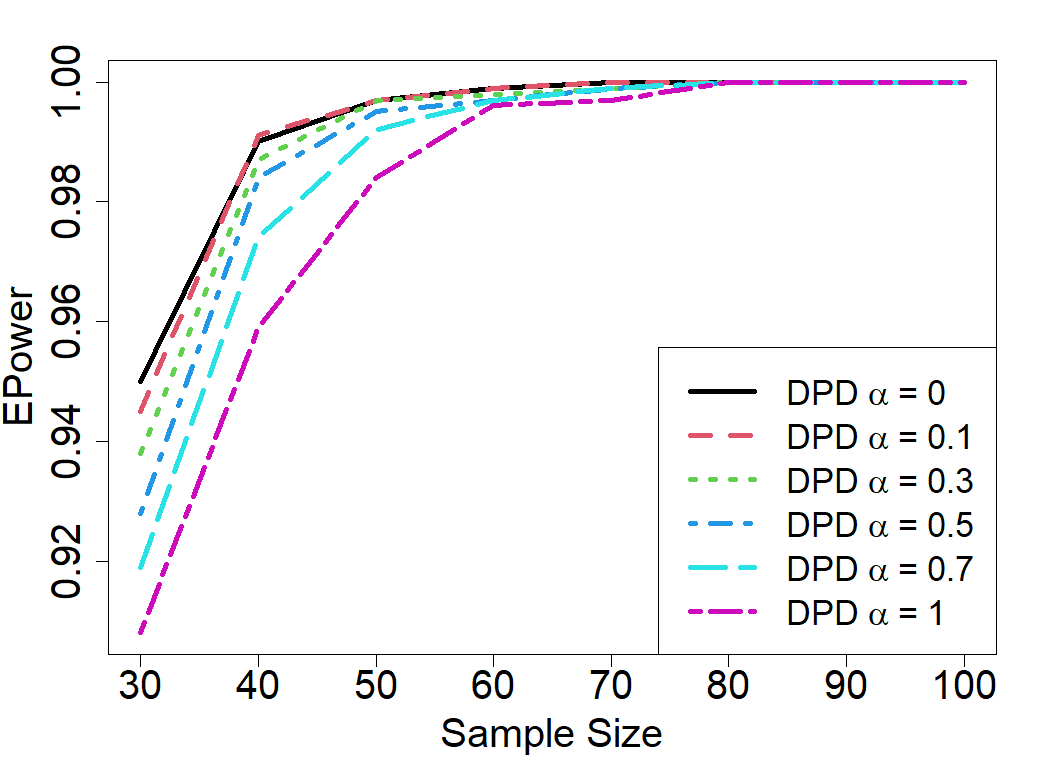}
         \caption{EPower, 0\% contamination (pure data)}
         \label{fig:EP-0}
     \end{subfigure}
     \hfill
     \begin{subfigure}[b]{0.43\textwidth}
         \centering
         \includegraphics[width=\textwidth]{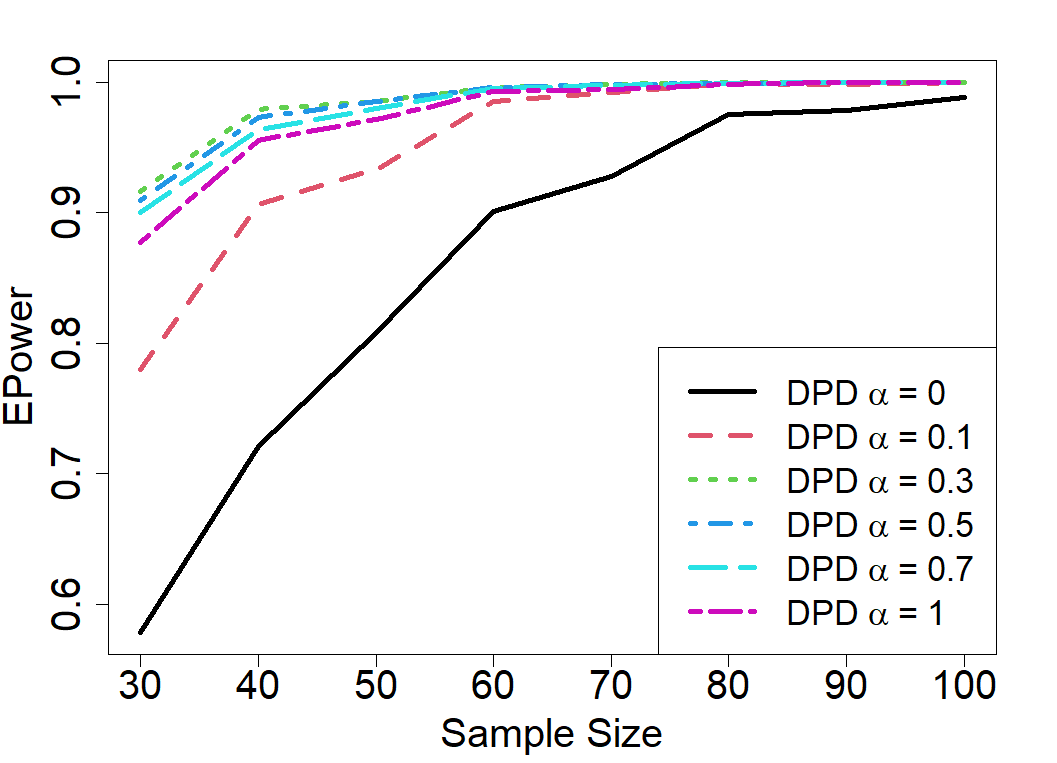}
         \caption{EPower, 10\% contamination}
         \label{fig:EP-10}
     \end{subfigure}
     \hfill
     \begin{subfigure}[b]{0.43\textwidth}
         \centering
         \includegraphics[width=\textwidth]{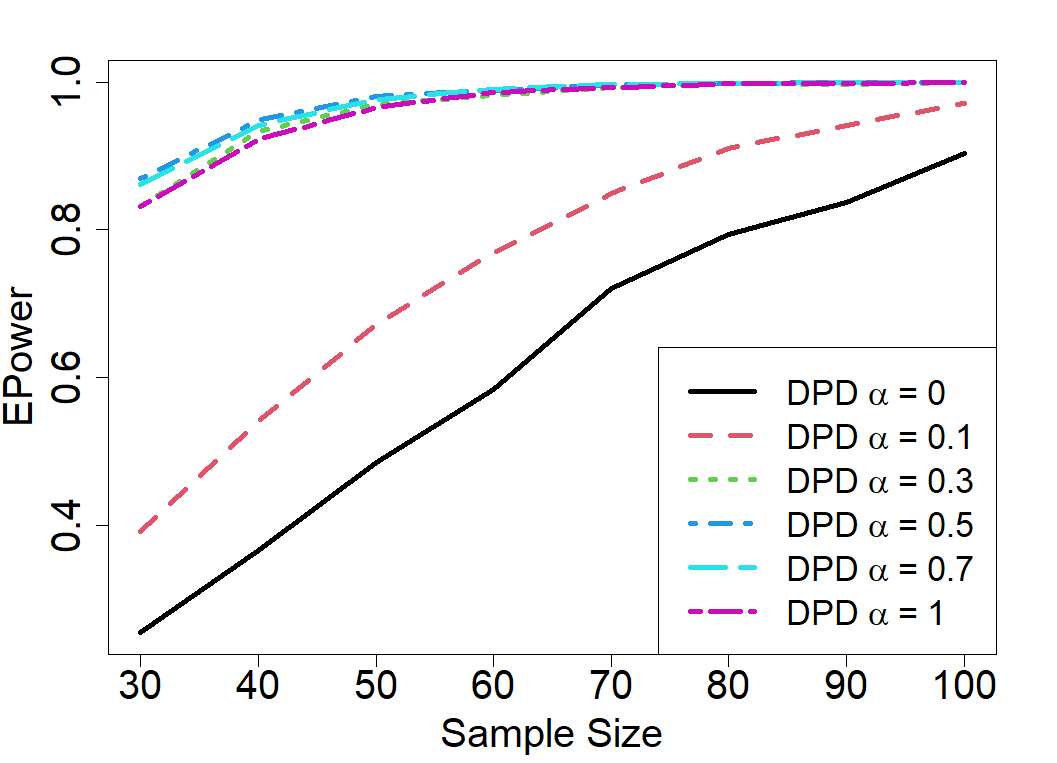}
         \caption{EPower, 20\% contamination}
         \label{fig:EP-20}
     \end{subfigure}
     \hfill
     \begin{subfigure}[b]{0.43\textwidth}
         \centering
         \includegraphics[width=\textwidth]{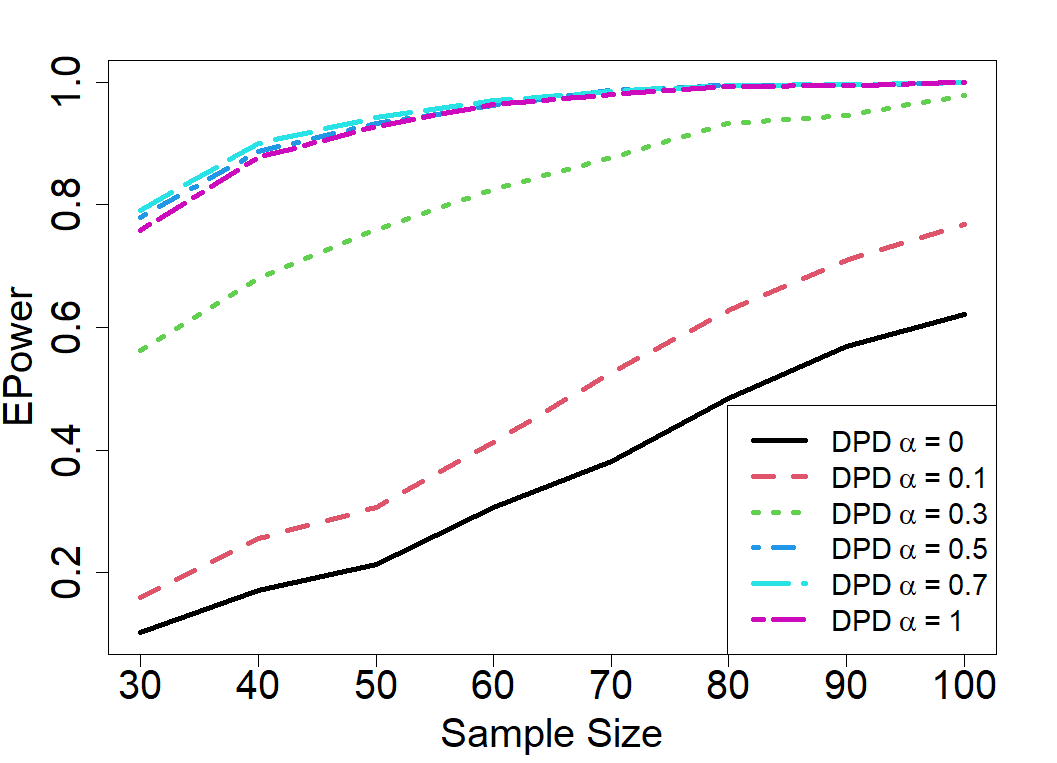}
         \caption{EPower, 30\% contamination}
         \label{fig:EP-30}
     \end{subfigure}
     \caption{Empirical powers (EPower) of the robust Wald-type tests}
     \label{fig:EPower}
\end{figure}

Figure \ref{fig:ELevel} shows that, under pure data, the empirical levels of the Wald test (at $\alpha=0$) and all the MDPDE-based Wald-type tests with $\alpha>0$ are quite close to each other and fluctuate closely around the desired $5\%$ level of significance. Further, for contaminated data, the empirical levels of the MDPDE-based Wald-type tests with moderate or large values of $\alpha$ $(\alpha\geq 0.3)$ tend to stabilize near $5\%$ level of significance, as sample size increases. But the empirical levels of the MLE-based Wald test fail to converge to the desired level under contamination due to its non-robust nature.

Figure \ref{fig:EPower} illustrates that, under pure data, the empirical powers of the Wald test and MDPDE-based Wald-type tests are all very close to 1 at a moderate sample size of $n=50$ or $60$, and clearly converge to 1 quite fast as the sample size increases. But for contaminated data, the empirical powers of the Wald tests decrease considerably with increase in the proportion of contamination, and the rate of consistency slows down significantly. However, the empirical powers of the Wald-type tests with $\alpha\geq 0.3$ decrease only a little at higher contamination proportion, but then they also regain their convergence rate to one very quickly as the sample size increases. These figures clearly illustrate significantly improved stability of our proposed MDPDE-based Wald-type tests with increasing values of $\alpha>0$ under possible data contamination.

\section{On the choice of optimum tuning parameter $\alpha$}\label{Optalpha}
\btxt{It is well studied that the tuning parameter $\alpha$ of the MDPDE acts as a trade-off between robustness and efficiency --- the robustness of the estimator under data contamination increases at a minimal cost of the loss in (asymptotic) efficiency under pure data as $\alpha$ increases. The same is also observed for the MDPDEs and the associated Wald-type tests under our NLR model \eqref{nlr} across all the simulation results discussed in the previous sections. So, while applying these DPD-based estimation and testing methods to a real dataset, it is very important to choose an optimal value of $\alpha$ providing the best possible `compromise' between robustness and efficiency, depending on the (unknown) amount of contamination in the given dataset.}

\btxt{Under the setup of IID data, Warwick and Jones \cite{warwick2005choosing} suggested using the optimal value of $\alpha$ that minimizes an asymptotic approximation of the MSE of the MDPDEs. Ghosh and Basu \cite{ghosh2015robust} extended this idea for the general INH setup, with detailed illustrations for the linear regression model. Following the same idea, an optimal value of $\alpha$ for our NLR model can be obtained by minimizing an estimate of the quantity}

\begin{equation*}\color{black}
    E(\bm{\widehat{\theta}}_\alpha - \bm{\theta}^*)^T(\bm{\widehat{\theta}}_\alpha - \bm{\theta}^*) = (\bm{\theta}_\alpha - \bm{\theta}^*)^T(\bm{\theta}_\alpha - \bm{\theta}^*) + n^{-1}Trace\left(\bm{\Psi}_n^{-1}(\bm{\theta}_\alpha) \bm{\Omega}_n(\bm{\theta}_\alpha) \bm{\Psi}_n^{-1}(\bm{\theta}_\alpha)\right),
\end{equation*}
\btxt{where $\bm{\theta}^*$ is the target (true) parameter value, $\bm{\theta}_\alpha = \bm{T}_\alpha(\bm{G})$ is the minimum DPD functional given in Eq. \eqref{nlrMDPDfn} and $\bm{\Psi_n,\Omega_n}$ are as given in Eq. \eqref{Psi}. We may estimate this MSE by replacing $\bm{\theta}_\alpha$ by $\bm{\widehat{\theta}}_\alpha$ (as defined in Section \ref{MDPDNLRM}) and $\bm{\theta}^*$ by a pilot estimator, which can then be minimized by the grid search over a suitably chosen grid of $\alpha\in[0,1]$. However, this method crucially depends on the choice of the pilot estimator $\bm{\theta}^*$, although it is shown to work very well for a properly chosen pilot.}

\btxt{To remove the dependence on the pilot estimate, Basak et al.~\cite{basak2021optimal} proposed an iterative version of the above procedure, namely the Iterated Warwick-Jones (IWJ) method, where one may use the above procedure iteratively with the pilot estimate being taken as the MDPDE with the optimum tuning parameter value obtained in the previous iteration. They have numerically illustrated that, starting with the MDPDE with any $\alpha>0$ as the pilot estimate in the first iteration, their IWJ algorithm eventually converges (very fast) to an optimal value of $\alpha$, which is independent of the initial choice of the pilot estimate. We have followed this IWJ approach to find the optimal values of $\alpha$ in all our real data examples presented in the following section.}

\section{Real data examples}\label{realdata}
Here we present the application of our proposed robust inference methodologies to fit the MM model \eqref{mmnlr} to some interesting real-life datasets, and compare them with the fits obtained by the existing robust estimators (except for the MOM estimator for the reason discussed in Section \ref{Simulation Studies}). 
In this section, two examples are presented having strong outlier effect, \btxt{and two additional examples with mild or no outliers are discussed in Web Appendix D. 
For computing the MDPDEs in all these data examples, we have again minimized the objective function in \eqref{Hnlr} using \texttt{optim} in \texttt{R} (with the BFGS method); 
several initial values around the LSE are tried, and the estimate leading to the minimum value of the objective function is finalized and reported.} 

Further, since we are unaware of the exact amount of contaminated observations for these datasets, the prediction accuracy of the fitted model is assessed by the $p\%$ trimmed average prediction error (TAPE), 
obtained after the largest $p\%$ prediction errors (squared residuals) are being trimmed.
We have reported the results with 30\% trimming as a conservative choice of $p$ for the two datasets with strong outlier effects,
and 20\% trimming for the others having milder contamination. 

\subsection{ONPG data}\label{ONPGAnalysis}
Our first example is the \btxt{ONPG} dataset, developed by Aledo \cite{aledo2022renz}. A batch of second-year biochemistry students of the University of Málaga, Spain, using $\beta$-galactosidase (EC. 3.2.1.23) as an enzyme model, carried out different experiments on the hydrolysis of o-nitrophenyl-$\beta$-D-galactopyranoside (ONPG) to galactose and o-nitrophenol, which is an enzyme-catalyzed reaction. Our data, obtained from the R package \texttt{renz} \cite{aledo2022renz}, consist of observations from 10 such reactions with different ONPG concentrations, where the initial rates (velocity) of these 10 reactions were recorded separately by 8 groups of students. Here, for illustrative purposes, we consider the data obtained by the second group of students, which contains some outlying observations (among the total of $n=10$ observations).

Since these observations are obtained from enzyme-catalyzed reactions, we fitted the MM model \eqref{mmnlr} to these data with the initial reaction rates (velocities) as the response and the ONPG concentration as the covariate, using our proposed MDPDEs as well as its robust competitors. The parameter estimates and their (estimated) asymptotic standard errors are given in Table \ref{tab:ONPG-table} along with the $30\%$ TAPE based on the fitted models for each method. The resulting fitted lines for some selected methods are also shown in Figure \ref{fig:ONPG-plot}, along with the scatter plot of the actual observations. For the selection of optimum $\alpha$, we used the IWJ method as described in Section \ref{Optalpha}, which yields $\widehat{\alpha} = 0.26$, in just one or two iterations for any choice of pilot $\alpha>0$.

\begin{table}[!h]
    \centering
    \caption{Parameter estimates of the MM model fitted to the ONPG data, along with their asymptotic standard errors (in parentheses) and 30\% TAPE. [For MDPDE, $\widehat{\alpha} = 0.26$ is the optimum choice obtained by IWJ method]}
    \resizebox{0.8\textwidth}{!}
    {\begin{tabular}{|l|rrrrrrrr|}
    \hline
        \multicolumn{9}{|c|}{MDPDE} \\
        $\alpha$ & 0 & 0.05 & 0.1 & 0.2 & 0.26 & 0.5 & 0.7 & 1 \\ \hline
        $\beta_1$ & 165.004 & 139.908 & 134.780 & 134.388 & 128.515 & 128.571 & 128.612 & 128.675 \\
        ~ & (55.673) & (36.349) & (6.736) & (6.723) & (1.621) & (1.901) & (2.162) & (2.552) \\
        $\beta_2$ & 12.350 & 7.726 & 4.096 & 4.008 & 2.582 & 2.587 & 2.590 & 2.596 \\
        ~ & (9.480) & (5.308) & (0.653) & (0.644) & (0.118) & (0.138) & (0.157) & (0.186) \\
        $\sigma$ & 21.207 & 20.687 & 5.695 & 5.664 & 1.667 & 1.847 & 1.995 & 2.180 \\
        ~ & (201.125) & (192.036) & (14.687) & (14.954) & (1.324) & (1.786) & (2.235) & (2.890) \\
        TAPE & 85.977 & 70.492 & 12.149 & 11.582 & 2.005 & 2.007 & 2.010 & 2.014 \\ \hline \hline
        ~ & \multicolumn{3}{c}{KPS} & M-est & M-est & ~ & ~ & ~ \\
        ~ & w=0.5 & w=1 & w=1.5 & Tukey & Huber & ~ & ~ & ~ \\ \hline
        $\beta_1$ & 99.010 & 101.006 & 100.861 & 132.957 & 136.401 & ~ & ~ & ~ \\
        ~ & (4.095) & (3.330) & (4.055) & (5.182) & (5.888) & ~ & ~ & ~ \\
        $\beta_2$ & 1.947 & 1.861 & 1.850 & 3.732 & 4.556 & ~ & ~ & ~ \\
        ~ & (0.311) & (0.239) & (0.291) & (0.476) & (0.609) & & & \\
        TAPE & 38.075 & 47.358 & 47.650  &10.350 & 14.461 & & & \\ \hline
    \end{tabular}}
    \label{tab:ONPG-table}
\end{table}

\begin{figure}[!h]
	\centering
	\includegraphics[width = 0.6\textwidth]{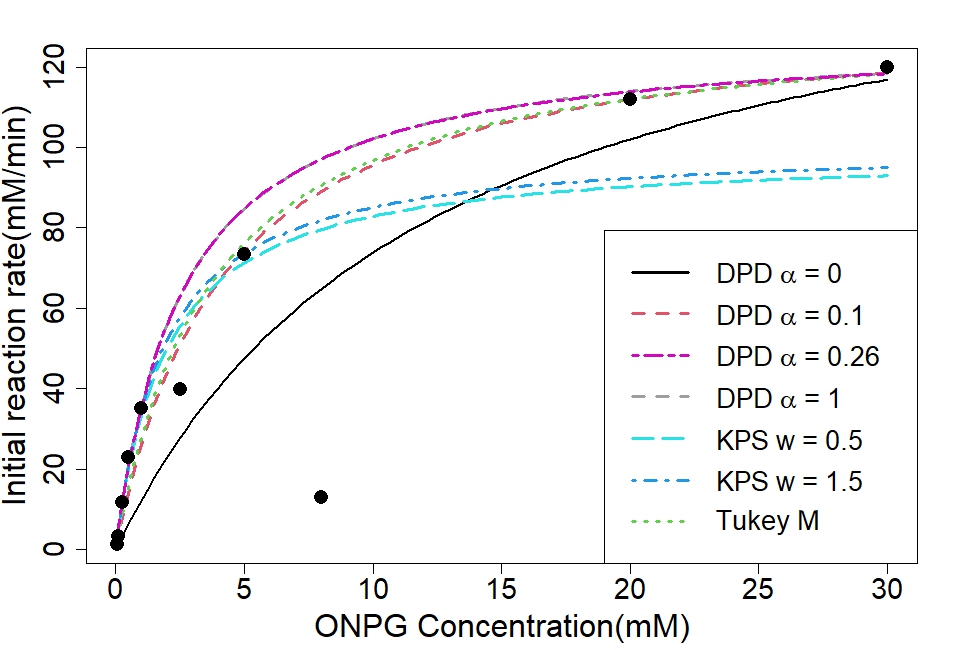}
	\caption{ONPG data and the plot of the fitted MM models using different methods}
	\label{fig:ONPG-plot}
\end{figure}

We can see from Table \ref{tab:ONPG-table} and Figure \ref{fig:ONPG-plot} that the classical OLS fit is heavily affected by the presence of outliers in the data. The KPS method is ignoring the rightmost 2 observations even though they are not outlying observations. Our proposed MDPDE with the optimum value of $\alpha=0.26$ provides much better and stable parameter estimates with lower prediction error for most of the data; in fact, all the values of $\alpha \geq \widehat{\alpha}$ yield similar estimators and model fits. MDPDE with $\alpha = 0.1$ and Tukey's M-estimator also yield reasonable model fits but higher TAPE compared to the MDPDE with $\alpha = 0.26$ (optimum choice).

\btxt{The robust MDPDEs (at the optimum $\alpha=0.26$) of the MM model parameters suggest a significantly lower value of $V_{max}$ (128.515 vs 165.004) and $K_m$ (2.582 vs 12.350) compared to the classical OLS estimates, indicating an upward bias in the classical non-robust fit. This is potentially caused by a clear outlier present in the data having a significantly low recorded initial reaction velocity. 
This smaller estimated $K_m$ indicates a higher substrate affinity, suggesting the $\beta$-galactosidase reaches the half-maximal velocity at a lower ONPG concentration than the value predicted by the OLS estimate. 
As a result, our estimate enhances the reliability in the enzyme-kinetic modelling of $\beta$-galactosidase enzyme with ONPG substrate, 
avoiding overestimation of the enzyme efficiency that might otherwise lead to inaccurate decisions during industrial bio-catalysis scale-up.}

Subsequently, we also tested for the validity of the Michaelis-Menten model (i.e, test of $\beta_2$ following the hypothesis in \eqref{testofth2}) for this dataset using the classical Wald test and the proposed MDPDE-based Wald-type test. The observed values of the test statistic $\widetilde{W}_{n,\alpha}$, described in \eqref{mmsigniftestst}, are given in Table \ref{tab:ONPGtest} along with the associated (one-sided) p-values. It is evident that the classical Wald test (at $\alpha = 0$) fails to reject the null hypothesis in \eqref{testofth2} at $5\%$ level of significance and hence fails to justify the validity of the MM model for this particular dataset, due to the presence of outliers. However, our MDPDE-based Wald-type tests with $\alpha\geq 0.1$ provide robust inferences by rejecting the null hypothesis in \eqref{testofth2} and hence justifying the validity of the MM model. In fact, these Wald-type tests at the optimum $\alpha$ and beyond (i.e., for $\alpha \geq 0.26$) all reject the null hypothesis at any reasonable level of significance, successfully suppressing any effect of outlying observations.

\begin{table}[!h]
	\centering
	\caption{Observed values of test statistic $\widetilde{W}_{n,\alpha}$ for ONPG dataset, and the associated p-values}
    \resizebox{0.8\textwidth}{!}{
	\begin{tabular}{|l|crrrrrrr|}
		\hline
		$\alpha$ & 0 (Wald test) & 0.05 & 0.1 & 0.2 & 0.26 & 0.5 & 0.7 & 1 \\ \hline
		$\widetilde{W}_{n,\alpha}$ & 1.303 & 1.456 & 6.271 & 6.227 & 21.954 & 18.740 & 16.489 & 13.982 \\
		p-value & 0.096 & 0.073 & 0.000 & 0.000 & 0.000 & 0.000 & 0.000 & 0.000\\ \hline
	\end{tabular}}
	\label{tab:ONPGtest}
\end{table}

\subsection{Drug concentration data}
Our second dataset, taken from Tabatabai et al.~\cite{tabatabai2014new}, contains a set of responses for 7 different concentrations of an agonist in a functional assay. Among the 7 observations, the 5th one is clearly seen to be an outlier; see the scatter plot in Figure \ref{fig:Drugdata-plot}. Once again, we have fitted the MM model \eqref{mmnlr} for this dataset using the proposed MDPDE, as well as its existing competitors. The parameter estimates are given in the Table \ref{tab:Drugdata} along with the asymptotic standard errors and $30\%$ TAPEs for the fitted models, while the model fits are illustrated in Figure \ref{fig:Drugdata-plot}. The optimum value of $\alpha$ in our MDPDEs, using the IWJ method, turns out to be $\widehat{\alpha}=0.39$ for this dataset.

\begin{table}[!h]
    \centering
    \caption{Parameter estimates of the MM model fitted to the Drug concentration data, along with their asymptotic standard errors (in parentheses) and 30\% TAPE. [For MDPDE, $\widehat{\alpha} =0.39$ is the optimum choice obtained by IWJ method]}
    \resizebox{0.8\textwidth}{!}
    {\begin{tabular}{|r|rrrrrrrr|}
    \hline
        \multicolumn{9}{|c|}{MDPDE} \\
        $\alpha$ & 0 & 0.05 & 0.1 & 0.2 & 0.39 & 0.5 & 0.7 & 1 \\ \hline
        $\beta_1$ & 89.212 & 88.679 & 88.143 & 86.479 & 80.892 & 80.893 & 80.893 & 80.892 \\
        ~ & (6.202) & (6.174) & (6.121) & (5.697) & (0.002) & (0.006) & (0.008) & (0.006) \\
        $\beta_2$ & 0.166 & 0.167 & 0.167 & 0.168 & 0.111 & 0.111 & 0.111 & 0.111 \\
        ~ & (0.054) & (0.054) & (0.054) & (0.052) & (0.000) & (0.000) & (0.000) & (0.000) \\
        $\sigma$ & 8.612 & 8.557 & 8.44 & 7.732 & 0.003 & 0.008 & 0.011 & 0.007 \\
        ~ & (39.646) & (39.269) & (38.546) & (33.307) & (0.000) & (0.000) & (0.000) & (0.000) \\
        TAPE & 18.488 & 17.885 & 17.279 & 13.616 & 5.447 & 5.447 & 5.446 & 5.446 \\ \hline
        ~ & \multicolumn{3}{c}{KPS} & M-est & M-est & ~ & ~ & ~ \\
        ~ & w=0.5 & w=1 & w=1.5 & Tukey & Huber & ~ & ~ & ~ \\ \hline
        $\beta_1$ & 132.890 & 132.741 & 131.901 & 82.562 & 82.986 & ~ & ~ & ~ \\
        ~ & (2.557) & (2.967) & (3.479) & (1.869) & (1.649) & ~ & ~ & ~ \\
        $\beta_2$ & 0.388 & 0.386 & 0.380 & 0.190 & 0.170 & ~ & ~ & ~ \\
        ~ & (0.031) & (0.036) & (0.042) & (0.020) & (0.017) & & & \\
        TAPE & 1.012 & 0.950 & 0.783 & 1.480 & 4.285 & & & \\ \hline
    \end{tabular}}
    \label{tab:Drugdata}
\end{table}
\begin{figure}[!h]
    \centering
    \includegraphics[width = 0.6\textwidth]{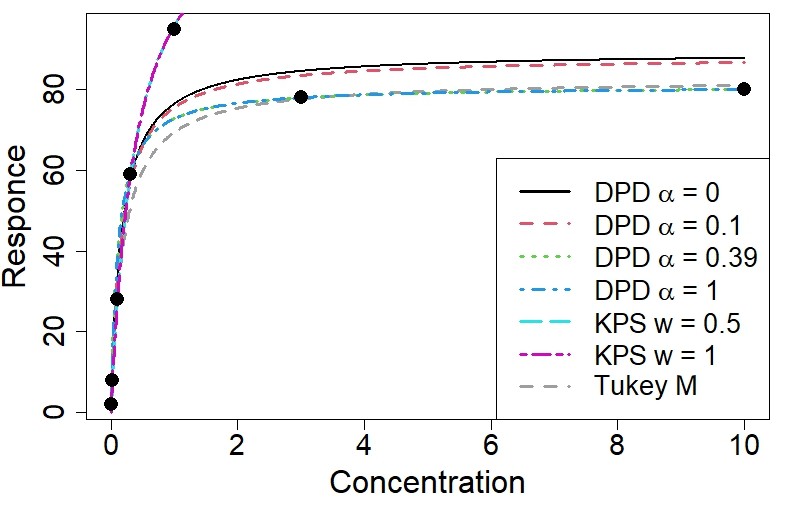}
    \caption{Drug concentration data and the plot of the fitted MM models using different methods}
    \label{fig:Drugdata-plot}
\end{figure}

Like the previous example, here also we can see that the proposed MDPDE at the optimum $\alpha = 0.39$ (in fact, for all $\alpha\geq 0.39$) yields robust parameter estimates and model fits ignoring the effects of the outlier. Tukey's M-estimator still remains a close contender with similar parameter estimates with a lower TAPE, but from the Figure \ref{fig:Drugdata-plot} it can be clearly observed that it does not provide a fit as good as the MDPDEs with $\alpha\geq 0.39$. Also, the KPS estimators for both $w=0.05$ and $w=1.5$ are heavily affected by outliers, as are the OLS and the MDPDE with $\alpha = 0.1$.

\btxt{Overall, as in the previous example, our proposed MDPDEs can again successfully address the effects of potential outliers in the data, producing lower estimates of 
	the MM parameters $V_{max}$ and $K_m$ (80.892 and 0.111, at optimum $\alpha\geq 0.39$) compared to the classical OLS estimates (89.212, 0.166). 
	Biologically, it implies that the true maximum metabolic rate of the enzyme (as predicted by the robust estimates) is lower than the one estimated by the classical non-robust estimates. 
	More importantly, in both these examples, our proposal produces significantly smaller estimates of the residual error variance ($\sigma^2$) and hence provide greater confidence in the fitted MM model than the classical competitors. This is extremely important while incorporating the fitted enzyme-kinetics models in practice, particularly in the area of drug discovery, dose-optimization strategies, and industrial process development.}

We also performed the Wald-type tests to test the validity of the MM model for this dataset. Interestingly, all these tests (for all $\alpha \geq 0$) reject the null hypothesis of insignificance of the MM model with a p-value $\leq 0.001$ in all cases. The outlier has no effect on the validity of MM model for this example, and hence all values of $\alpha$ lead to the same inference as the classical Wald test.

\section{Concluding remarks}\label{Conclution}
In this paper, we have proposed a robust method for parameter estimation under an NLR model based on the concept of minimum DPD estimation, introduced by Basu et al.~\cite{basu1998robust}. We have established the desired asymptotic properties of our estimator and provided some mild conditions on the covariates and the regression function, under which those asymptotic properties hold. Their robustness has been verified by showing that these estimators have bounded influence functions. Under the Michaelis-Menten model in enzyme kinetics, extensive simulation study has been performed to empirically justify the robust performance of our method, which is then applied to robustly analyze several real-life datasets. Another fundamental paradigm of statistical inference, namely the hypothesis testing problem, has also been discussed in this paper. The MDPDE-based Wald-type tests have been developed to test the significance of the Michaelis-Menten model. The robustness of the testing procedure has also been validated through extensive simulation studies and theoretically via influence function analyses.

Our method can be treated as an optimization procedure with the minimum DPD objective function because, in most cases, minimizing the objective function is comparatively easier than solving the estimating equations. But this optimization is also a challenging problem because of the complexity of the objective function, particularly when the dimension of the regression parameter is high, where it may take a considerable amount of time. However, in our simulation experiments with the MM model, we have observed that the computation of the MDPDEs took less time than the existing robust methods, at least for the cases of MM models. For more general nonlinear models, one can devise appropriate strategies for computing the proposed MDPDEs through a careful examination of the model-dependent objective function.

Finally, we would like to note that our NLR model is assumed to be homoscedastic, and the errors are independently distributed. However, these assumptions might not always be true in real-life situations involving NLR models. For example, some variants of the MM model (see, e.g., \cite{raaijmakers1987statistical}) assume that the responses have a constant coefficient of variation instead of having constant variance, implying heteroscedasticity in the assumed NLR model. So, an important future extension of our MDPDE-based procedures developed in this paper would be to further extend them for such heteroscedastic NLR models and the NLR models with dependent (e.g., autocorrelated) errors. We hope to pursue them in our sequel research works.


\clearpage

\input{JAS_Robust_NLR_MM_Supplementary0.tex}
\end{document}

%% file: JAS_Robust_NLR_MM_Supplementary0.tex
\appendix

\section{Some background concepts}\label{SBG}
\subsection{The DPD and the MDPDE}\label{MDPDBG}
Basu et al.~\cite{basu1998robust} proposed a robust method of estimation by minimizing the density power divergence (DPD) between the true data-generating distribution (estimated suitably from the observed data) and the model distribution under the setup of IID data. Depending on a tuning parameter $\alpha\geq 0$, the DPD measure between two probability density functions $f$ and $g$, is given by
\begin{equation*}
    d_\alpha(g,f) = \left\{
    \begin{matrix}
        \int \left\{f^{1+\alpha}(x) - \left(1+\frac{1}{\alpha}\right)f^\alpha(x) g(x) + \frac{1}{\alpha}g^{1+\alpha}(x)\right\}dx, & \mbox{if} & \alpha > 0,\\\\
        \int g(x) \ln\left(\frac{g(x)}{f(x)}\right)dx, & \mbox{if} & \alpha = 0.
    \end{matrix}
    \right.
\end{equation*}

Note that the DPD is defined as a continuous limit at $\alpha=0$, i.e., $d_0(g,f) = \lim_{\alpha\to 0}d_\alpha(g,f)$, which coincides with the Kullback-Leibler divergence. Also, the case $\alpha=1$ gives the squared $L_2$ distance. It has been well-established in the literature \cite{basu2011statistical} that the tuning parameter $\alpha$ controls the trade-off between efficiency and robustness of the associated minimum DPD estimator (MDPDE) defined below.

Let $G$ be the (unknown) true data-generating distribution with density $g$ and we want to model $g$ by the parametric family of densities $\mathcal{F} = \left\{f_{\bm{\theta}}(\cdot): \bm{\theta}\in \bm{\Theta}\subseteq\mathbb{R}^p \right\}$. Now, in the minimum DPD approach, the most appropriate model density to fit $g$ well is that member of $\mathcal{F}$ for which $d_\alpha(g,f_{\bm{\theta}})$ is minimum over $\bm{\theta}\in\bm{\Theta}$. So, the minimum DPD functional at true distribution $G$ is defined by \cite{basu1998robust, basu2011statistical}
\begin{equation*}
    \bm{T}_\alpha(G) = \arg\min_{\bm{\theta}\in\bm{\Theta}}d_\alpha(g,f_{\bm{\theta}}) = \arg\min_{\bm{\theta}\in\bm{\Theta}} \left\{ \int f_{\bm{\theta}}^{1+\alpha}(x) dx - \left(1+\frac{1}{\alpha}\right)\int f_{\bm{\theta}}^\alpha(x) dG(x) \right\}.
\end{equation*}

Now, if an IID sample $\{x_1, x_2,\ldots,x_n\}$ is available from $G$, then the MDPDE $\widehat{\bm{\theta}}_\alpha$ of $\bm{\theta}$ can be obtained by replacing $G$ with the empirical distribution function $G_n$ in the above definition, leading to
\[\widehat{\bm{\theta}}_\alpha = \bm{T}_\alpha(G_n) = \arg\min_{\bm{\theta}\in\bm{\Theta}} \left\{ \int f_{\bm{\theta}}^{1+\alpha}(x) dx - \left(1+\frac{1}{\alpha}\right)\frac{1}{n}\sum_{i=1}^nf_{\bm{\theta}}^\alpha(x_i) \right\}.\]
Note that, unlike many other minimum divergence estimators, this DPD based method avoids the (kernel) density estimation of $g$ and the associated bandwidth selection problems even under continuous models. The consistency and asymptotic normality of the MDPDEs are proved for all $\alpha\geq 0$ under suitable conditions; see, e.g., \cite{basu2011statistical}, for further details and examples.

Ghosh and Basu \cite{ghosh2013robust} generalized the definition of the MDPDE for independent non-homogeneous (INH) observations. Suppose we have $n$ independent observations $y_1, y_2,\hdots,y_n$ with $y_i\sim g_i$ for each $i=1,2,\hdots,n$, where $g_i$'s are possibly different densities. We wish to model $g_i$ by the parametric model family $\mathcal{F}_i = \left\{f_i(\cdot; \bm{\theta}): \bm{\theta}\in \bm{\Theta}\right\}$, for each $i=1,2,\hdots,n$. Note that, in this INH setup, although the distributions may be different, all of them must depend on the same parameter $\bm{\theta}$. Ghosh and Basu \cite{ghosh2013robust} defined the MDPDE of $\bm{\theta}$ as the minimizer of the average DPD measure between the data points and the corresponding models, given by
\[\frac{1}{n}\sum_{i=1}^n d_\alpha(\widehat{g}_i, f_i(\cdot; \bm{\theta})),\]
where $\widehat{g}_i$ is the density of the empirical distribution function of the data $y_i$, for each $i$. Upon simplifications, the MDPDE $\widehat{\bm{\theta}}_\alpha$ can equivalently be obtained by minimizing the objective function
\begin{equation}\label{H}
    H_n(\bm{\theta}) = \frac{1}{n}\sum_{i=1}^n V_i(y_i;\bm{\theta}),
\end{equation}
over $\bm{\theta}\in\bm{\bm{\Theta}}$, where $V_i(y_i; \bm{\theta}) = \int f_i^{1+\alpha} (y; \bm{\theta})dy - \left(1+\frac{1}{\alpha}\right) f_i^\alpha (y_i;\bm{\theta})$. So, the associated estimating equations of the MDPDEs are given by
\begin{equation*}
    \sum_{i=1}^n\left[f_i^\alpha(y_i; \bm{\theta}) \bm{u}_i(y_i;\bm{\theta}) - \int f_i^{1+\alpha}(y; \bm{\theta}) \bm{u}_i(y; \bm{\theta})dy\right] = 0,
\end{equation*}
where $\bm{u}_i(y; \bm{\theta}) = \nabla_{\bm{\theta}}\ln f_i(y; \bm{\theta})$ denotes the score function. Note that, as $\alpha\to 0$, the MDPDE is to be obtained by minimizing $\sum_{i=1}^n\left[-\ln\left(f_i(y_i; \bm{\theta})\right)\right]$, or equivalently maximizing $\sum_{i=1}^n\ln\left(f_i(y_i; \bm{\theta})\right)$ with respect to $\bm{\theta}$, or by solving the estimating equation $\sum_{i=1}^n \bm{u}_i(y_i; \bm{\theta}) = 0$. So, the MDPDE coincides with the MLE at $\alpha = 0$.

Further, under this INH setup, the minimum DPD functional corresponding to the MDPDE $\widehat{\bm{\theta}}_\alpha$ at the true distributions $\bm{G}=(G_1,G_2,\hdots,G_n)$, where $G_i$ is the distribution function associated with $g_i$ for each $i$, is given by
\begin{equation} \label{MDPDF}
    \bm{T}_\alpha(\bm{G}) = \arg\min_{\bm{\theta\in\bm{\Theta}}}\frac{1}{n}\sum_{i=1}^n d_\alpha(g_i, f_i(\cdot; \bm{\theta})).
\end{equation}

Ghosh and Basu \cite{ghosh2013robust} proved the consistency and asymptotic normality of the MDPDE $\widehat{\bm{\theta}}_\alpha$
under general INH set-ups satisfying certain assumptions, which are referred to as Assumptions (A1)--(A7) in their paper. 
In particular, it was shown that 
$$
\bm{\Omega}_n^{-\frac{1}{2}}(\bm{\theta}^g)  \bm{\Psi}_n(\bm{\theta}^g) [\sqrt{n}(\bm{\widehat{\theta}}_\alpha - \bm{\theta}^g)] \xrightarrow[]{\mathcal{D}} \mathcal{N}_p(\bm{0}, \mathbb{I}_p), 
$$
where $\bm{\theta}^g=\bm{T}_\alpha(\bm{G})$, $\bm{\Psi}_n(\bm{\theta}^g) = \frac{1}{n}\sum_{i=1}^n \bm{J}^{(i)} (\bm{\theta}^g)$, with
\begin{multline*}
    \bm{J}^{(i)}(\bm{\theta}^g)  =  \int \bm{u}_i(y;\bm{\theta}^g)\bm{u}_i(y;\bm{\theta}^g)^T f_i^{1+\alpha}(y;\bm{\theta}^g)dy\\ - \int \left\{\nabla \bm{u}_i(y;\bm{\theta}^g) + \alpha \bm{u}_i(y;\bm{\theta}^g)\bm{u}_i(y;\bm{\theta}^g)^T\right\} \left\{g_i(y) - f_i(y;\bm{\theta}^g)\right\} f_i^\alpha (y;\bm{\theta}^g) dy,
\end{multline*}
and
\begin{equation*}
    \bm{\Omega}_n(\bm{\theta}^g)  = \frac{1}{n}\sum_{i=1}^n \left[ \int \bm{u}_i(y;\bm{\theta}^g)\bm{u}_i(y;\bm{\theta}^g)^T f_i^{2\alpha}(y;\bm{\theta}^g) g_i(y) dy - \bm{\xi}_i(\bm{\theta}^g) \bm{\xi}_i(\bm{\theta}^g)^T  \right],
\end{equation*}
with $\bm{\xi}_i(\bm{\theta}^g)  = \int \bm{u}_i(y;\bm{\theta}^g) f_i^{\alpha}(y;\bm{\theta}^g) g_i(y) dy$.

The local robustness of the MDPDE under such INH set-ups is also studied by suitably defined influence functions. 
Considering the contaminated density for $i$-th observation being $g_{i,\epsilon} = (1-\epsilon)g_i + \epsilon \wedge_{t_i}$, where $\wedge_{t_i}$ is the density of degenerate distribution at the contamination point $t_i$ for $i=1,2,\hdots,n$, Ghosh and Basu \cite{ghosh2013robust} obtained the influence function of the MDPDE functional $\bm{T}_\alpha$. In particular, for contamination only in the $i_0$-th observation, this influence function turns out to be
\begin{equation}\label{IFi0}
    IF_{i_0}(t_{i_0}, \bm{T}_\alpha, \bm{G}) = \bm{\Psi}_n^{-1}(\bm{\theta}^g) \frac{1}{n}\left[ f_{i_0}^\alpha(t_{i_0};\bm{\theta}^g) \bm{u}_{i_0}(t_{i_0};\bm{\theta}^g) - \bm{\xi}_{i_0}(\bm{\theta}^g) \right],
\end{equation}
while the same with contamination in all the observations, at contamination points in $\bm{t} = (t_1,\hdots,t_n)$, is given by
\begin{equation}\label{IF}
    IF(\bm{t}, \bm{T}_\alpha, \bm{G}) = \bm{\Psi}_n^{-1}(\bm{\theta}^g)  \frac{1}{n}\sum_{i=1}^n\left[ f_i^\alpha (t_i;\bm{\theta}^g)\bm{u}_i(t_i;\bm{\theta}^g) - \bm{\xi}_i(\bm{\theta}^g) \right].
\end{equation}

\subsection{Existing robust estimators under the NLR models}\label{ExistingRM}
\textbf{M-estimators:}\\
For the NLR model given in Eq. (3) of the main paper, an M-estimator $\widehat{\bm{\beta}}_M$ of $\bm{\beta}$ is defined as
\begin{equation}
\label{MEst}
    \widehat{\bm{\beta}}_M = \arg\min_{\bm{\Theta_\beta}} \sum_{i=1}^n\rho\left(\frac{r_i}{\sigma}\right),
\end{equation}
where $r_i = y_i - \mu(\bm{x}_i, \bm{\beta}),\mbox{ for } i=1,\hdots,n$, $\bm{\Theta_\beta}$ is the parameter space of the regression parameter $\bm{\beta}$ and, if $\sigma$ is unknown, it needs to be estimated by any usual robust scale estimate, e.g., 
$$
\widehat{\sigma} = 1.4826 ~ \mbox{Med}_i(|r_i - \mbox{Med}_j (r_j)|),
$$ 
and $\rho: \mathbb{R} \to \mathbb{R}^+$ is an appropriate even function which is non-decreasing in $\mathbb{R}^+$. When $\rho(\cdot)$ is differentiable and $\psi=\rho'$, the associated estimating equation is given by
\begin{equation}
\label{Mesteq}
    \sum_{i=1}^n\psi\left(\frac{y_i-\mu(\bm{x}_i, \bm{\beta})}{\sigma}\right)\nabla_{\bm{\beta}}\mu(\bm{x}_i, \bm{\beta}) = 0.
\end{equation}

Following the asymptotic normality results of regression M estimators from Maronna et al.~\cite{maronna2019robust}, it follows that under certain conditions the M-estimator $\bm{\widehat{\beta}}_M$, obtained by solving \eqref{Mesteq}, satisfies
    $$\left(\bm{\dot{\mu}}(\bm{\beta})^T\bm{\dot{\mu}}(\bm{\beta})\right)^{1/2}\left(\widehat{\bm{\beta}}_M - \bm{\beta}\right) \xrightarrow{\mathcal{D}} \mathcal{N}_p(\bm{0},\ v_M \mathbb{I}_p),$$
where \btxt{$v_M = \sigma^2\frac{E\psi^2(U)}{\left[E\psi^{'} (U)\right]^2},\mbox{ with } U\sim\mathcal{N}(0,1)$. Hence, the ARE of the M-estimator of $\bm{\beta}$ can be computed as follows (which is the same for all the components of $\bm{\beta}$):
\begin{equation}\label{are-others}
    ARE = \frac{\sigma^2}{v_M} = \frac{\left[E\psi^{'}(U)\right]^2}{E\psi^2(U)}.
\end{equation}}
Different forms of the function $\rho(\cdot)$ lead to different M-estimators. Liu et al.~\cite{Rliu2005robust} considered the Huber's $\rho$ function (provided by Huber \cite{huber1964robust}) given by 
\[\rho_c(x) = \left\{
\begin{matrix}
    \frac{x^2}{2}, & \mbox{if} & |x|\leq c,\\
    cx - \frac{c^2}{2},& \mbox{if} & |x|>c,
\end{matrix}
\right.\]
with usual choice $c = 1.345$ for analysing different NLR models. Later, Marasovic et al.~\cite{marasovic2017robust} studied M-estimators only for the MM model (and not for general NLR models) with  the Tukey’s biweight $\rho$ function, given by
\[\rho_c(x) = \left\{
\begin{matrix}
    \frac{c^2}{6}(1 - [1 - (\frac{x}{c})^2]^3), & \mbox{if} & |x|\leq c,\\
    \frac{c^2}{6}, & \mbox{if} & |x|>c.
\end{matrix}
\right.\]
In this case, the tuning parameter $c=4.685$ yields $95\%$ efficiency under normal error assumption and can successfully handle contamination up to $10\%$ outliers (as illustrated in Section 2.1 of \cite{marasovic2017robust}).\\

\noindent\textbf{The KPS estimator:}\\
Tabatabai et al.~\cite{tabatabai2014new} proposed a robust estimator of $\bm{\beta}$ under the NLR model, by suitably modifying the M-estimation approach, as given by
\[\bm{\widehat{\beta}} = \arg\min_{\bm{\beta} \in \mathbb{R}^p}\sum_{i=1}^n\frac{\rho_\omega(t_i)}{L_i},\]
where $\rho_\omega(x) = 1 - Sech(\omega x)$ with $Sech(x)=\frac{2}{e^x + e^{-x}}$, $t_i = \frac{1}{\sigma}(1-h_{ii})(y_i - \mu(\bm{x}_i, \bm{\beta}))$, with $h_{ii}$ being the $i$-th diagonal element of the matrix $\bm{H} = \bm{\dot{\mu}}(\bm{\dot{\mu}}^T\bm{\dot{\mu}})^{-1} \bm{\dot{\mu}}^T$ and $L_i = \sum_{j=1}^k\max \{M_j, |x_{ij}|\},\ i=1,\hdots,n$, where $M_j = \mbox{Median}\{|x_{1j}|,\hdots,|x_{nj}|\},\ j=1,\hdots,k$, with $x_{ij}$ be the $j$-th element of the $k$-dimensional covariate vector $\bm{x}_i$ of the $i$-th observation. If $\sigma$ is unknown, they suggested to use either of the two robust estimates defined as 
$$
\widehat{\sigma} = 1.1926\mbox{ Med}_{1\leq i\leq n}\left(\mbox{Med}_{1\leq j\leq n}|r_i - r_j|\right),
~~~\mbox{ or }~~
\widehat{\sigma} = 2.2219\ \left\{|r_i - r_j|: 1\leq i<j\leq n\right\}_{(l)},
$$ 
with $r_i = y_i - \mu(\bm{x}_i, \bm{\widehat{\beta}})$, $l = \dbinom{\left[n/2\right] + 1}{2}$, the binomial coefficient, and $S_{(l)}$ denotes the $l$-th order statistic in the set $S$.

\btxt{The formula for the asymptotic variance of the KPS estimator of $\bm{\beta}$  has been provided by Tabatabai et al.~\cite{tabatabai2014new}, which may be used directly to compute the ARE of these estimates (as in the case of M-estimators above).}\\

\noindent\textbf{The MOM-estimator:}\\
Liu et al.~\cite{Pliu2021robust} proposed a very simple robust estimation method for the NLR model, called the median of means (MOM) method. This method consists of the following steps:
\begin{itemize}
    \item Dataset of size $n$ is divided randomly into $g$ groups of equal size, assuming $n$ is divisible by $g$. Some suitable suggestion for the choice of g is mentioned in \cite{Pliu2021robust}.
    \item Obtain the OLS estimate of $\bm{\beta}$ from every group of observations separately, which we denote as $\widehat{\bm{\beta}}^{(j)} = (\widehat{\beta}_1^{(j)},\hdots,\widehat{\beta}_p^{(j)})^T$ for group $j,\ j=1,2,\hdots,g$.
    \item The MOM estimator is then given by $\widehat{\bm{\beta}}^{MOM} =(\widehat{\beta}_1^{MOM},\hdots,\widehat{\beta}_p^{MOM})^T$, where
$\widehat{\beta}_k^{MOM} = \mbox{Median}\{\widehat{\beta}_k^{(1)},\hdots,\widehat{\beta}_k^{(g)}\},\mbox{ for all } k=1,2,\hdots,p$.
\end{itemize}

Liu et al.~\cite{Pliu2021robust} have also proved asymptotic properties (consistency and asymptotic normality) of the estimator under certain assumptions. However, from the definition of the method, it is intuitively clear that it gives good, stable performance only when the sample size is moderately large and the proportion of outlier is small, which is also justified through our simulation studies, presented in Section 4.3 of the main paper.

\subsection{\btxt{Level and Power influence functions under INH observations}}
\btxt{The level influence function (LIF) and power influence function (PIF) are effective measures to study the local robustness of the level and power of a statistical test (see, Heritier and Ronchetti \cite{heritier1994robust}, Ghosh et al.~\cite{ghosh2016influence}, Basu et al.~\cite{basu2018robust}). Let us consider the setup of INH observations from Section \ref{MDPDBG}, where $\mathcal{F}_i = \left\{f_i(\cdot; \bm{\theta}): \bm{\theta}\in \bm{\Theta}\subseteq \mathbb{R}^p\right\}$ be the parametric model family of densities of $y_i$, for each $i=1,\ldots,n$, with $F_{i,\bm{\theta}}$ denoting the distribution function associated to $f_i(\cdot; \bm{\theta})$. Suppose that $W_{n}$ denotes any general test statistic (e.g., the Wald-type test statistic) for testing the hypotheses
\begin{equation} \label{lp-if-hyp}
    H_0: \bm{\theta} = \bm{\theta}_0 \mbox{ against } H_1: \bm{\theta} \neq \bm{\theta}_0, \mbox{ for some } \bm{\theta}_0\in \bm{\Theta},
\end{equation}
and let the critical region for this test at the significance level $\gamma$ is denoted by $C_\gamma$. Further, let $W(\bm{G})$ denotes the statistical functional associated with the test statistic $W_n$.} 

\btxt{We now consider the asymptotic power of the test under the contiguous alternative hypothesis $H_{1,n}: \bm{\theta} = \bm{\theta}_n$, where $\bm{\theta}_n = \bm{\theta}_0 + \bm{d}/\sqrt{n},~~\bm{d}\neq \bm{0}$, and incorporate appropriate contamination to both the null and the contiguous alternative hypotheses, for each $i=1,\ldots,n$, as follows:
\begin{equation*}
    F^{L}_{i,n,\epsilon,t_i}
    = \left( 1 - \frac{\epsilon}{\sqrt{n}} \right) F_{i,\bm{\theta}_0}
      + \frac{\epsilon}{\sqrt{n}} \Delta_{t_i},
\qquad
F^{P}_{i,n,\epsilon,t_i}
    = \left( 1 - \frac{\epsilon}{\sqrt{n}} \right) F_{i,\bm{\theta}_n}
      + \frac{\epsilon}{\sqrt{n}} \Delta_{t_i},
\end{equation*}
where $\Delta_{t_i}$ is the one-point distribution function supported on $\{t_i\}$, which are then used to analyze the robustness of level and power of the tests, respectively. Now, denoting $\bm{t}=(t_1,\ldots,t_n)^T,\ \bm{F}^{P}_{n,\epsilon,\bm{t}} = \left( F^{P}_{1,n,\epsilon,t_1}, \ldots, F^{P}_{n,n,\epsilon,t_n} \right)^T$ and $\bm{F}^{L}_{n,\epsilon,\bm{t}} = \left( F^{L}_{1,n,\epsilon,t_1}, \ldots, F^{L}_{n,n,\epsilon,t_n} \right)^T$, the LIF and PIF of the test functional $W(\bm{G})$, for the simple null hypothesis $H_0$, are defined as (at the level of significance $\gamma$)
\begin{align*}
    LIF(\bm{t}; W, \bm{F}_{\bm{\theta}_0})
    &= \lim_{n \to \infty}
      \frac{\partial}{\partial \epsilon}
      P_{\bm{F}^{L}_{n,\epsilon,\bm{t}}}
      \left( W_n \in C_\gamma \right)
      \Big|_{\epsilon = 0},\\
      PIF(\bm{t}; W, \bm{F}_{\bm{\theta}_0})
    &= \lim_{n \to \infty}
      \frac{\partial}{\partial \epsilon}
      P_{\bm{F}^{P}_{n,\epsilon,\bm{t}}}
      \left( W_n \in C_\gamma \right)
      \Big|_{\epsilon = 0},
\end{align*}
For more details about LIF and PIF, see \cite{ghosh2016influence, basu2018robust}.}

\btxt{Particularly, for the MDPDE-based Wald-type tests under the general INH setups, their LIF and PIF are studied in detail in Basu et al.~\cite{basu2018robust}. It has been shown there that, for testing the hypothesis of the form \eqref{lp-if-hyp}, the associated test statistics and test functionals are, respectively, given by
\begin{align*}
    W_n &= n (\widehat{\bm{\theta}}_\alpha - \bm{\theta}_0)^T \bm{\Sigma}_\alpha^{-1}(\bm{\theta}_0) (\widehat{\bm{\theta}}_\alpha - \bm{\theta}_0), ~~~\mbox{ and }~~~~
    W_\alpha(\bm{G}) = (\bm{T}_\alpha(\bm{G}) - \bm{\theta}_0)^T \bm{\Sigma}_\alpha^{-1}(\bm{\theta}_0) (\bm{T}_\alpha(\bm{G}) - \bm{\theta}_0),
\end{align*}
where $\bm{\Sigma}_\alpha(\bm{\theta}_0) = \lim_{n\to\infty} \bm{\Psi}_n^{-1}(\bm{\theta}_0)\bm{\Omega}_n(\bm{\theta}_0)\bm{\Psi}_n^{-1}(\bm{\theta}_0)$, with $\bm{\Psi}_n$ and $\bm{\Omega}_n$ being as given in Section \ref{MDPDBG}. Now, the LIF and PIF of $W_\alpha$, in this case, at the vector of contamination points $\bm{t}$, are given by (Theorem 10 of \cite{basu2018robust})
\begin{align}
    LIF\left(\bm{t}; W_\alpha, \bm{F}_{\bm{\theta}_0}\right) &= 0, \mbox{ and }\\
    PIF\left(\bm{t}; W_\alpha, \bm{F}_{\bm{\theta}_0}\right)
    &= K_p^*(\bm{d}^T\bm{\Sigma}_\alpha^{-1}(\bm{\theta}_0)\bm{d}) \bm{d}^T\bm{\Sigma}_\alpha^{-1}(\bm{\theta}_0) IF\left(\bm{t}; \bm{T}_\alpha, \bm{F}_{\bm{\theta}_0}\right),
\end{align}
where the expression of $K_p^*(\cdot)$ is given in Section 3.2 of the main paper,
and $IF\left(\bm{t}; \bm{T}_\alpha, \bm{F}_{\bm{\theta}_0}\right)$ is the IF of the MDPDE functional $\bm{T}_\alpha$
presented in Section \ref{MDPDBG}. 
One can also define the test functional $W_\alpha$  directly using the matrix $\bm{\Psi}_n^{-1}(\bm{\theta}_0)\bm{\Omega}_n(\bm{\theta}_0)\bm{\Psi}_n^{-1}(\bm{\theta}_0)$ instead of its limit $\bm{\Sigma}_\alpha(\bm{\theta}_0)$;
then also the resulting LIF and PIF will have the same form as in (9) and (10) above 
and can be computed in practice by replacing $\bm{\Sigma}_\alpha(\bm{\theta}_0)$ with $\bm{\Psi}_n^{-1}(\bm{\theta}_0)\bm{\Omega}_n(\bm{\theta}_0)\bm{\Psi}_n^{-1}(\bm{\theta}_0)$,
as done in our main paper. 
}

\section{Proof of theorems}\label{Proofs}
\subsection{Proof of Theorem 2.1} \label{Proof-AN}
First, we will state and prove a lemma and then use Theorem 3.1 of Ghosh and Basu \cite{ghosh2013robust} to complete the proof of our theorem.
\begin{lemma}
\label{lemma-4.1}
    Consider the setup of the NLR model (3) of the main paper and suppose that the true data-generating distributions belong to the assumed parametric model families. Then Conditions (R1) and (R2) of the main paper imply Assumptions (A1)-(A7) of Ghosh and Basu \cite{ghosh2013robust}.
\end{lemma}
\begin{proof}
As the true distributions belong to the corresponding parametric model families, under which the distribution of $y_i$ is $\mathcal{N}\left(\mu(\bm{x}_i, \bm{\beta}), \sigma^2\right)$ with $\mu(\bm{x}_i, \bm{\beta})$ being thrice continuously differentiable, Assumptions (A1)-(A3) of Ghosh and Basu \cite{ghosh2013robust} directly follow from the property of the normal density in view of (R1). Assumption (A4) follows from Eq. (14) of Condition (R2). The form of $V_i(y_i;\bm{\theta})$ (given in Section 2.1 of the main paper) ensures that Assumption (A5) follows from Eq. (12) of (R1). Now, to prove Eq. (3.8) of Assumption (A6) of Ghosh and Basu \cite{ghosh2013robust}, we note that
\[\nabla_jV_i(y_i;\bm{\theta}) = \kappa e^{-\frac{\alpha}{2\sigma^2}(y_i-\mu(\bm{x}_i, \bm{\beta}))^2}(y_i-\mu(\bm{x}_i, \bm{\beta}))\dot{\mu}_{ij}(\bm{\beta}),~~~~ j=1,\hdots,p,\]
where $\nabla_j$ is the first order partial derivative with respect to the $j$-th component of $\bm{\beta}$ and $\kappa = -\frac{1+\alpha}{(2\pi)^{\alpha/2}\sigma^{\alpha+2}}$. Then,
\begin{align*}
    & \frac{1}{n}\sum_{i=1}^n E_i[|\nabla_jV_i(y_i;\bm{\theta})|I(|\nabla_jV_i(y_i;\bm{\theta})|>N)]\\
    & = |\kappa|\frac{1}{n}\sum_{i=1}^n E_i\left[e^{-\frac{\alpha z_i^2}{2\sigma^2}}|z_i||\dot{\mu}_{ij}(\bm{\beta})|\times I\left(e^{-\frac{\alpha z_i^2}{2\sigma^2}}|z_i||\dot{\mu}_{ij}(\bm{\beta})|>\frac{N}{|\kappa|}\right)\right], \mbox{ letting } z_i = y_i-\mu(\bm{x}_i, \bm{\beta})\\
    & \leq |\kappa|\frac{1}{n}\sum_{i=1}^n |\dot{\mu}_{ij}(\bm{\beta})| E_i\left[e^{-\frac{\alpha z_i^2}{2\sigma^2}}|z_i|\times I\left(e^{-\frac{\alpha z_i^2}{2\sigma^2}}|z_i|>\frac{N}{|\kappa|(\sup_{n>1}\max_{1\leq i\leq n}|\dot{\mu}_{ij}(\bm{\beta})|)}\right)\right]\\
    & = |\kappa|E_1\left[e^{-\frac{\alpha z_1^2}{2\sigma^2}}|z_1|\times I\left(e^{-\frac{\alpha z_1^2}{2\sigma^2}}|z_1|>\frac{N}{|\kappa|(\sup_{n>1}\max_{1\leq i\leq n}|\dot{\mu}_{ij}(\bm{\beta})|)}\right)\right] \left(\frac{1}{n}\sum_{i=1}^n|\dot{\mu}_{ij}(\bm{\beta})|\right),
\end{align*}
as $z_i$'s are IID. Since we have $\sup_{n>1}\max_{1\leq i\leq n}|\dot{\mu}_{ij}(\bm{\beta})| = O(1)$ by Eq. (12) of (R1), so, by dominated convergence theorem (DCT)
\[\lim_{N\to\infty}E_1\left[e^{-\frac{\alpha z_1^2}{2\sigma^2}}|z_1|\times I\left(e^{-\frac{\alpha z_1^2}{2\sigma^2}}|z_1|>\frac{N}{|\kappa|(\sup_{n>1}\max_{1\leq i\leq n}|\dot{\mu}_{ij}(\bm{\beta})|)}\right)\right] = 0.\]
We also have
\[\sup_{n>1}\left(\frac{1}{n}\sum_{i=1}^n|\dot{\mu}_{ij}(\bm{\beta})|\right) \leq \sup_{n>1}\max_{1\leq i\leq n}|\dot{\mu}_{ij}(\bm{\beta})| = O(1).\]
Hence, Eq. (3.8) of (A6) follows for all $j=1,2,\hdots,p$. Similarly, one can show that it also follows for $j=p+1$, i.e., for $\nabla_{\sigma^2}$, the gradient with respect to $\sigma^2$. Similarly, Eq. (3.9) of Assumption (A6) and Assumption (A7) of Ghosh and Basu \cite{ghosh2013robust} can also be shown to hold using the second part of Eq. (11) of (R1) and Eq. (14) of (R2), respectively.
\end{proof}

\noindent\textbf{Proof of the theorem:}\\
\btxt{An application of Lemma \ref{lemma-4.1} suggests that Part \textit{(i)} of the theorem follows directly from Theorem 3.1 of Ghosh and Basu \cite{ghosh2013robust}.} Further, using Lemma \eqref{lemma-4.1}, it also follows from Theorem 3.1 of Ghosh and Basu \cite{ghosh2013robust} that the asymptotic distribution of $\bm{\Omega}_n^{-\frac{1}{2}}\bm{\Psi}_n [\sqrt{n}(\bm{\widehat{\theta}}_\alpha - \bm{\theta}_0)]$ is $\mathcal{N}_{p+1}(\bm{0}, \mathbb{I}_{p+1})$, where $\bm{\Omega}_n$ and $\bm{\Psi}_n$ are as defined in Section \ref{MDPDBG}, and are to be computed at the model densities, with $\bm{\theta} = \bm{\theta}_0$. For the present case of NLR model (3) of our main paper, we can see that
\begin{equation}
\label{ui}
    \bm{u}_i(y;\bm{\theta}) = \nabla_{\bm{\theta}}\ln f_i(y; \bm{\theta}) = 
    \begin{pmatrix}
        \frac{(y-\mu(\bm{x}_i, \bm{\beta}))}{\sigma^2}\nabla_{\bm{\beta}}\mu(\bm{x}_i, \bm{\beta})\\\\
        \frac{(y-\mu(\bm{x}_i, \bm{\beta}))^2}{2\sigma^4} - \frac{1}{2\sigma^2}
    \end{pmatrix}.
\end{equation}
After some simple calculations we obtain the $(p+1)\times(p+1)$ matrix $\bm{J}_\alpha^{(i)}$ as
\begin{equation}
\label{Ji}
   \bm{J}_\alpha^{(i)} = \int \bm{u}_i(y;\bm{\theta})\bm{u}_i(y;\bm{\theta})^T f_i^{1+\alpha}(y;\bm{\theta})dy = 
   \begin{bmatrix}
       \zeta_\alpha\nabla_{\bm{\beta}}\mu(\bm{x}_i, \bm{\beta})\nabla_{\bm{\beta}}\mu(\bm{x}_i, \bm{\beta})^T & \bm{0}_p\\
       \bm{0}_p^T & \varsigma_\alpha
   \end{bmatrix}.
\end{equation}
So, a simplified form of the matrix $\bm{\Psi}_n$ is given by
\begin{equation*}
    \bm{\Psi}_n = \frac{1}{n}\sum_{i=1}^n \bm{J}_\alpha^{(i)} =
    \begin{bmatrix}
        \frac{\zeta_\alpha}{n} \left(\bm{\dot{\mu}}(\bm{\beta})^T\bm{\dot{\mu}}(\bm{\beta})\right) & \bm{0}_p\\
        \bm{0}_p^T & \varsigma_\alpha
    \end{bmatrix}.
\end{equation*}
We also have
\begin{equation}
\label{xii}
    \bm{\xi}_i = \int \bm{u}_i(y;\bm{\theta})f_i^{1+\alpha}(y;\bm{\theta})dy = 
\begin{pmatrix}
    \bm{0}_p\\\\
    -\frac{\alpha}{2}\zeta_\alpha
\end{pmatrix},
\end{equation}
and using it, we obtain a simplified form of the matrix $\bm{\Omega}_n$ as 
\begin{equation*}
    \bm{\Omega}_n = \frac{1}{n}\sum_{i=1}^n(\bm{J}^{(i)}_{2\alpha} - \bm{\xi}_i\bm{\xi}_i^T) =
    \begin{bmatrix}
        \frac{\zeta_{2\alpha}}{n} \left(\bm{\dot{\mu}}(\bm{\beta})^T\bm{\dot{\mu}}(\bm{\beta})\right) & \bm{0}_p\\
        \bm{0}_p^T & \varsigma_{2\alpha} - \frac{\alpha^2}{4}\zeta_\alpha^2
    \end{bmatrix}.
\end{equation*}
Now using these simplified form of $\bm{\Psi}_n$ and $\bm{\Omega}_n$, at $\bm{\theta} = \bm{\theta}_0$, we have
\[\bm{\Omega}_n^{-\frac{1}{2}}\bm{\Psi}_n\left[\sqrt{n}\left(\bm{\widehat{\theta}}_\alpha - \bm{\theta}_0\right)\right] =
\begin{pmatrix}
    \frac{\zeta_\alpha}{\sqrt{\zeta_{2\alpha}}}\left(\bm{\dot{\mu}}(\bm{\beta}_0)^T\bm{\dot{\mu}}(\bm{\beta}_0)\right)^{1/2}\left(\bm{\widehat{\beta}}_\alpha - \bm{\beta}_0\right)\\\\
    \frac{\varsigma_\alpha}{\sqrt{\varsigma_{2\alpha}-\frac{\alpha^2}{4}\zeta_\alpha^2}}\sqrt{n}\left(\widehat{\sigma}_\alpha^2 - \sigma_0^2\right)
\end{pmatrix},
\]
which asymptotically follows  $\mathcal{N}_{p+1}(\bm{0}, \mathbb{I}_{p+1})$, where $\zeta_p$'s and $\varsigma_p$'s are to be computed at $\sigma = \sigma_0$. Therefore, asymptotically, $\frac{\zeta_\alpha}{\sqrt{\zeta_{2\alpha}}}\left(\bm{\dot{\mu}}(\bm{\beta}_0)^T\bm{\dot{\mu}}(\bm{\beta}_0)\right)^{1/2} \left(\bm{\widehat{\beta}}_\alpha - \bm{\beta}_0\right)$ and 
$\frac{\varsigma_\alpha}{\sqrt{\varsigma_{2\alpha}-\frac{\alpha^2}{4}\zeta_\alpha^2}}\sqrt{n}\left(\widehat{\sigma}^2 - \sigma_0^2\right)$ are independently distributed with their respective asymptotic distributions being $\mathcal{N}_p(\bm{0}, \mathbb{I}_p)$ and $\mathcal{N}(0,1)$. This completes the proof.

\subsection{Proof of Proposition 4.1}\label{pf-th-R1}
For the case of the MM model (Eq. (28) of the main paper), $\bm{\dot{\mu}}(\bm{\beta})$ (given in Eq. (29) of the main paper) is of order $n\times 2$ and
\begin{equation}\label{verify-R1-1}
    \sup_{n>1}\max_{1\leq i\leq n}|\dot{\mu}_{i1}(\bm{\beta})| = \sup_{n>1}\max_{1\leq i\leq n} \frac{x_i}{\beta_2 + x_i},
\end{equation}
\begin{equation}\label{verify-R1-2}
    \sup_{n>1}\max_{1\leq i\leq n}|\dot{\mu}_{i2}(\bm{\beta})| = |\beta_1|\sup_{n>1}\max_{1\leq i\leq n} \frac{x_i}{(\beta_2 + x_i)^2},
\end{equation}
\begin{equation}\label{verify-R1-3}
    \sup_{n>1}\max_{1\leq i\leq n}|\dot{\mu}_{i1}(\bm{\beta})\dot{\mu}_{i2}(\bm{\beta})| = |\beta_1|\sup_{n>1}\max_{1\leq i\leq n} \frac{x_i^2}{(\beta_2 + x_i)^3}.
\end{equation}
We have $x_i>0$ for all $i$ and $\beta_2\geq 0$. So, $\frac{x_i}{\beta_2 + x_i}$ is bounded between $0$ and $1$. So, the quantity in Equation \eqref{verify-R1-1} is $O(1)$.

Now consider the functions $u_1(x) = \frac{x}{(\beta_2 + x)^2}$ and $u_2(x) = \frac{x^2}{(\beta_2 + x)^3}$, both defined for $x>0$. So, both $u_1$ and $u_2$ are positive real valued functions. And, by simple calculations, we get
\[\frac{d}{dx}u_1(x) =  \frac{\beta_2 - x}{(\beta_2 + x)^3} = \left\{
\begin{matrix}
    >0, & \mbox{if} & x<\beta_2,\\
    =0, & \mbox{if} & x=\beta_2,\\
    <0, & \mbox{if} & x>\beta_2,
\end{matrix}
\right.\]
and
\[\frac{d}{dx}u_2(x) =  \frac{x(2\beta_2 - x)}{(\beta_2 + x)^4} = \left\{
\begin{matrix}
    >0, & \mbox{if} & x<2\beta_2,\\
    =0, & \mbox{if} & x=2\beta_2,\\
    <0, & \mbox{if} & x>2\beta_2.
\end{matrix}
\right.\]
So, $u_1(x)$ attains its unique maxima at $x=\beta_2$, and hence it is bounded in $(0, u_1(\beta_2)]$. This shows that the quantity in Equation \eqref{verify-R1-2} is $O(1)$. Similarly, $u_2(x)$ attains its unique maxima at $x=2\beta_2$, and hence it is bounded in $(0, u_2(2\beta_2)]$, which indicates that the quantity in Equation \eqref{verify-R1-3} is $O(1)$. Therefore, (R1) is satisfied.

\subsection{Proof of Proposition 4.2}\label{pf-th-R2}
Let us define $A_n = \frac{1}{n}\bm{\dot{\mu}}(\bm{\beta})^T\bm{\dot{\mu}}(\bm{\beta})$. The characteristic equation of $A_n$ is
$\lambda^2 - tr(A_n)\lambda + det(A_n) = 0.$
So, the eigenvalues of $A_n$ are given by
$\lambda = \frac{1}{2}\left[tr(A_n) \pm \sqrt{tr(A_n)^2 - 4det(A_n)}\right]$.
Here $\lambda$'s are real as $A_n$ is a real symmetric matrix. 

Now, in order to have the minimum eigenvalue of $A_n$ strictly positive, we must have
\begin{align*}
    & tr(A_n) - \sqrt{tr(A_n)^2 - 4det(A_n)} > 0 ~~~~~~~ \Leftrightarrow ~~~det(A_n) > 0\\
    & \Leftrightarrow \sum_{i=1}^n a_i^2\ \  \sum_{i=1}^n b_i^2 > \left(\sum_{i=1}^n a_ib_i\right)^2,
   ~~~\mbox{ with }~~a_i = \frac{x_i}{\beta_2 + x_i}, ~b_i = \frac{\beta_1 x_i}{(\beta_2 + x_i)^2},
\end{align*}
which is true by Cauchy-Schwartz inequality unless the equality occurs. The equality occurs only when $a_i \propto b_i\  \mbox{ for all } i$, i.e., when $x_i =$ constant,  for all $i$ or $\beta_1 = 0$.

Therefore, in order to have the minimum eigenvalue of $A_n$ strictly positive and satisfy Condition (R1), it is necessary that the values $x_i$'s are not all same, for $n\geq 2$, and $\beta_1 \neq 0$.

\section{Additional Illustrations under the MM model}
Here we present additional illustrations of the theoretical results and empirical performances under the MM model, as indicated in the main manuscript.

\subsection{More Influence function plots of the MDPDEs}\label{IFMM}
Figures \ref{fig:IF-S1-10-40}-\ref{fig:IF-S2-15-45} present the IFs of the MDPDEs of the MM model parameters ($\beta_1, \beta_2$ and $\sigma^2$) for contamination in a particular observation, corresponding to the $i_0$-th smallest $x_i$. These IFs are again seen to be bounded for all $\alpha>0$ and unbounded at $\alpha=0$, for any choice of $i_0$ under both setups.

\begin{figure}[!h]
     \centering
     \begin{subfigure}[b]{0.45\textwidth}
         \centering
         \includegraphics[width=\textwidth]{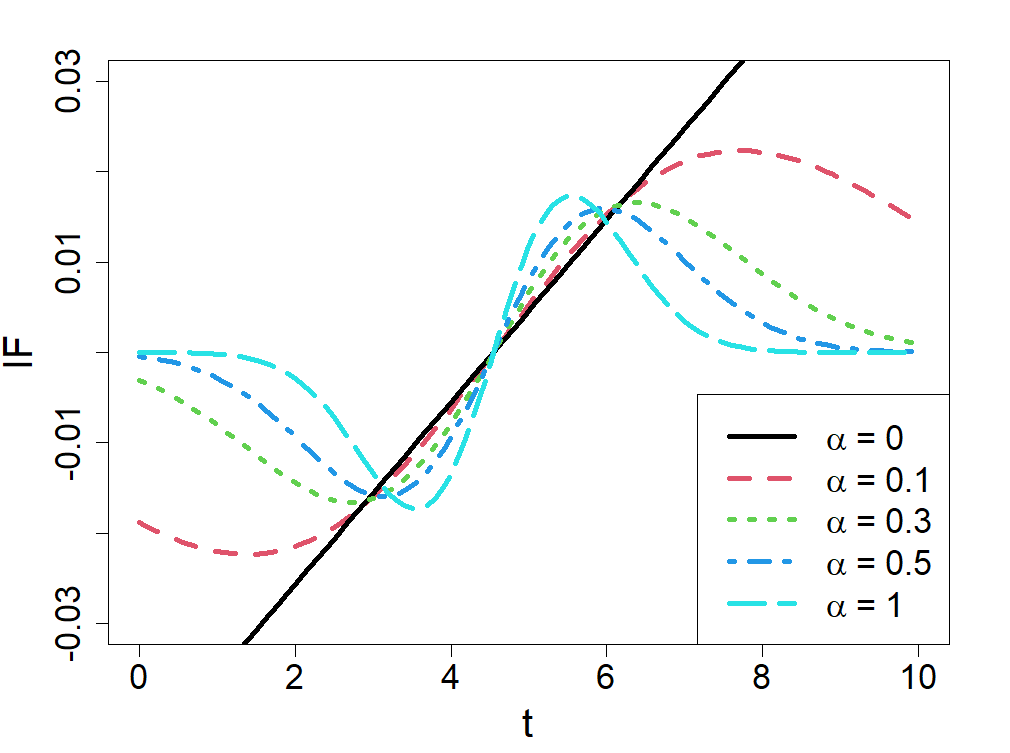}
         \caption{IF of $T_\alpha^{\beta_1}$ with $i_0 = 10$}
         \label{fig:IF-S1-10-th1}
     \end{subfigure}
     \hfill
     \begin{subfigure}[b]{0.45\textwidth}
         \centering
         \includegraphics[width=\textwidth]{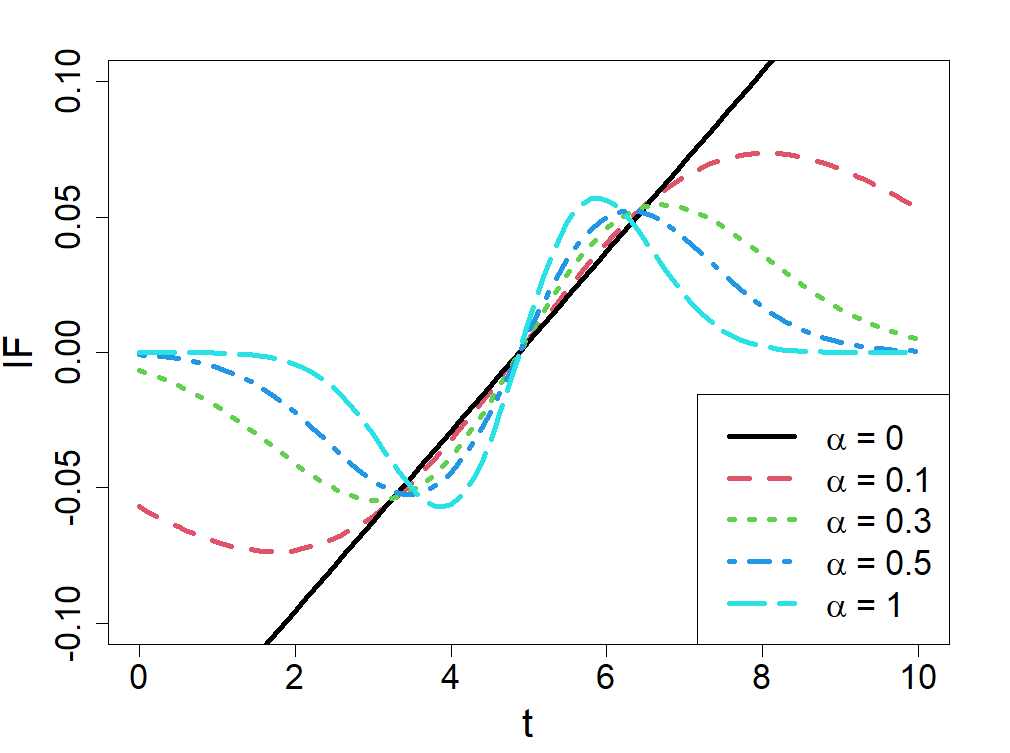}
         \caption{IF of $T_\alpha^{\beta_1}$ with $i_0 = 40$}
         \label{fig:IF-S1-40-th1}
     \end{subfigure}
     \hfill
     \begin{subfigure}[b]{0.45\textwidth}
         \centering
         \includegraphics[width=\textwidth]{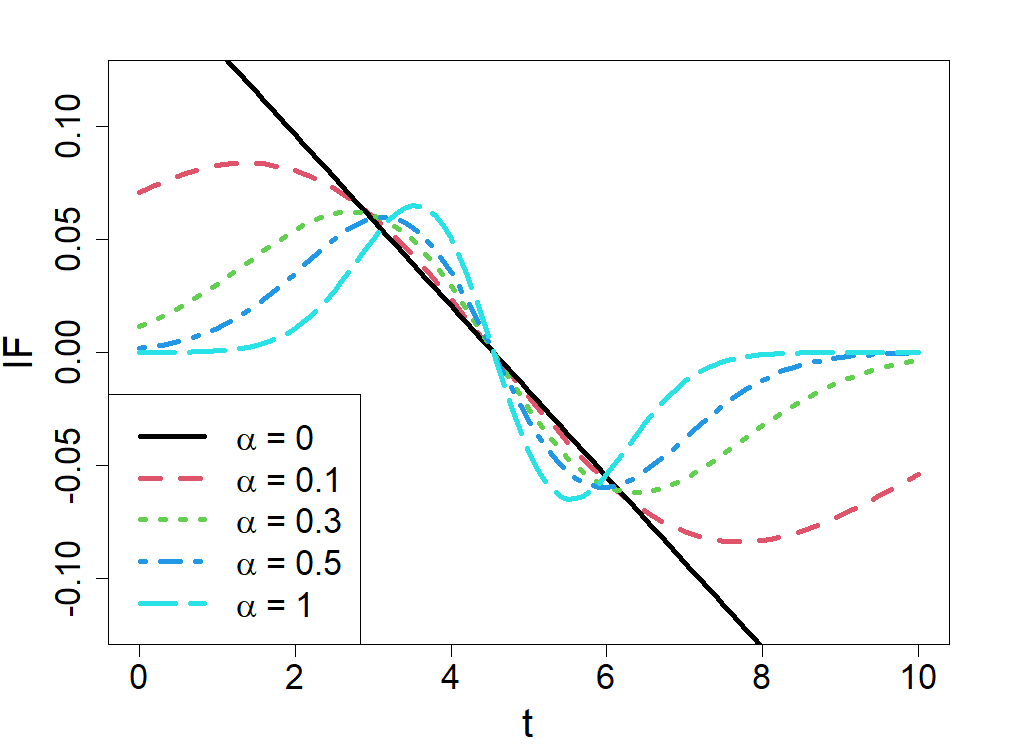}
         \caption{IF of $T_\alpha^{\beta_2}$ with $i_0 = 10$}
         \label{fig:IF-S1-10-th2}
     \end{subfigure}
     \hfill
     \begin{subfigure}[b]{0.45\textwidth}
         \centering
         \includegraphics[width=\textwidth]{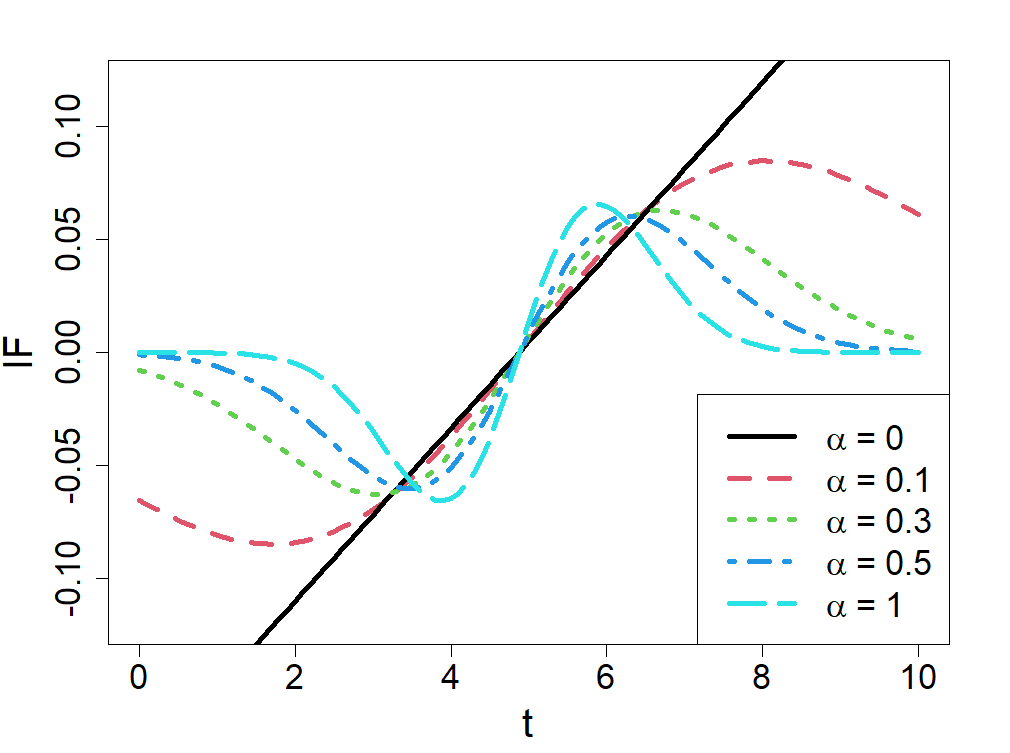}
         \caption{IF of $T_\alpha^{\beta_2}$ with $i_0 = 40$}
         \label{fig:IF-S1-40-th2}
     \end{subfigure}
     \hfill
     \begin{subfigure}[b]{0.45\textwidth}
         \centering
         \includegraphics[width=\textwidth]{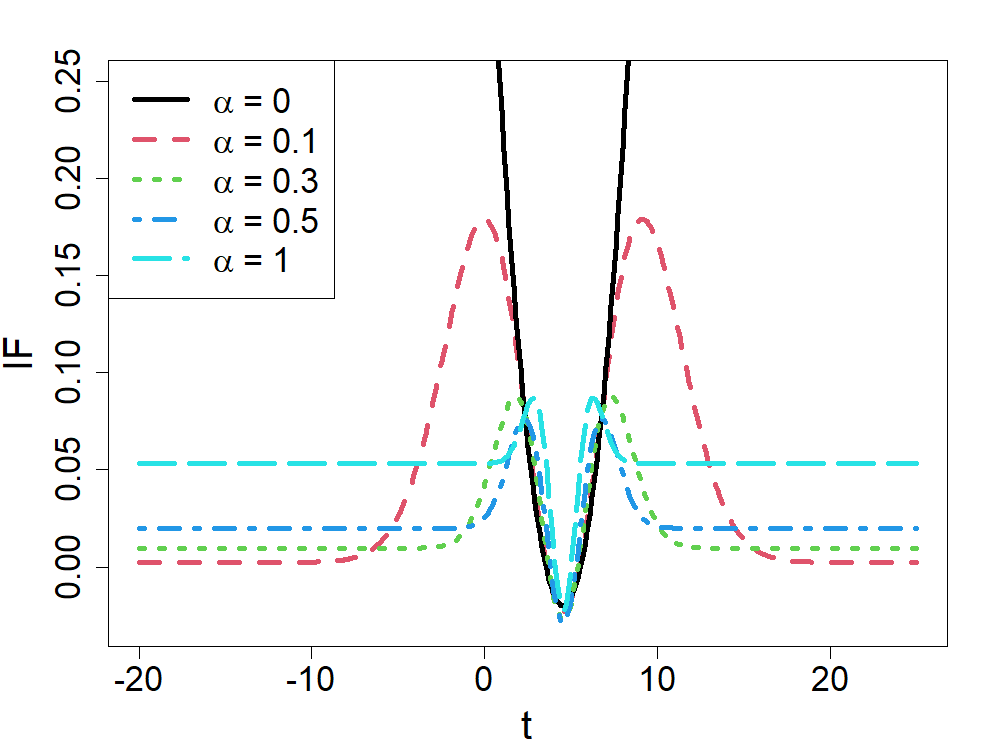}
         \caption{IF of $T_\alpha^{\sigma^2}$ with $i_0 = 10$}
         \label{fig:IF-S1-10-sig}
     \end{subfigure}
     \hfill
     \begin{subfigure}[b]{0.45\textwidth}
         \centering
         \includegraphics[width=\textwidth]{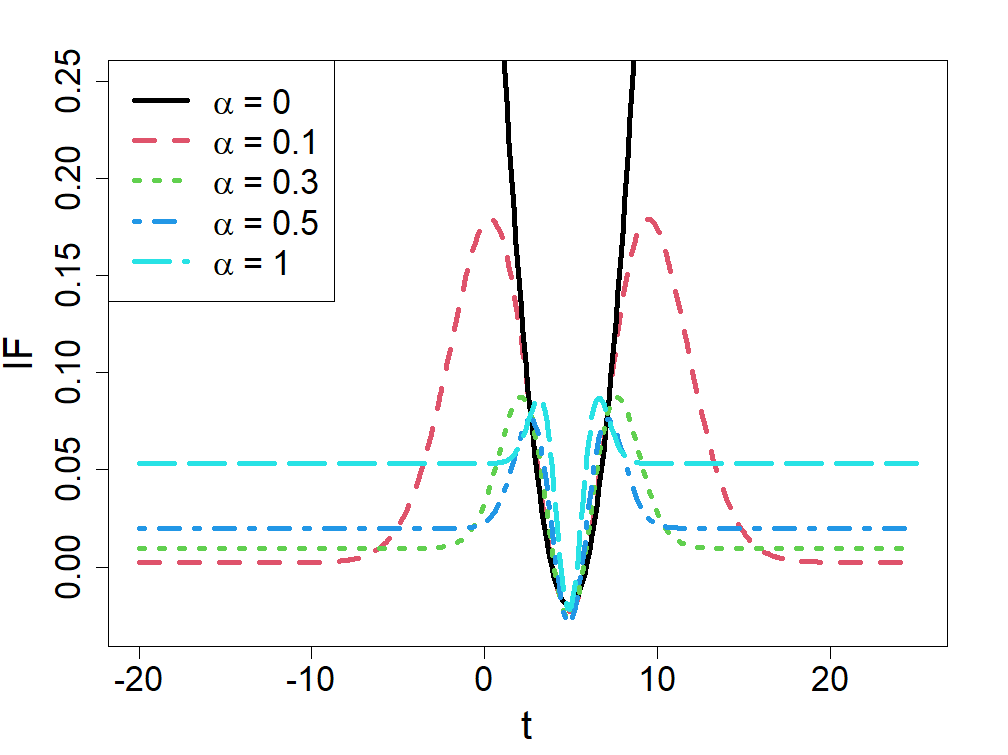}
         \caption{IF of $T_\alpha^{\sigma^2}$ with $i_0 = 40$}
         \label{fig:IF-S1-40-sig}
     \end{subfigure}
     \caption{\btxt{IFs of the MDPDEs of the parameters $\beta_1$, $\beta_2$, and $\sigma^2$ in the MM model, for contamination only in one observation, corresponding to the $i_0$-th order statistic of the covariate $x_i$s, under Setup 1}}
     \label{fig:IF-S1-10-40}
\end{figure}

\begin{figure}[!h]
     \centering
     \begin{subfigure}[b]{0.45\textwidth}
         \centering
         \includegraphics[width=\textwidth]{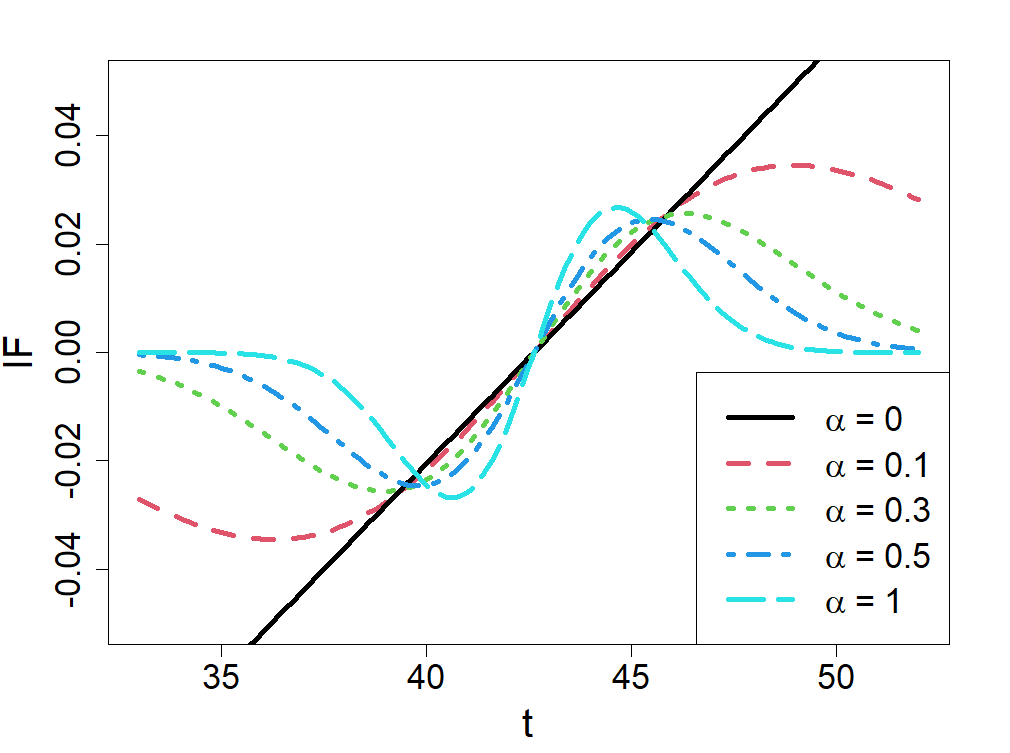}
         \caption{IF of $T_\alpha^{\beta_1}$ with $i_0 = 15$}
         \label{fig:IF-S2-15-th1}
     \end{subfigure}
     \hfill
     \begin{subfigure}[b]{0.45\textwidth}
         \centering
         \includegraphics[width=\textwidth]{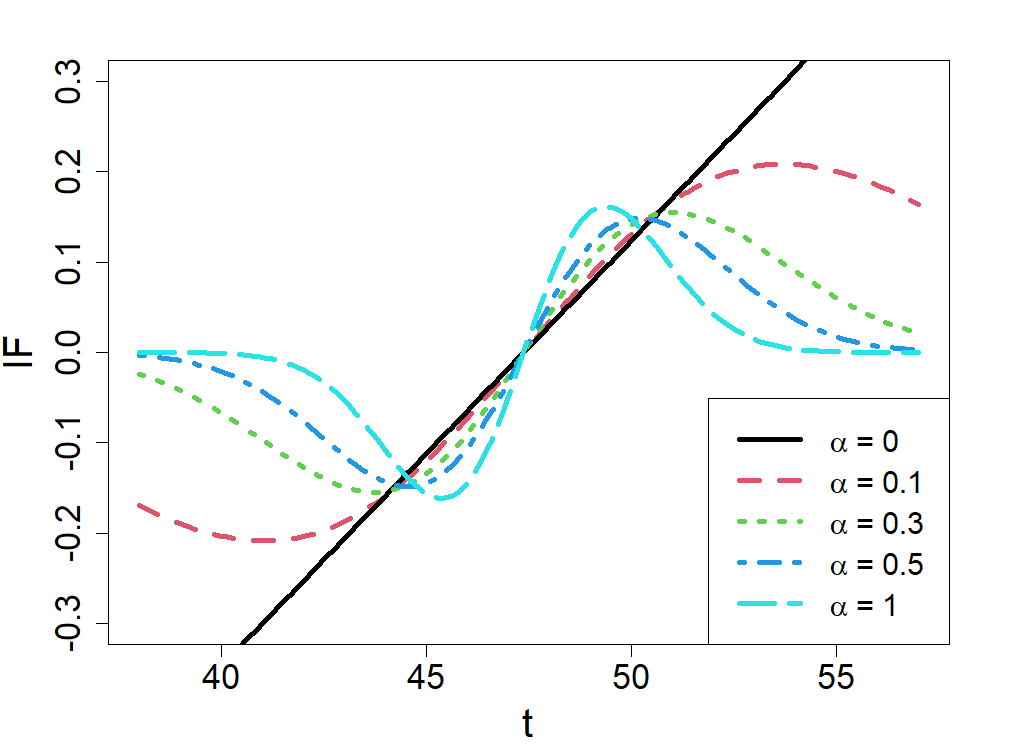}
         \caption{IF of $T_\alpha^{\beta_1}$ with $i_0 = 45$}
         \label{fig:IF-S2-45-th1}
     \end{subfigure}
     \hfill
     \begin{subfigure}[b]{0.45\textwidth}
         \centering
         \includegraphics[width=\textwidth]{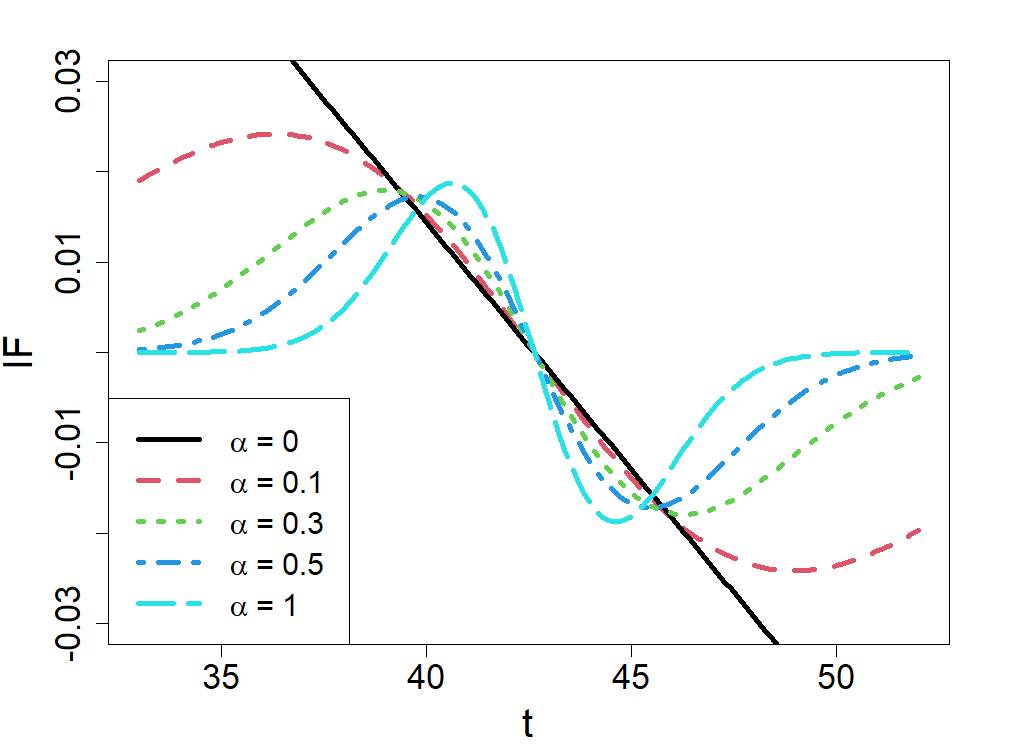}
         \caption{IF of $T_\alpha^{\beta_2}$ with $i_0 = 15$}
         \label{fig:IF-S2-15-th2}
     \end{subfigure}
     \hfill
     \begin{subfigure}[b]{0.45\textwidth}
         \centering
         \includegraphics[width=\textwidth]{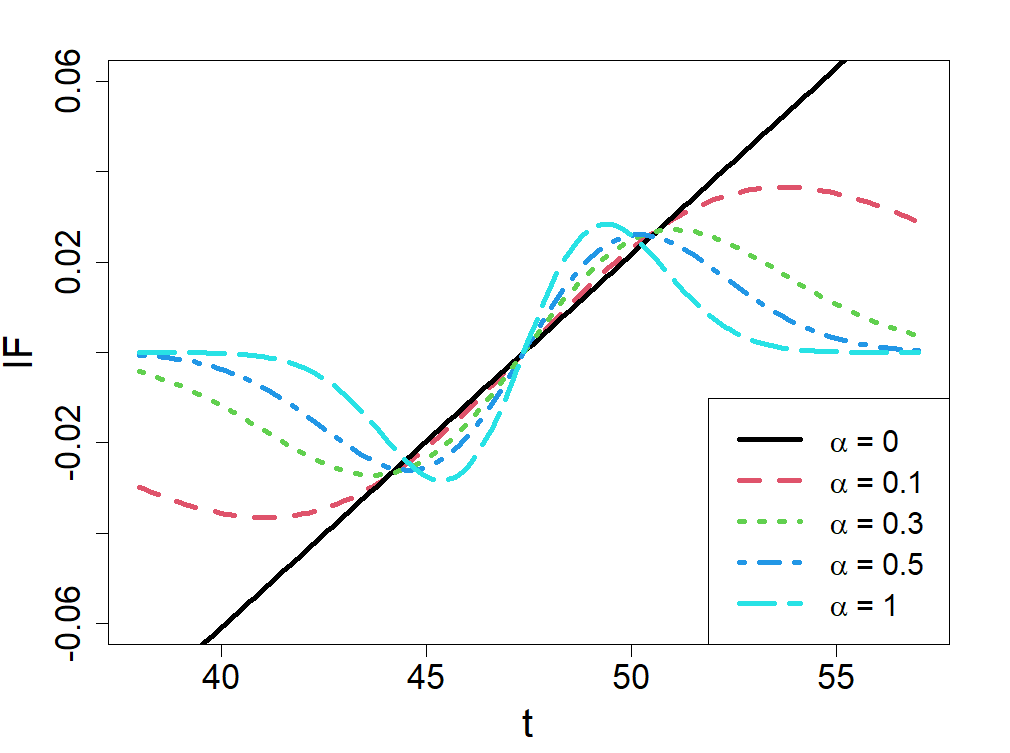}
         \caption{IF of $T_\alpha^{\beta_2}$ with $i_0 = 45$}
         \label{fig:IF-S2-45-th2}
     \end{subfigure}
     \hfill
     \begin{subfigure}[b]{0.45\textwidth}
         \centering
         \includegraphics[width=\textwidth]{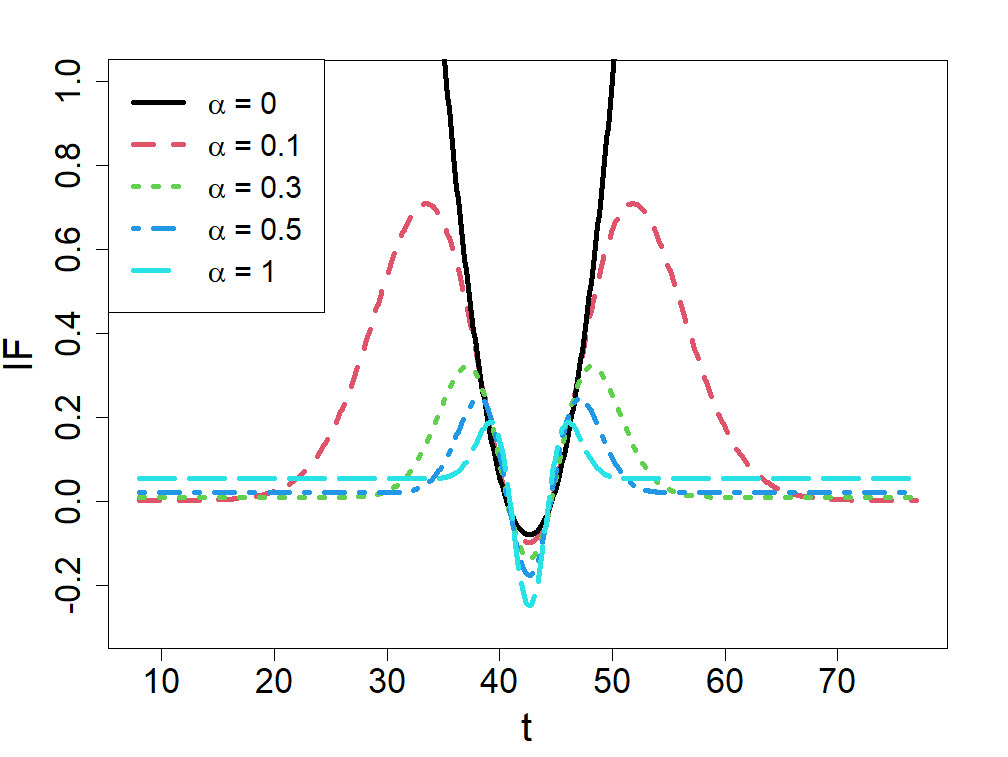}
         \caption{IF of $T_\alpha^{\sigma^2}$ with $i_0 = 15$}
         \label{fig:IF-S2-15-sig}
     \end{subfigure}
     \hfill
     \begin{subfigure}[b]{0.45\textwidth}
         \centering
         \includegraphics[width=\textwidth]{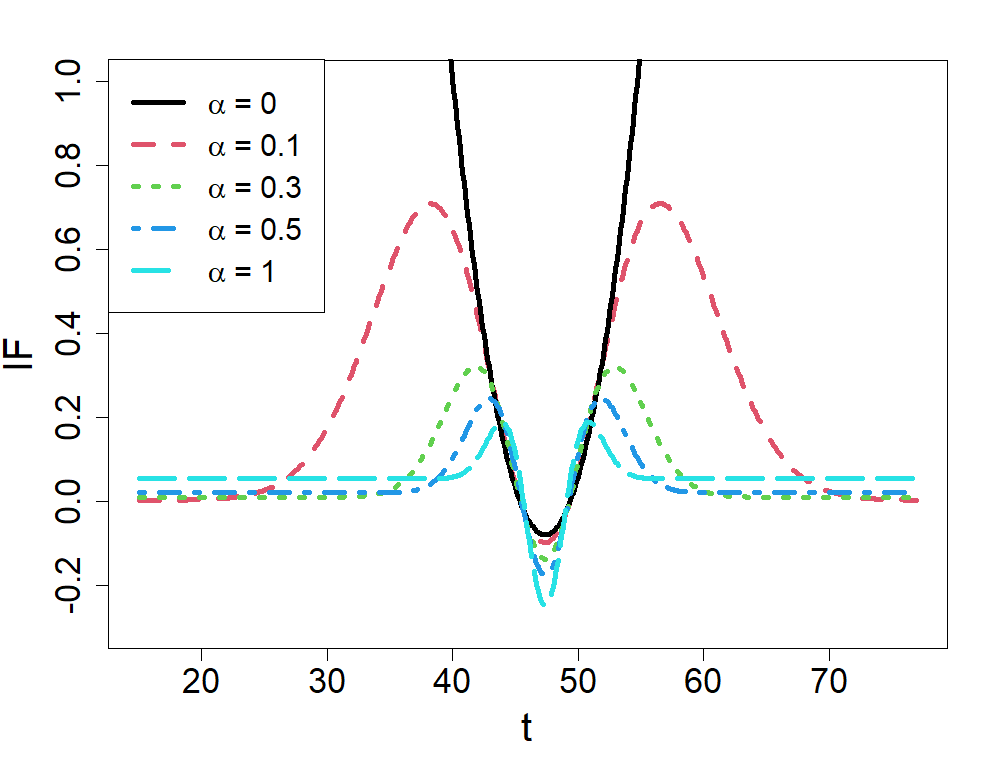}
         \caption{IF of $T_\alpha^{\sigma^2}$ with $i_0 = 45$}
         \label{fig:IF-S2-45-sig}
     \end{subfigure}
     \caption{\btxt{IFs of the MDPDEs of the parameters $\beta_1$, $\beta_2$, and $\sigma^2$ in the MM model, for contamination only in one observation, corresponding to the $i_0$-th order statistic of the covariate $x_i$s, under Setup 2}}
     \label{fig:IF-S2-15-45}
\end{figure}

\subsection{Asymptotic efficiency and contiguous power}
The AREs of different robust estimators of $\bm{\beta}$ under the MM model are presented in Table \ref{tab:ARETable}, which are, in fact, also applicable for general NLR models (since the MDPDEs for all these estimators are independent of the choice of non-linear mean function). For the proposed MDPDEs, these AREs are computed as per the formula given in Section 2.2 of the main manuscript, and the same for the remaining robust estimators are computed as per the discussions in Section \ref{ExistingRM}. The results show that the ARE of the MDPDE decreases with increasing $\alpha$, but the loss is very small for the smaller values of $\alpha$, as compared to other robust estimators at suitable tuning parameter values required to produce equivalent robustness.

Table \ref{tab:contpower} presents the asymptotic contiguous power of the MDPDE based Wald-type tests for testing the simple null hypothesis $H_0$, given in (25) of the main manuscript, against the one sided contiguous alternative of the form $H_{1,n}: \beta_k = \beta_k^0 + n^{-1/2}d,\ d>0$; these are computed based on the formula given by (27) of the main manuscript. As noted there, these shows that, for any fixed $\alpha$, the power increases with $d^*$; for a fixed $d^*$, the power decreases with $\alpha$, but the loss in power is not quite significant.

\begin{table}[H]
	\centering
	\caption{AREs of different estimators of $\bm{\beta}$ of the NLR model at several values of the associated tuning parameters ($c$, $w$ or $\alpha$)}
	\begin{tabular}{|c|c|c|c|c|c|c|c|}
		\hline
		\multicolumn{2}{|c}{Huber's M-est} & \multicolumn{2}{|c}{Tukey's M-est} & \multicolumn{2}{|c}{KPS est} & \multicolumn{2}{|c|}{MDPDE} \\ \hline
		c & ARE & c & ARE & w & ARE & $\alpha$ & ARE \\ \hline
		1.000 & 0.466 & 3.000 & 0.776 & 0.5 & 0.911 & 0.1 & 0.988 \\
		1.345 & 0.589 & 3.500 & 0.858 & 1.0 & 0.647 & 0.3 & 0.921 \\
		1.700 & 0.710 & 4.000 & 0.910 & 1.3 & 0.503 & 0.5 & 0.838 \\
		2.000 & 0.802 & 4.685 & 0.950 & 1.5 & 0.424 & 0.7 & 0.757 \\
		2.500 & 0.916 & 5.500 & 0.973 & 1.7 & 0.357 & 1.0 & 0.650 \\ \hline
	\end{tabular}
	\label{tab:ARETable}
\end{table}
\clearpage

\begin{table}[!h]
    \centering
    \caption{Asymptotic contiguous power of the Wald-type test for testing simple null hypothesis with one-sided alternatives as discussed in Section 3.3 of the main paper}
    \begin{tabular}{|r|rrrrrr|}
    \hline
        \multirow{2}{*}{$d^*$} & \multicolumn{6}{c|}{$\alpha$}  \\ \cline{2-7}
        & 0 & 0.1 & 0.3 & 0.5 & 0.7 & 1 \\ \hline
        0 & 0.050 & 0.050 & 0.050 & 0.050 & 0.050 & 0.050 \\
        1 & 0.260 & 0.258 & 0.247 & 0.233 & 0.219 & 0.201 \\
        2 & 0.639 & 0.634 & 0.608 & 0.574 & 0.538 & 0.487 \\
        3 & 0.912 & 0.909 & 0.891 & 0.865 & 0.833 & 0.780 \\
        5 & 1.000 & 1.000 & 0.999 & 0.998 & 0.997 & 0.991 \\
        7 & 1.000 & 1.000 & 1.000 & 1.000 & 1.000 & 1.000 \\ \hline
    \end{tabular}
    \label{tab:contpower}
\end{table}

\subsection{\btxt{Runtime comparison}}
\btxt{To compare the computation time required by the robust estimators discussed throughout this paper, we have conducted the simulation study under Setup 2 from Section 4.3 of the main paper, again with $e_c = 20\%$ contamination in both $x$-$y$ direction, and compute all these estimators with increasing sample sizes $n=25,50,\ldots,250$. For each sample size, we fitted the MM model to the simulated data with each of the robust estimators -- MDPDE with different $\alpha$, KPS with $w=1$, Huber's M-estimator with $c=1.345$, and Tukey's M estimator with $c = 4.685$. We replicated this process $1000$ times and plotted the average time (in seconds) required by each method in Figure \ref{fig:time}. The plot clearly demonstrates that the computation of MDPDEs takes the lowest time compared to the other existing robust methods (with their available algorithms) and scales much slowly with increasing sample sizes. This makes our proposal feasible for even larger sample sizes as well.}
\begin{figure}[H]
     \centering
     \begin{subfigure}[b]{0.48\textwidth}
         \centering
         \includegraphics[width=\textwidth]{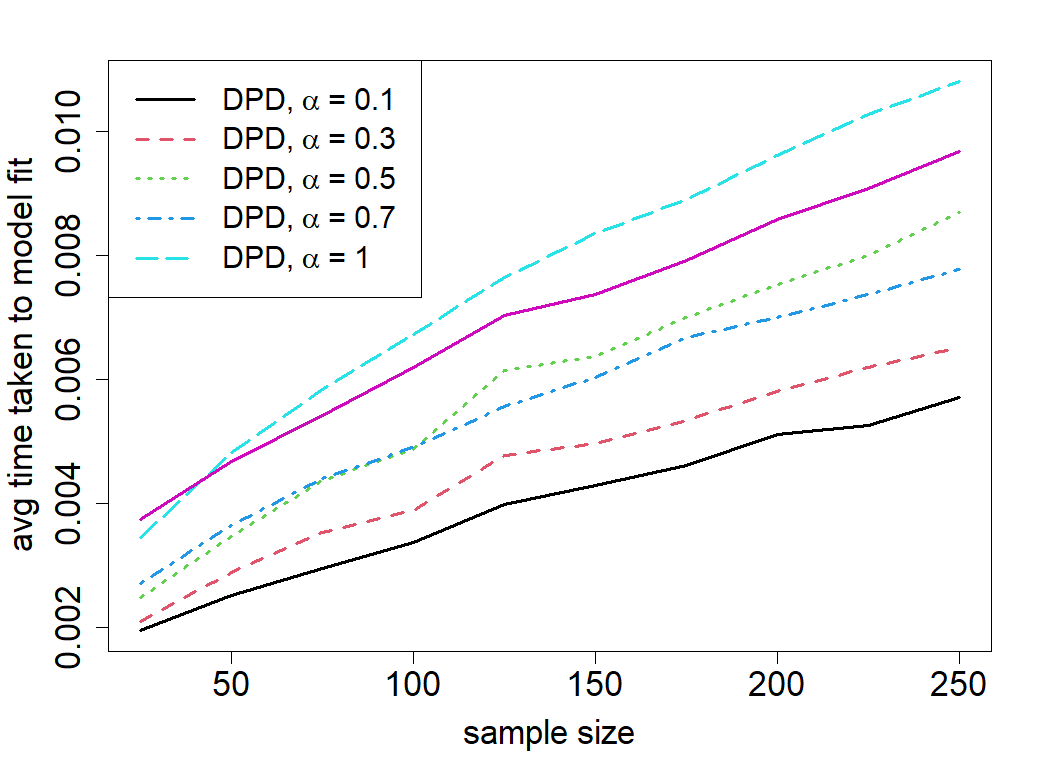}
    \caption{\btxt{Only the MDPDEs with different $\alpha$}}
    \label{fig:time-dpd}
     \end{subfigure}
     \hfill
     \begin{subfigure}[b]{0.48\textwidth}
         \centering
         \includegraphics[width=\textwidth]{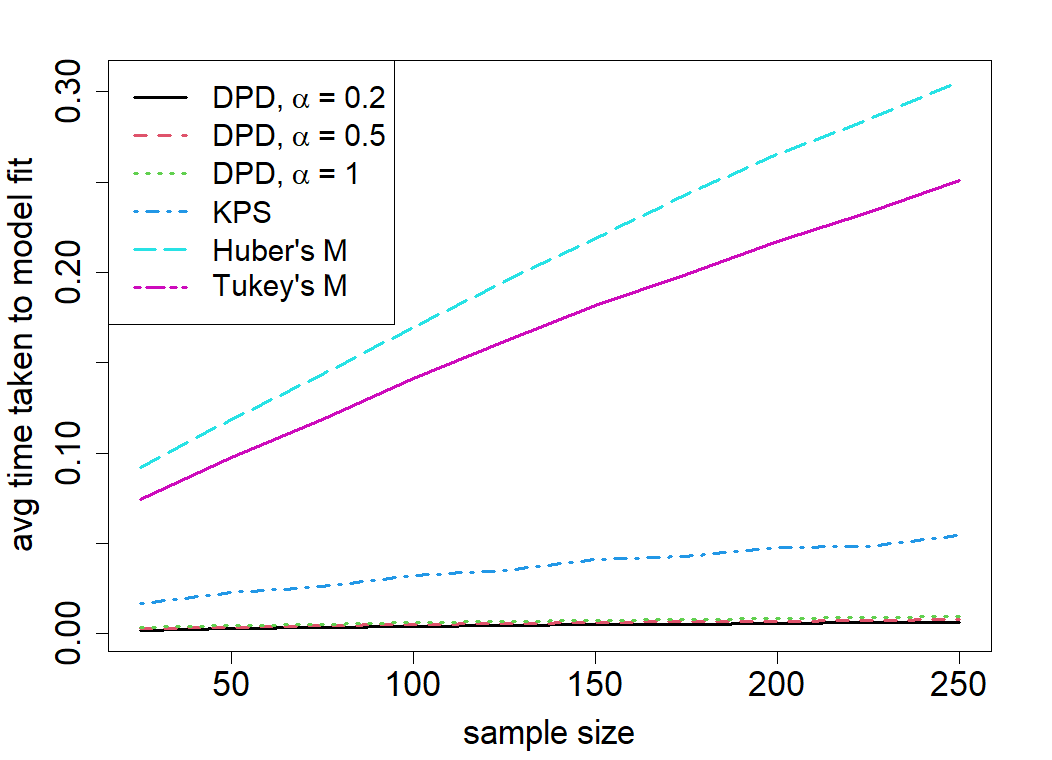}
    \caption{\btxt{The MDPDEs with its competitors}}
    \label{fig:time-comp}
     \end{subfigure}
     \caption{\btxt{Plot of average computation time (in seconds) required by the robust estimators at various sample sizes}}
     \label{fig:time}
\end{figure}

\subsection{Further results from  simulation studies}
Tables \ref{tab:S1BMAppendix} and \ref{tab:S2BMAppendix} present the EBias and EMSE results of the MDPDE and its robust competitors under the MM model, based on 1000 replications for contaminated data with contamination proportions $e_c = 20\%$ and $30\%$; associated average prediction errors are given in Table \ref{tab:S12APEAppendix}. Table \ref{tab:S1BM} presents the similar EBias and EMSE results, under Setup 1, for the pure data, as well as the contaminated data with contamination proportions $e_c = 10\%$ and $40\%$.

\begin{landscape}
\begin{table}
    \centering
    \caption{Empirical absolute bias and empirical MSE obtained by different estimators in simulations with Setup 1 and different amounts of contamination (the minimum values obtained for each parameter are highlighted in bold font)}
    \resizebox{1.5\textwidth}{!}
    {\begin{tabular}{|l|rrr|rrr|rrr|rrr|rrr|}
        \hline
        Outlier Direction & \multicolumn{3}{c|}{No outlier (pure data)} & \multicolumn{6}{c|}{Contamination only in response} & \multicolumn{6}{c|}{Contamination in both response and covariate} \\
        Outlier proportion & ~ & ~&~ & \multicolumn{3}{c|}{10\%} & \multicolumn{3}{c|}{40\%} & \multicolumn{3}{c|}{10\%} & \multicolumn{3}{c|}{40\%} \\ \hline
        ~ & $\beta_1$  & $\beta_2$  & $\sigma$  & $\beta_1$ & $\beta_2$  & $\sigma$ & $\beta_1$ & $\beta_2$ & $\sigma$ & $\beta_1$ & $\beta_2$  & $\sigma$ & $\beta_1$ & $\beta_2$  & $\sigma$ \\ \hline\hline
        \multicolumn{16}{|l|}{\textbf{\underline{Empirical Bias}}} \\
        \multicolumn{16}{|l|}{}\\
        OLS & \textbf{0.0111} & 0.0681 & ~ & 3.1515 & 4.1252 & ~ & 9.2057 & 2.1921 & ~ & 68.96 & 365.49 & ~ & 36.35 & 119.67 & ~ \\
        MOM & 0.0814 & 0.1973 & ~ & 85.37 & 1.1631 & ~ & 59.26 & 6.7111 & ~ & 0.4602 & 205.45 & ~ & 7.1474 & 363.16 & ~ \\
        KPS(w=0.5) & 0.0143 & 0.0728 & ~ & 0.0083 & 0.0466 & ~ & 0.3470 & 0.7564 & ~ & 0.0184 & 0.0734 & ~ & 31.80 & 116.04 & ~ \\
        KPS(w=1) & 0.0256 & 0.1087 & ~ & 0.0200 & 0.0964 & ~ & 0.0300 & 0.1732 & ~ & 0.0220 & 0.1006 & ~ & 0.0789 & 0.1252 & ~ \\
        KPS(w=1.5) & 0.0533 & 0.2256 & ~ & 0.0394 & 0.1717 & ~ & 0.0036 & 0.0036 & ~ & 0.0410 & 0.1767 & ~ & 0.0276 & 0.0815 & ~ \\
        M-Est (Huber) & 0.0141 & 0.0809 & ~ & 0.2059 & 0.0326 & ~ & 4.4738 & 1.2853 & ~ & 0.4880 & 1.0417 & ~ & 11.23 & 9.0000 & ~ \\
        M-Est (Tukey) & 0.0135 & 0.0800 & ~ & 0.0104 & 0.0682 & ~ & 2.5193 & 1.4471 & ~ & 0.0137 & 0.0778 & ~ & 11.25 & 8.9182 & ~ \\ \hline
        MDPDE($\alpha$=0.05) & \textbf{0.0110} & \textbf{0.0677} & 0.0232 & 0.4978 & 0.0948 & 1.3548 & 7.7955 & 0.4942 & 8.6529 & 12.79 & 63.72 & 1.0553 & 31.68 & 96.26 & 4.9905 \\
        MDPDE($\alpha$=0.1) & 0.0112 & 0.0690 & 0.0224 & \textbf{0.0010} & \textbf{0.0271} & 0.0037 & 2.3448 & 1.1511 & 4.0321 & 0.0173 & \textbf{0.0721} & \textbf{0.0003} & 15.69 & 51.07 & 2.3681 \\
        MDPDE($\alpha$=0.2) & 0.0120 & 0.0723 & \textbf{0.0223} & 0.0056 & 0.0538 & 0.0106 & 0.0156 & 0.0208 & \textbf{0.0957} & \textbf{0.0124}& 0.0733 & 0.0098 & 0.2449 & 0.8665 & \textbf{0.1055} \\
        MDPDE($\alpha$=0.3) & 0.0130 & 0.0766 & 0.0238 & 0.0098 & 0.0704 & 0.0064 & \textbf{0.0063} & \textbf{0.0496} & 0.1081 & 0.0128 & 0.0782 & 0.0058 & 0.0634 & 0.2640 & 0.1095 \\
        MDPDE($\alpha$=0.5) & 0.0156 & 0.0861 & 0.0299 & 0.0143 & 0.0883 & \textbf{0.0027} & 0.0122 & 0.0851 & 0.2053 & 0.0153 & 0.0897 & 0.0032 & \textbf{0.0267} & \textbf{0.1263} & 0.2087 \\
        MDPDE($\alpha$=0.7) & 0.0195 & 0.0986 & 0.0384 & 0.0179 & 0.1006 & 0.0106 & 0.0173 & 0.1032 & 0.3147 & 0.0189 & 0.1019 & 0.0109 & 0.0274 & 0.1312 & 0.3180 \\
        MDPDE($\alpha$=1) & 0.0272 & 0.1259 & 0.0514 & 0.0243 & 0.1226 & 0.0198 & 0.0210 & 0.1174 & 0.4614 & 0.0248 & 0.1226 & 0.0200 & 0.0287 & 0.1377 & 0.4646 \\ \hline \hline
        \multicolumn{16}{|l|}{\textbf{\underline{Empirical MSE}}} \\
        \multicolumn{16}{|l|}{}\\
        OLS & \textbf{0.0421} & \textbf{0.2790} & ~ & 68.24 & 937.06 & ~ & 131.39 & 221.49 & ~ & 5342 & 149877 & ~ & 1358 & 15185 & ~ \\
        MOM & 0.1283 & 0.9449 & ~ & 7.1E+06 & 33.18 & ~ & 2.2E+06 & 533.73 & ~ & 0.9885 & 50252 & ~ & 63.87 & 1.2E+7 & ~ \\
        KPS(w=0.5) & 0.0483 & 0.3253 & ~ & 0.0575 & 0.4145 & ~ & 1.0585 & 2.6777 & ~ & 0.0547 & 0.3925 & ~ & 1376 & 18957 & ~ \\
        KPS(w=1) & 0.0677 & 0.4583 & ~ & 0.0691 & 0.4915 & ~ & 0.1233 & 0.9721 & ~ & 0.0681 & 0.4817 & ~ & 0.1185 & 0.9277 & ~ \\
        KPS(w=1.5) & 0.1121 & 0.9119 & ~ & 0.0977 & 0.7668 & ~ & 0.1095 & 0.8442 & ~ & 0.0974 & 0.7585 & ~ & 0.1028 & 0.8029 & ~ \\
        M-Est (Huber) & 0.0446 & 0.3071 & ~ & 0.1032 & 0.4390 & ~ & 31.11 & 10.15 & ~ & 0.4239 & 3.7898 & ~ & 127.64 & 81.00 & ~ \\
        M-Est (Tukey) & 0.0447 & 0.3084 & ~ & 0.0504 & 0.3715 & ~ & 27.40 & 11.60 & ~ & 0.0494 & 0.3587 & ~ & 129.72 & 80.25 & ~ \\ \hline
        MDPDE($\alpha$=0.05) & 0.0423 & 0.2808 & \textbf{0.0100} & 0.9278 & 1.9099 & 5.0569 & 94.07 & 160.72 & 77.42 & 473.90 & 11987  & 2.8641 & 1041  & 9818  & 25.92 \\
        MDPDE($\alpha$=0.1) & 0.0426 & 0.2842 & 0.0101 & 0.0538 & 0.4031 & 0.0144 & 59.55  & 463.88 & 38.60 & \textbf{0.0492} & \textbf{0.3540} & 0.0150 & 532.03 & 5752 & 12.40 \\
        MDPDE($\alpha$=0.2) & 0.0438 & 0.2958 & 0.0107 & \textbf{0.0525} & \textbf{0.3904} & \textbf{0.0124} & 0.1546 & \textbf{0.7214} & 0.3195 & 0.0499 & 0.3605 & \textbf{0.0126} & 6.6965 & 84.18 & 0.2012 \\
        MDPDE($\alpha$=0.3) & 0.0457 & 0.3142 & 0.0116 & 0.0533 & 0.3928 & 0.0132 & 0.0953 & 0.7249 & 0.0961 & 0.0518 & 0.3771 & 0.0133 & 1.3977 & 21.98 & \textbf{0.0552} \\
        MDPDE($\alpha$=0.5) & 0.0513 & 0.3641 & 0.0140 & 0.0576 & 0.4302 & 0.0158 & \textbf{0.0906} & 0.7270 & \textbf{0.0739} & 0.0569 & 0.4219 & 0.0159 & \textbf{0.0875} & \textbf{0.6895} & 0.0757 \\
        MDPDE($\alpha$=0.7) & 0.0577 & 0.4189 & 0.0169 & 0.0634 & 0.4804 & 0.0189 & 0.0922 & 0.7410 & 0.1373 & 0.0627 & 0.4717 & 0.0189 & 0.0901 & 0.7245 & 0.1397 \\
        MDPDE($\alpha$=1)   & 0.0695 & 0.5153 & 0.0214 & 0.0730 & 0.5654 & 0.0230 & 0.0946 & 0.7683 & 0.2605 & 0.0726 & 0.5597 & 0.0230 & 0.0934 & 0.7629 & 0.2637 \\ \hline
    \end{tabular}}
    \label{tab:S1BM}
\end{table} 
\end{landscape}

\begin{landscape}
\begin{table}
    \centering
    \caption{Empirical absolute bias and empirical MSE obtained by different estimators in simulations with Setup 1 and different amounts of contamination (the minimum values obtained for each parameter are highlighted in bold font)}
\resizebox{1.3\textheight}{!}
{  \begin{tabular}{|l|rrr|rrr|rrr|rrr|rrr|}
        \hline
        Outlier Direction & \multicolumn{6}{c|}{Contamination only in response} & \multicolumn{6}{c|}{Contamination in both response and covariate} \\
        Outlier proportion & \multicolumn{3}{c|}{20\%} & \multicolumn{3}{c|}{30\%} & \multicolumn{3}{c|}{20\%} & \multicolumn{3}{c|}{30\%} \\ \hline
        ~ & $\beta_1$ & $\beta_2$  & $\sigma$ & $\beta_1$ & $\beta_2$ & $\sigma$ & $\beta_1$ & $\beta_2$  & $\sigma$ & $\beta_1$ & $\beta_2$  & $\sigma$ \\ \hline\hline
        \multicolumn{13}{|l|}{\textbf{\underline{Empirical Bias}}} \\ 
        \multicolumn{13}{|l|}{} \\
        OLS & 5.9305 & 5.4389 & ~ & 7.3261 & 2.8823 & ~ & 53.12 & 243.45 & ~ & 42.21 & 163.89 & ~ \\
        MOM & 4.1274 & 3.9194 & ~ & 124.75 & 7.3124 & ~ & 3.1807 & 221.13 & ~ & 5.1049 & 356.55 & ~ \\
        KPS(w=0.5) & 0.0132 & 0.0094 & ~ & 0.0221 & 0.1983 & ~ & 0.0521 & 0.1183 & ~ & 1.1789 & 4.7107 & ~ \\
        KPS(w=1) & 0.0234 & 0.1011 & ~ & 0.0115 & 0.0457 & ~ & 0.0287 & 0.1117 & ~ & 0.0314 & 0.0955 & ~ \\
        KPS(w=1.5) & 0.0388 & 0.1583 & ~ & 0.0247 & 0.1038 & ~ & 0.0405 & 0.1605 & ~ & 0.0335 & 0.1268 & ~ \\
        M-Est (Huber) & 0.5502 & 0.0696 & ~ & 1.3424 & 0.3195 & ~ & 3.1837 & 8.8381 & ~ & 5.7900 & 9.0000 & ~ \\
        M-Est (Tukey) & 0.0192 & 0.0855 & ~ & 0.0248 & 0.0710 & ~ & 0.0255 & 0.1010 & ~ & 0.5843 & 1.0422 & ~ \\ \hline
        MDPDE($\alpha$=0.05) & 2.9360 & 0.5796 & 5.0024 & 5.5897 & 0.9690 & 7.4614 & 29.49 & 123.06 & 2.9775 & 32.75 & 116.42 & 4.2146 \\
        MDPDE($\alpha$=0.1) & 0.0941 & \textbf{0.0249} & 0.2047 & 0.5167 & 0.9748 & 1.6376 & 1.6620 & 7.5578 & 0.1757 & 8.1187 & 31.41 & 0.9716 \\
        MDPDE($\alpha$=0.2) & \textbf{0.0123} & 0.0701 & \textbf{0.0049} & \textbf{0.0103} & \textbf{0.0496} & \textbf{0.0337} & 0.0225 & \textbf{0.0977} & \textbf{0.0071} & 0.0765 & 0.3353 & \textbf{0.0347} \\
        MDPDE($\alpha$=0.3) & 0.0168 & 0.0894 & 0.0171 & 0.0125 & 0.0746 & 0.0505 & \textbf{0.0222} & 0.1036 & 0.0186 & \textbf{0.0247} & \textbf{0.1083} & 0.0529 \\
        MDPDE($\alpha$=0.5) & 0.0228 & 0.1145 & 0.0475 & 0.0198 & 0.1028 & 0.1097 & 0.0254 & 0.1201 & 0.0485 & 0.0261 & 0.1201 & 0.1116 \\
        MDPDE($\alpha$=0.7) & 0.0272 & \textbf{}0.1313 & 0.0774 & 0.0243 & 0.1214 & 0.1704 & 0.0290 & 0.1348 & 0.0783 & 0.0285 & 0.1315 & 0.1721 \\
        MDPDE($\alpha$=1) & 0.0335 & 0.1537 & 0.1156 & 0.0291 & 0.1411 & 0.2499 & 0.0342 & 0.1539 & 0.1162 & 0.0318 & 0.1460 & 0.2515 \\ \hline\hline
        \multicolumn{13}{|l|}{\textbf{\underline{Empirical MSE}}} \\ 
        \multicolumn{13}{|l|}{} \\
        OLS & 178.70 & 1400 & ~ & 85.89 & 196.85 & ~ & 3136 & 67920 & ~ & 1866 & 29045 & ~ \\
        MOM & 49.79 & 612.19 & ~ & 6.8E+06 & 1737.08 & ~ & 12.67 & 136261 & ~ & 55.56 & 5.7E+06 & ~ \\
        KPS(w=0.5) & 0.0712 & 0.5051 & ~ & 0.1184 & 0.8523 & ~ & 0.0700 & 0.4974 & ~ & 50.62 & 1037 & ~ \\
        KPS(w=1) & 0.0757 & 0.5376 & ~ & 0.0864 & 0.6142 & ~ & 0.0740 & 0.5260 & ~ & 0.0822 & 0.5970 & ~ \\
        KPS(w=1.5) & 0.0941 & 0.7241 & ~ & 0.0971 & 0.7193 & ~ & 0.0943 & 0.7115 & ~ & 0.0935 & 0.6786 & ~ \\
        M-Est (Huber) & 0.3977 & 0.5981 & ~ & 2.4094 & 2.2897 & ~ & 10.40 & 79.48 & ~ & 34.58 & 81.00 & ~ \\
        M-Est (Tukey) & 0.0579 & 0.4184 & ~ & 0.0699 & 0.5068 & ~ & 0.0569 & 0.4100 & ~ & 3.5212 & 8.6970 & ~ \\ \hline
        MDPDE($\alpha$=0.05) & 18.62 & 96.30 & 32.87 & 40.39 & 23.60 & 60.56 & 1150 & 20448 & 11.73 & 1175 & 15229 & 19.55 \\
        MDPDE($\alpha$=0.1) & 0.3157 & 0.6118 & 1.0584 & 26.06 & 266.19 & 13.62 & 58.44 & 1237 & 0.5567 & 296.85 & 4527 & 4.2317 \\
        MDPDE($\alpha$=0.2) & \textbf{0.0612} & \textbf{0.4511} & \textbf{0.0148} & 0.0818 & 0.5559 & 0.0578 & \textbf{0.0582} & \textbf{0.4177} & \textbf{0.0151} & 2.5361 & 54.08 & 0.0307 \\
        MDPDE($\alpha$=0.3) & 0.0623 & 0.4568 & 0.0158 & \textbf{0.0743} & \textbf{0.5525} & \textbf{0.0221} & 0.0601 & 0.4336 & 0.0160 & \textbf{0.0707} & \textbf{0.5117} & \textbf{0.0226} \\
        MDPDE($\alpha$=0.5) & 0.0666 & 0.4946 & 0.0209 & 0.0772 & 0.5721 & 0.0355 & 0.0648 & 0.4791 & 0.0210 & 0.0748 & 0.5493 & 0.0362 \\
        MDPDE($\alpha$=0.7) & 0.0716 & 0.5481 & 0.0280 & 0.0808 & 0.6101 & 0.0570 & 0.0701 & 0.5346 & 0.0282 & 0.0792 & 0.5952 & 0.0578 \\
        MDPDE($\alpha$=1) & 0.0793 & 0.6320 & 0.0402 & 0.0863 & 0.6760 & 0.0965 & 0.0784 & 0.6238 & 0.0404 & 0.0848 & 0.6623 & 0.0976 \\ \hline
    \end{tabular}}
    \label{tab:S1BMAppendix}
\end{table}
\end{landscape}

\begin{landscape}
\begin{table}
    \centering
    \caption{Empirical absolute bias and empirical MSE obtained by different estimators in simulations with Setup 2 and different amounts of contamination (the minimum values obtained for each parameter are highlighted in bold font)}
 \resizebox{1.3\textheight}{!}
 {   \begin{tabular}{|l|rrr|rrr|rrr|rrr|rrr|}
        \hline
        Outlier Direction & \multicolumn{6}{c|}{Contamination only in response} & \multicolumn{6}{c|}{Contamination in both response and covariate} \\
        Outlier proportion & \multicolumn{3}{c|}{20\%} & \multicolumn{3}{c|}{30\%} & \multicolumn{3}{c|}{20\%} & \multicolumn{3}{c|}{30\%} \\ \hline
        ~ & $\beta_1$ & $\beta_2$  & $\sigma$ & $\beta_1$ & $\beta_2$ & $\sigma$ & $\beta_1$ & $\beta_2$  & $\sigma$ & $\beta_1$ & $\beta_2$  & $\sigma$ \\ \hline \hline
        \multicolumn{13}{|l|}{\textbf{\underline{Empirical Bias}}} \\
        \multicolumn{13}{|l|}{} \\
        OLS & 10.54 & 0.2015 & ~ & 15.72 & 0.2438 & ~ & 28.55 & 7.6917 & ~ & 41.55 & 11.09 & ~ \\
        MOM & 8.3100 & 0.4567 & ~ & 15.99 & 0.7324 & ~ & 28.70 & 6.6655 & ~ & 47.51 & 11.22 & ~ \\
        KPS(w=0.5) & 0.0666 & 0.0289 & ~ & 0.2759 & 0.1323 & ~ & 0.0822 & 0.0411 & ~ & 0.1605 & 0.1090 & ~ \\
        KPS(w=1) & \textbf{0.0003} & 0.0016 & ~ & 0.0134 & 0.0109 & ~ & 0.0180 & 0.0094 & ~ & 0.0575 & 0.0309 & ~ \\
        KPS(w=1.5) & 0.0186 & 0.0042 & ~ & 0.0198 & 0.0014 & ~ & 0.0050 & \textbf{0.0011} & ~ & 0.0125 & 0.0117 & ~ \\
        M-Est (Huber) & 0.8973 & 0.0919 & ~ & 2.2299 & 0.1811 & ~ & 4.8576 & 1.0918 & ~ & 27.05 & 5.8772 & ~ \\
        M-Est (Tukey) & 0.0685 & 0.0031 & ~ & 0.2401 & 0.0128 & ~ & 3.5690 & 0.9716 & ~ & 26.05 & 5.7811 & ~ \\ \hline
        MDPDE($\alpha$=0.05) & 2.3100 & 0.1031 & 4.4833 & 7.5019 & 0.0088 & 9.4953 & 10.28 & 2.4025 & 6.1701 & 28.37 & 7.3262 & 10.83 \\
        MDPDE($\alpha$=0.1) & 0.0403 & 0.0292 & 0.0835 & 0.2805 & 0.0588 & 0.6703 & 0.0607 & 0.0487 & 0.1362 & 2.1178 & 0.4127 & 1.3049 \\
        MDPDE($\alpha$=0.2) & 0.0118 & 0.0059 & \textbf{0.0054} & 0.0125 & 0.0108 & \textbf{0.0601} & 0.0422 & 0.0199 & \textbf{0.0225} & 0.0778 & 0.0410 & \textbf{0.0951} \\
        MDPDE($\alpha$=0.3) & 0.0033 & 0.0022 & 0.0276 & \textbf{0.0053} & 0.0035 & 0.1037 & 0.0251 & 0.0120 & 0.0368 & 0.0418 & 0.0243 & 0.1232 \\
        MDPDE($\alpha$=0.5) & \textbf{0.0014} & \textbf{0.0003} & 0.0841 & 0.0144 & \textbf{1.3E-05} & 0.2200 & 0.0131 & 0.0067 & 0.0873 & 0.0165 & 0.0131 & 0.2280 \\
        MDPDE($\alpha$=0.7) & 0.0040 & 0.0007 & 0.1416 & 0.0192 & 0.0019 & 0.3398 & 0.0075 & 0.0043 & 0.1418 & 0.0077 & 0.0094 & 0.3419 \\
        MDPDE($\alpha$=1) & 0.0073 & 0.0019 & 0.2163 & 0.0221 & 0.0030 & 0.4973 & \textbf{0.0027} & \textbf{0.0024} & 0.2136 & \textbf{0.0010} & \textbf{0.0067} & 0.4942 \\ \hline \hline
        \multicolumn{13}{|l|}{\textbf{\underline{Empirical MSE}}} \\ 
        \multicolumn{13}{|l|}{} \\
        OLS & 127.86 & 1.2657 & ~ & 268.46 & 1.3927 & ~ & 1026.58 & 104.58 & ~ & 1993.45 & 181.41 & ~ \\
        MOM & 154.28 & 3.5655 & ~ & 383.09 & 6.8593 & ~ & 1324.69 & 119.51 & ~ & 2784.57 & 227.17 & ~ \\
        KPS(w=0.5) & 0.4331 & 0.0375 & ~ & 0.8790 & 0.1089 & ~ & 0.4197 & 0.0382 & ~ & 0.5939 & 0.0709 & ~ \\
        KPS(w=1) & 0.4821 & 0.0392 & ~ & 0.5435 & 0.0437 & ~ & 0.4766 & 0.0391 & ~ & 0.5443 & 0.0469 & ~ \\
        KPS(w=1.5) & 0.6142 & 0.0497 & ~ & 0.6296 & 0.0495 & ~ & 0.6076 & 0.0495 & ~ & 0.6297 & 0.0518 & ~ \\
        M-Est (Huber) & 1.7325 & 0.1060 & ~ & 12.83 & 0.4652 & ~ & 91.42 & 6.7782 & ~ & 912.82 & 45.14 & ~ \\
        M-Est (Tukey) & 2.9843 & 0.0343 & ~ & 9.4442 & 0.2652 & ~ & 104.27 & 7.8639 & ~ & 892.73 & 44.69 & ~ \\ \hline
        MDPDE($\alpha$=0.05) & 20.97 & 0.2398 & 65.12 & 114.42 & 0.5666 & 168.60 & 290.24 & 25.46 & 79.64 & 1291.37 & 113.64 & 161.66 \\
        MDPDE($\alpha$=0.1) & 0.4586 & 0.0479 & 0.8740 & 4.6137 & 0.0825 & 10.33 & 0.8948 & 0.0521 & 0.8054 & 81.54 & 6.0233 & 17.80 \\
        MDPDE($\alpha$=0.2) & \textbf{0.4023} & \textbf{0.0346} & \textbf{0.0593} & 0.4903 & 0.0429 & \textbf{0.0828} & \textbf{0.4053} & \textbf{0.0357} & \textbf{0.0641} & 0.5233 & 0.0516 & \textbf{0.0985} \\
        MDPDE($\alpha$=0.3) & 0.4121 & 0.0350 & 0.0638 & \textbf{0.4872} & \textbf{0.0410} & 0.0910 & 0.4154 & 0.0359 & 0.0663 & \textbf{0.5150} & 0.0471 & 0.0997 \\
        MDPDE($\alpha$=0.5) & 0.4464 & 0.0382 & 0.0828 & 0.5156 & 0.0433 & 0.1449 & 0.4488 & 0.0386 & 0.0839 & 0.5284 & \textbf{0.0462} & 0.1489 \\
        MDPDE($\alpha$=0.7) & 0.4839 & 0.0416 & 0.1088 & 0.5435 & 0.0456 & 0.2289 & 0.4866 & 0.0420 & 0.1090 & 0.5559 & 0.0482 & 0.2300 \\
        MDPDE($\alpha$=1) & 0.5399 & 0.0465 & 0.1521 & 0.5835 & 0.0490 & 0.3834 & 0.5444 & 0.0471 & 0.1517 & 0.6008 & 0.0522 & 0.3794 \\ \hline
    \end{tabular}}
    \label{tab:S2BMAppendix}
\end{table}
\end{landscape}

\begin{table}[H]
	\centering
	\caption{Average prediction errors obtained by different estimators and different amounts of contamination (the minimum values obtained for each parameter are highlighted in bold font)}
	\resizebox{\textwidth}{!}
	{\begin{tabular}{|l||rr|rr||rr|rr|}
			\hline
			Setup & \multicolumn{4}{c||}{Setup 1} & \multicolumn{4}{c|}{Setup 2} \\ \hline
			Outlier direction & \multicolumn{2}{c|}{response} & \multicolumn{2}{c||}{response + covariate} & \multicolumn{2}{c|}{response} & \multicolumn{2}{c|}{response + covariate} \\ \hline
			Outlier Prop & 20\% & 30\% & 20\% & 30\% & 20\% & 30\% & 20\% & 30\% \\ \hline \hline
			OLS & 15.62 & 32.91 & 8.233 & 11.59 & 79.51 & 170.43 & 100.69 & 172.67 \\
			KPS(w=0.5) & 1.094 & 13.02 & 0.976 & 2.009 & 3.802 & 3.997 & 3.810 & 3.901 \\
			KPS(w=1) & 1.041 & 2.745 & 1.080 & 2.258 & 3.827 & 3.800 & 3.829 & 3.809 \\
			KPS(w=1.5) & 1.208 & 1.642 & 1.010 & 1.214 & 3.889 & 3.838 & 3.891 & 3.843 \\
			Huber & 1.221 & 2.408 & 2.784 & 9.485 & 5.008 & 11.15 & 14.18 & 127.95 \\
			Tukey & 0.947 & 0.941 & 0.947 & 1.829 & 5.930 & 5.548 & 14.81 & 124.12 \\ \hline
			MDPDE($\alpha$=0.05) & 12.93 & 30.06 & 9.498 & 13.37 & 21.05 & 82.61 & 39.45 & 120.84 \\
			MDPDE($\alpha$=0.1) & 10.37 & 27.32 & 8.156 & 11.74 & 3.965 & 7.798 & 4.006 & 14.10 \\
			MDPDE($\alpha$=0.2) & 1.500 & 21.77 & 3.735 & 9.261 & \textbf{3.788} & 3.785 & \textbf{3.797} & 3.810 \\
			MDPDE($\alpha$=0.3) & \textbf{0.953} & 4.785 & \textbf{0.950} & 5.183 & 3.794 & \textbf{3.782} & 3.799 & \textbf{3.797} \\
			MDPDE($\alpha$=0.5) & 0.956 & \textbf{0.950} & 0.954 & \textbf{0.946} & 3.811 & 3.795 & 3.815 & 3.804 \\
			MDPDE($\alpha$=0.7) & 0.962 & 0.951 & 0.961 & 0.949 & 3.829 & 3.808 & 3.833 & 3.817 \\
			MDPDE($\alpha$=1) & 0.972 & 0.956 & 0.971 & 0.955 & 3.856 & 3.827 & 3.861 & 3.837 \\ \hline
	\end{tabular}}
	\label{tab:S12APEAppendix}
\end{table}

\section{Additional Real data examples}
Here, we illustrate the performance of the proposed MDPDE as compared to its competitors for data examples where there is no outlier or only some mild outliers.

\subsection{Tyrosinase enzyme data}
In this example, we illustrate the performance of the proposed MDPDE-based procedure as compared to its competitors for the cases where there is no outlier or only some mild outliers.  We consider data from Marasovic et al.~\cite{marasovic2017robust} which contains observations on reaction rates and concentrations of the tyrosinase ubiquitous enzyme extracted from two species of mushroom, viz \textit{Agaricus bisporus} (Ab) and \textit{Pleurotus ostreatus} (Po). The reaction rates of Po are considered in $10^{-4}$. The first set of observations on Ab contains no outlier, while there is a mild outlier in the second set on Po. We fitted the MM model (Eq.~(28) of the main paper) to these two datasets using the proposed MDPDE and different competitors. Table \ref{tab:Tyrosianase-enzyme-data} contains the parameter estimates, their standard errors, and $20\%$ TAPEs; the model fits are shown in Figure \ref{fig:ubiquitous} along with the scatter plots of the data.

\begin{table}[!h]
    \caption{Parameter estimates of the MM model fitted to the Tyrosinase enzyme data, along with their asymptotic standard errors (in parentheses) and $20\%$ TAPE. [For MDPDE, $\widehat{\alpha} =0.35$ and $0.53$ are the optimum choice for Ab and Po data, respectively, obtained by IWJ method]}
    \centering
    \resizebox{0.9\textwidth}{!}
    {\begin{tabular}{|c|rrrrrrrr|}
        \multicolumn{9}{c}{\textbf{\textit{Agaricus bisporus} (Ab)}} \\ \hline
        \multicolumn{9}{|c|}{MDPDE}\\
        $\alpha$ & 0 & 0.05 & 0.1 & 0.2 & 0.35 & 0.5 & 0.7 & 1 \\ \hline
        $\beta_1$ & 0.377 & 0.377 & 0.376 & 0.376 & 0.378 & 0.378 & 0.378 & 0.376 \\
        ~ & (0.006) & (0.006) & (0.006) & (0.006) & (0.003) & (0.002) & (0.001) & (0.003) \\
        $\beta_2$ & 545 & 545 & 545 & 545 & 577 & 588 & 588 & 578 \\
        ~ & (36.59) & (36.51) & (36.32) & (35.53) & (22.45) & (9.87) & (8.89) & (17.38) \\
        $\sigma$ & 0.008 & 0.008 & 0.008 & 0.008 & 0.005 & 0.002 & 0.002 & 0.003 \\
        ~ & (0.000) & (0.000) & (0.000) & (0.000) & (0.000) & (0.000) & (0.000) & (0.000) \\
        TAPE  & 2.88E-05 & 2.82E-05 & 2.72E-05 & 2.58E-05 & 1.26E-05 & 1.07E-05 & 1.07E-05 & 1.28E-05 \\ \hline \hline
        ~ & \multicolumn{3}{c}{KPS} & M-est & M-est & ~ & ~ & ~ \\
        ~ & w=0.5 & w=1 & w=1.5 & Tukey & Huber & ~ & ~ & ~ \\ \hline
        $\beta_1$ & 0.367 & 0.380 & 0.380 & 0.380 & 0.380 & ~ & ~ & ~ \\
        ~ & (0.008) & (0.001) & (0.001) & (0.000) & (0.000) & ~ & ~ & ~ \\
        $\beta_2$ & 519 & 599 & 599 & 599 & 598 & ~ & ~ & ~ \\
        ~ & (51.29) & (3.87) & (4.81) & (2.14) & (3.26) & ~ & ~ & ~ \\
        TAPE & 3.49E-05 & 8.14E-06 & 8.16E-06 & 8.03E-06 & 8.07E-06 & ~ & ~ & ~ \\ \hline
        \multicolumn{9}{c}{}\\
        \multicolumn{9}{c}{\textbf{\textit{Pleurotus ostreatus} (Po)}} \\ \hline
        \multicolumn{9}{|c|}{MDPDE} \\
        $\alpha$ & 0 & 0.05 & 0.1 & 0.2 & 0.3 & 0.53 & 0.7 & 1 \\ \hline
        $\beta_1$ & 94.61 & 94.44 & 94.18 & 93.47 & 92.04 & 89.13 & 89.27 & 89.29 \\
        ~ & (8.683) & (8.596) & (8.480) & (8.018) & (6.642) & (0.147) & (0.495) & (0.652) \\
        $\beta_2$ & 5693 & 5699 & 5699 & 5699 & 5699 & 5788 & 5809 & 5811 \\
        ~ & (1095) & (1086) & (1075) & (1024) & (861) & (20) & (67) & (88) \\
        $\sigma$ & 2.277 & 2.250 & 2.210 & 2.059 & 1.672 & 0.034 & 0.111 & 0.135 \\
        ~ & (2.319) & (2.273) & (2.212) & (1.976) & (1.352) & (0.001) & (0.007) & (0.011) \\
        TAPE & 1.366 & 1.340 & 1.318 & 1.303 & 1.439 & 0.731 & 0.732 & 0.732 \\ \hline \hline
        ~ & \multicolumn{3}{c}{KPS} & M-est & M-est & ~ & ~ & ~ \\
        ~ & w=0.5 & w=1 & w=1.5 & Tukey & Huber & ~ & ~ & ~ \\ \hline
        $\beta_1$ & 95.65 & 90.01 & 90.03 & 102.41 & 99.55 & ~ & ~ & ~ \\
        ~ & (8.015) & (3.912) & (4.708) & (3.527) & (5.892) & ~ & ~ & ~ \\
        $\beta_2$ & 6121 & 5876 & 5898 & 6938 & 6538 & ~ & ~ & ~ \\
        ~ & (1044) & (529) & (637) & (464) & (768) & ~ & ~ & ~ \\
        TAPE & 1.059 & 0.738 & 0.739 & 0.785 & 0.823 & ~ & ~ & ~ \\ \hline
    \end{tabular}}
    \label{tab:Tyrosianase-enzyme-data}
\end{table}

\begin{figure}[H]
     \centering
     \begin{subfigure}[b]{0.48\textwidth}
         \centering
         \includegraphics[width=\textwidth]{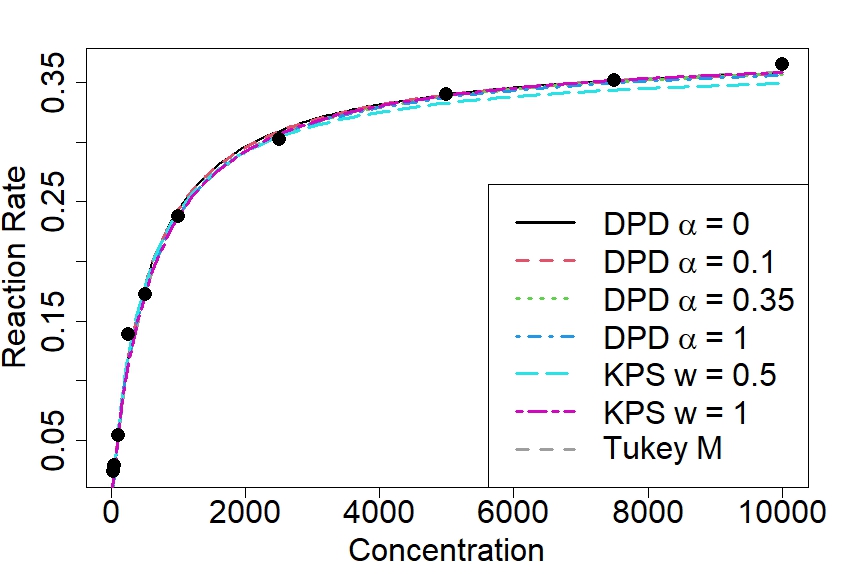}
    \caption{\textit{Agaricus bisporus} (Ab)}
    \label{fig:Ab-plot}
     \end{subfigure}
     \hfill
     \begin{subfigure}[b]{0.48\textwidth}
         \centering
         \includegraphics[width=\textwidth]{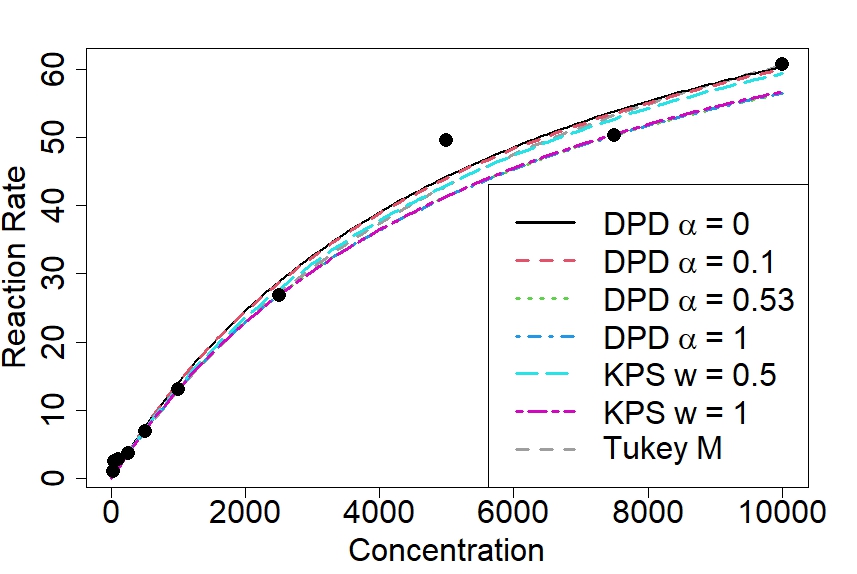}
    \caption{\textit{Pleurotus ostreatus} (Po)}
    \label{fig:Po-plot}
     \end{subfigure}
     \caption{Tyrosinase enzyme data and the plot of the fitted MM models using different methods}
     \label{fig:ubiquitous}
\end{figure}

\bigskip
The IWJ method was used for the selection of the optimum tuning parameter value for these datasets. For \textit{Agaricus bisporus} (Ab), it took 5 iterations to converge to the optimum value $\widehat{\alpha}=0.35$, while for \textit{Pleurotus ostreatus} (Po) data the method required only one iteration to select the same optimum value $\widehat{\alpha}=0.53$ for all pilot choices of $\alpha$.

The scatter plot of Ab data is provided in the Figure \ref{fig:Ab-plot} along with the MDPDE fits (for various $\alpha\geq 0$) and the robust competitors. The figure shows that
the proposed MDPDE (for all $\alpha\geq 0$), as well as all robust competitors, provide good fit. However, Figure \ref{fig:Po-plot}, corresponding to the Po dataset, shows that the MDPDEs with $\alpha\geq 0.53$ and KPS with $w=1$ provide good robust fits with significantly lower TAPE than other robust competitors.

We further tested the validity of the MM model for these data using the Wald-type tests, which resulted a p-value of $<0.001$ for all values of $\alpha \geq 0$ (including the classical Wald test). It shows that the assumed MM model is valid, and all tests give the same inference in the presence of no or mild contamination.

\subsection{\btxt{A putative enzyme data}}
\btxt{This dataset is taken from Table 3.1 of Maragoni \cite{marangoni2003enzyme}, containing the reaction velocities of an enzyme kinetic reaction at 15 different concentration levels of a putative enzyme, recorded over 5 different replications. So, there are 75 data points in total, making it a large sample in the context of enzyme kinetics. We have fitted the MM model to this dataset using the proposed MDPDE, as well as its existing competitors. The parameter estimates are given in Table \ref{tab:putative-data} along with their asymptotic standard errors and $20\%$ TAPEs for the fitted models, while the model fits are illustrated in Figure \ref{fig:putative-plot}. The optimum value of $\alpha$ in the MDPDEs, for this dataset, was observed to converge to $0.49$ (starting with the MDPDE with $\alpha > 0.4$), using the IWJ method. One can clearly observe that the proposed MDPDE at optimum $\alpha = 0.49$ (in fact, for all $\alpha\geq 0.49$) yields robust estimates and improved model fits, ignoring the effects of the outliers.}

\btxt{We also fitted separate MM models to the individual replications using all other robust methods under discussion. However, no significant variations in the results were observed across replications, and hence they are omitted here for brevity; only the results from the combined sample of size 75 are reported (as an illustration of the MDPDEs under larger sample sizes).}


\begin{table}[H]
    \color{blue}
    \centering
    \caption{\btxt{Parameter estimates of the MM model fitted to the putative enzyme data, along with their asymptotic standard errors (in parentheses) and $20\%$ TAPE. [For MDPDE, $\widehat{\alpha}$ near $0.49$ is the optimum choice, obtained by IWJ method]}}
    \begin{tabular}{|c|rrrrrrrr|}
    \hline
        \multicolumn{9}{|c|}{MDPDE}\\
        $\alpha$ & 0 & 0.05 & 0.1 & 0.2 & 0.3 & 0.49 & 0.7 & 1 \\ \hline
        $\beta_1$ & 81.098 & 80.569 & 80.053 & 79.145 & 78.575 & 78.126 & 77.858 & 77.313 \\
        & (1.985) & (2.622) & (2.552) & (2.419) & (2.331) & (2.298) & (2.342) & (2.428) \\
        $\beta_2$ & 38.617 & 38.350 & 38.098 & 37.693 & 37.525 & 37.529 & 37.506 & 37.092 \\
        & (2.413) & (3.193) & (3.114) & (2.964) & (2.868) & (2.845) & (2.907) & (3.013) \\
        $\sigma$ & 3.831 & 5.074 & 4.938 & 4.643 & 4.399 & 4.147 & 4.007 & 3.876 \\
        & (2.397) & (4.219) & (4.032) & (3.669) & (3.417) & (3.273) & (3.291) & (3.335) \\
        TAPE & 8.836 & 8.659 & 8.411 & 8.039 & 7.900 & 7.831 & 7.787 & 7.765 \\ \hline \hline
        ~ & \multicolumn{3}{c}{KPS} & M-est & M-est & ~ & ~ & ~ \\ 
        ~ & w = 0.5 & w = 1 & w = 1.5 & Huber & Tukey & ~ & ~ & ~ \\ \hline
        $\beta_1$ & 79.538 & 79.229 & 79.451 & 79.416 & 78.446 & ~ & ~ & ~ \\
        & (2.118) & (2.504) & (3.100) & (2.771) & (2.186) & ~ & ~ & ~ \\
        $\beta_2$ & 38.270 & 38.488 & 38.736 & 38.040 & 37.320 & ~ & ~ & ~ \\
        & (2.609) & (3.108) & (3.854) & (3.405) & (2.685) & ~ & ~ & ~ \\
        TAPE & 8.046 & 7.922 & 7.948 & 8.055 & 7.906 & ~ & ~ & ~ \\ \hline
    \end{tabular}
    \label{tab:putative-data}
\end{table}

\begin{figure}[!h]
    \centering
    \includegraphics[width=0.55\linewidth]{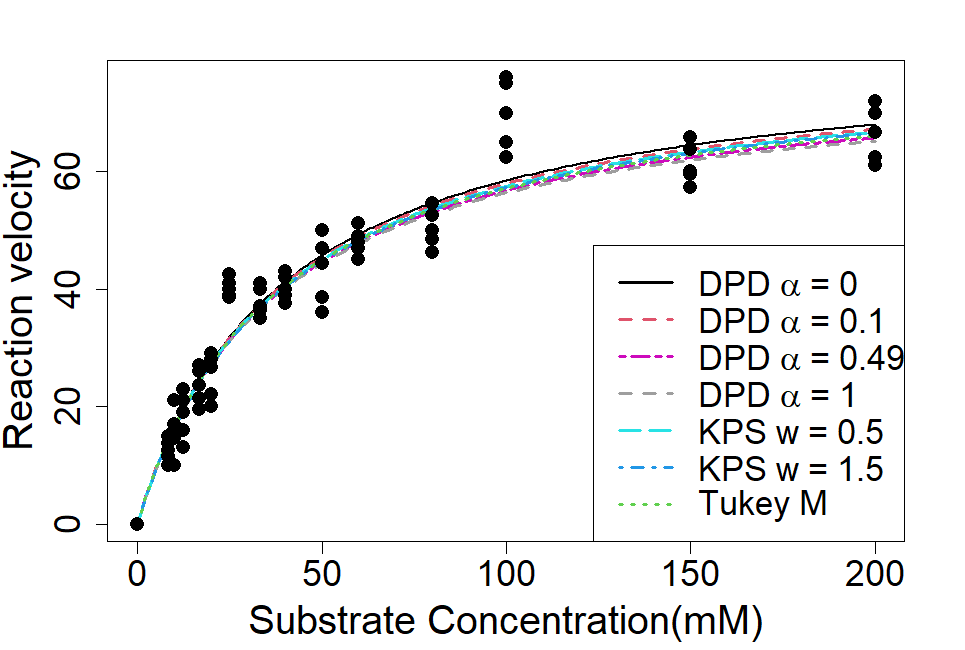}
    \caption{\btxt{The putative enzyme data and the plot of the fitted MM models using different methods}}
    \label{fig:putative-plot}
\end{figure}

\bibliographystyle{acm}
\bibliography{Reference.bib}